\documentclass[journal]{IEEEtran}
\usepackage[cmex10]{amsmath}
\usepackage{amssymb}
\usepackage{amsthm}
\usepackage{bm}
\usepackage{bbm}
\usepackage{subfigure}

\usepackage{hyperref}
\usepackage{mathrsfs}
\hypersetup{colorlinks=true, linkcolor=blue, citecolor=blue, urlcolor = blue}
\usepackage[linesnumbered,ruled,vlined]{algorithm2e}
\usepackage{cleveref}
\usepackage{graphicx}
\usepackage{psfrag}
\usepackage{epsfig}
\usepackage{epstopdf}

\newtheorem{theorem}{Theorem}

\newtheorem{lemma}{Lemma}

\newtheorem{remark}{Remark}
\newtheorem{definition}{Definition}
\usepackage{comment}
\newcommand{\bunderline}[1]{\underline{#1\mkern-4mu}\mkern4mu }
\DeclareMathOperator*{\argmin}{arg\,min}
\DeclareMathOperator*{\argmax}{arg\,max}

\usepackage{algorithm2e}
\usepackage[table]{xcolor}
\ifCLASSINFOpdf
\else
\fi

\begin{document}
%
\title{Reliability function for streaming over a DMC with feedback}
%
%
%

\author{Nian~Guo,~\IEEEmembership{Member,~IEEE,}
        Victoria~Kostina,~\IEEEmembership{Senior Member,~IEEE}
\thanks{N. Guo and V. Kostina are with the Department
of Electrical Engineering, California Institute of Technology, Pasadena, CA, 91125 USA. E-mail: \{nguo,vkostina\}@caltech.edu. This work was supported in part by the National Science Foundation (NSF) under grants CCF-1751356 and CCF-1956386.}}

%
%

\markboth{}%
{Shell \MakeLowercase{\textit{et al.}}: Bare Demo of IEEEtran.cls for IEEE Journals}
%



\maketitle


\begin{abstract}
Conventionally, posterior matching is investigated in channel coding and block encoding contexts -- the source symbols are equiprobably distributed and are entirely known by the encoder before the transmission. In this paper, we consider a streaming source, whose symbols progressively arrive at the encoder at a sequence of deterministic times. We derive the joint source-channel coding (JSCC) reliability function for streaming over a discrete memoryless channel (DMC) with feedback. We propose a novel \emph{instantaneous encoding phase} that operates during the symbol arriving period and achieves the JSCC reliability function for streaming when followed by a block encoding scheme that achieves the JSCC reliability function for a classical source whose symbols are fully accessible before the transmission. During the instantaneous encoding phase, the evolving message alphabet is partitioned into groups whose priors are close to the capacity-achieving distribution, and the encoder determines the group index of the actual sequence of symbols arrived so far and applies randomization to exactly match the distribution of the transmitted index to the capacity-achieving one. Surprisingly, the JSCC reliability function for streaming is equal to that for a fully accessible source, implying that the knowledge of the entire symbol sequence before the transmission offers no advantage in terms of the reliability function. For streaming over a symmetric binary-input DMC, we propose a one-phase \emph{instantaneous small-enough difference (SED) code} that not only achieves the JSCC reliability function, but also, thanks to its single-phase time-invariant coding rule, can be used to stabilize an unstable linear system over a noisy channel.
For equiprobably distributed source symbols, we design low complexity algorithms to implement both the instantaneous encoding phase and the instantaneous SED code. The algorithms group the source sequences into sets we call types, which enable the encoder and the decoder to track the priors and the posteriors of source sequences jointly, leading to a log-linear complexity in time. While the reliability function is derived for non-degenerate DMCs, i.e., DMCs whose transition probability matrix has all positive entries, for degenerate DMCs, we design a code with instantaneous encoding that achieves zero error for all rates below Shannon's joint source-channel coding limit.
\end{abstract}


\section{Introduction}\label{Sec_intro}
This paper considers joint source-channel coding of streaming data over a DMC with full feedback using variable-length feedback codes. With the emergence of the Internet of Things, communication systems, such as those employed in distributed control and tracking scenarios, are becoming increasingly dynamic, interactive, and delay-sensitive. The source symbols in such real-time systems arrive at the encoder in a streaming fashion. For example, the height and the speed data of an unmanned aerial vehicle stream into the encoder in real time. An intriguing question is: What codes can transmit streaming data with both high reliability and low latency over a channel with feedback? Classical posterior matching schemes \cite{Horstein}--\cite{Yang} can reliably transmit messages over a channel with feedback but under the assumption that the source sequence is fully accessible to the encoder before the transmission. One can simply buffer the arriving data into a block and then transmit the data block using a classical posterior matching scheme. Intuitively, the buffer-then-transmit code is a good choice if the buffering time is negligibly short, i.e., if data packets arrive at the encoder at an extremely fast rate. However, if data packets arrive at the encoder steadily rather than in a burst, the buffer-then-transmit code becomes ill-suited due to the delay introduced by collecting data into a block before the transmission \cite{Nian}. The encoder in this paper performs \emph{instantaneous} encoding: it starts transmitting as soon as the first message symbol arrives and incorporates new message symbols into the continuing transmission on the fly. Like classical posterior matching schemes, it relies on full channel feedback.

Designing good channel block encoding schemes with feedback is a classical problem in information theory \cite{Horstein}--\cite{Yang}, since feedback, though unable to increase the capacity of a memoryless channel \cite{Shannon}, can simplify the design of capacity-achieving codes \cite{Horstein}--\cite{Shayevitz} and improve achievable delay-reliability tradeoffs \cite{Burnashev}\cite{PPV}. The underlying principle behind capacity-achieving block encoding schemes with feedback \cite{Horstein}--\cite{Yang}, termed posterior matching \cite{Shayevitz}, is to transmit a channel input that has two features. First, the channel input is independent of the past channel outputs, representing the new information in the message that the decoder has not yet observed. Second, the probability distribution of the channel input is matched to the capacity-achieving one using the posterior of the message. 

While asymptotically achieving the channel capacity ensures the best possible transmission rates in the limit of large delay, optimizing the tradeoff between delay and reliability is critical for time-sensitive applications. 
The delay-reliability tradeoff is often measured by the reliability function (a.k.a. optimal error exponent), which is defined as the maximum rate of the exponential decay of the error probability at a rate strictly below the channel capacity as the blocklength is taken to infinity. It is a classical fundamental limit that helps to gain insight into the finite blocklength performance of codes via large deviations theorems in probability.  In the context of channel coding, the reliability function of a DMC with feedback was first shown by Burnashev \cite{Burnashev}.  Variable-length channel codes with block encoding that achieve Burnashev's reliability function are proposed in \cite{Burnashev}--\cite{Naghshvar2}, \cite{Yang}. 
Burnashev's \cite{Burnashev} and Yamamoto and Itoh (Y-I)'s schemes \cite{Yamamoto} are structurally similar in that they both have two phases. In the communication phase, the encoder matches the distribution of its output to the capacity-achieving input distribution, while aiming to increase the decoder's belief about the true message. In the confirmation phase, the encoder repeatedly transmits one of two symbols indicating whether or not the decoder's estimate at the end of the communication phase is correct. 
Caire et al. \cite{Caire} showed that the code transmitted in the communication phase of the Y-I scheme can be replaced by any non-feedback block channel code, provided that the error probability of the block code is less than a constant determined by the code rate as the blocklength goes to infinity. Naghshvar et al. \cite{Naghshvar2} challenged the convention of using a two-phase code \cite{Burnashev}--\cite{Caire} to achieve Burnashev's reliability function by proposing the MaxEJS code, which searches for the deterministic encoding function that maximizes an extrinsic Jensen-Shannon (EJS) divergence at each time. Since the MaxEJS code has a double exponential complexity in the length of the message sequence $k$, for symmetric binary-input DMCs, Naghshvar et al.~\cite{Naghshvar2} proposed a simplified encoding function that is referred to as the \emph{small-enough difference} (SED) rule in \cite{Antonini}. The SED encoder partitions the message alphabet into two groups such that the difference between group posteriors and the Bernoulli$\left(\frac{1}{2}\right)$ capacity-achieving distribution is small.  While the SED rule still has an exponential complexity in the length of the message, Antonini et al. \cite{Antonini} designed a systematic variable-length code for transmitting $k$ bits over a binary symmetric channel (BSC) with feedback that has complexity $O(k^2)$. The complexity reduction is realized by grouping messages with the same posterior. Yang et al. \cite{Yang} generalized Naghshvar et al.'s SED rule-based code \cite{Naghshvar2} to binary-input binary-output asymmetric channels.

While the message in \cite{Burnashev}--\cite{Yang} is equiprobably distributed on its alphabet, the JSCC reliability function for transmitting a non-equiprobable discrete-memoryless source (DMS) over a DMC has also been studied \cite{Gallager}--\cite{Truong}. For fixed-length almost lossless coding without feedback, Gallager \cite{Gallager} derived an achievability bound on the JSCC reliability function, which indicates that JSCC leads to a strictly larger error exponent than separate source and channel coding in some cases; Csisz\`{a}r \cite{Csiszar} provided achievability and converse bounds on the JSCC reliability function using random coding and type counting; Zhong et al. \cite{Zhong2} showed that Csisz\`{a}r's achievability bound \cite{Csiszar} is tighter than Gallager's bound \cite{Gallager} and provided sufficient conditions for the JSCC reliability function to be strictly larger than the separate source and channel coding reliability function. For variable-length lossy coding with feedback, Truong and Tan \cite{Truong} derived the JSCC excess-distortion reliability function under the assumption that $1$ source symbol is transmitted per channel use on average.  To achieve the excess-distortion reliability function, Truong and Tan \cite{Truong} used separate source and channel codes: the source is compressed down to its rate-distortion function, and the compressed symbols are transmitted using the Y-I communication phase, while the Y-I confirmation phase is modified to compare the uncompressed source and its lossy estimate instead of the compressed symbol and the estimate thereof. Due to the modification, some channel coding errors bear no effect on the overall decoding error, and the overall decoding error is dominated by the decoding error of the repetition code in the confirmation phase.

While most feedback coding schemes in the literature considered block encoding of a source whose outputs are accessible in their entirety before the transmission \cite{Horstein}--\cite{Yang}, \cite{Truong}, several existing works considered instantaneous encoding of a streaming source \cite{Sahai}--\cite{Antonini2}, \cite{Nian}. A large portion of them \cite{Sahai}--\cite{Lalitha} explores instantaneous (causal) encoding schemes for stabilizing a control system. The evolving system state is considered as a streaming data source, the observer instantaneously transmits information about the state to the controller, and the controller injects control signals into the plant.
Sahai and Mitter \cite{Sahai} defined the \emph{anytime} capacity at anytime reliability $\alpha$ as the maximum transmission rate $R$ (nats per channel use) such that the decoding error of the first $k$ $R$-nat symbols at time $t$ decays as $e^{-\alpha(t-k)}$ for any $k\leq t$; they showed that the scalar linear system can be stabilized provided that the logarithm of its unstable coefficient is less than the anytime capacity; they suggested that codes that lead to an exponentially decaying error have a natural tree structure (similar to Schulman's code \cite{Schulman} for interactive computing) that tracks the state evolution over time. Tree coding schemes for stabilizing control systems have been studied in \cite{RT}--\cite{Khina}. Assuming that the inter-arrival times of message bits are known by the decoder and that the channel is a BSC, Lalitha et al.~\cite{Lalitha} proposed an anytime code \cite{Sahai} that achieves a positive anytime reliability and derived a lower bound on the maximum rate that leads to an exponentially vanishing error probability. Instantaneous encoding schemes have also been studied in pure communication settings, where one may evaluate the error exponent \cite{Chang}\cite{Nian}, consider a streaming source with finite length \cite{Antonini2}\cite{Nian}, and allow non-periodic deterministic \cite{Lalitha} or random \cite{Nian} streaming times. Chang and Sahai \cite{Chang} considered instantaneous encoding of i.i.d. message symbols that arrive at the encoder at consecutive times for the transmission over a binary erasure channel (BEC) with feedback, and showed the zero-rate JSCC error exponent of erroneously decoding the $k$-th message symbol at time $t$ for fixed $k$ and $t\rightarrow\infty$. Antonini et al. \cite{Antonini2} designed a causal encoding scheme for $k<\infty$ streaming bits with a fixed arrival rate over a BSC and showed by simulation that the code rate approaches the channel capacity as the bit arrival rate approaches the transmission rate. In our previous work~\cite{Nian}, we proposed a code that uses an adapted SED rule \cite{Naghshvar2} to instantaneously transmit $k<\infty$ randomly arriving bits and that leads to an achievability bound on the reliability function for binary-input DMCs with instantaneous encoding, and we designed a polynomial-time version of it. While the instantaneous encoding schemes in \cite{Sahai}--\cite{Antonini2}, \cite{Nian} employ feedback, transmission schemes for streaming data \emph{without} feedback have been investigated for finite memory encoders \cite{Draper}, for distributed sources \cite{Draper2}, and for point-to-point channels in the moderate deviations \cite{Lee2} and the central limit theorem \cite{Lee1} regimes. 

In this paper, we propose a novel coding phase -- the instantaneous encoding phase -- for transmitting a sequence of $k$ source symbols over a DMC with feedback. It performs instantaneous encoding during the arriving period of the symbols. At time $t$, the encoder and the decoder calculate the priors of all possible symbol sequences using the source distribution and the posteriors at time $t-1$. Then, they partition the evolving message alphabet into groups, so that the group priors are close to the capacity-achieving distribution. In contrast to Naghshvar et al.'s SED rule \cite{Naghshvar2} for symmetric binary-input channels, our partitioning rule applies to any DMCs, and it uses group priors instead of group posteriors for the partitioning. Using group priors is necessary because if a new symbol arrives at time $t$, the posteriors at time $t-1$ are insufficient to describe the symbol sequences at time $t$. Feedback codes with block encoding \cite{Horstein}--\cite{Yang}, \cite{Truong} only need to consider the posteriors, since block encoding implies that the priors at time $t$ are equal to the posteriors at time $t-1$. Once the groups are partitioned, the encoder determines the index of the group that contains the true symbol sequence it received so far and applies randomization to match the distribution of the transmitted index to the capacity-achieving one.

We derive the JSCC reliability function for the almost lossless transmission of a discrete streaming source over a DMC with feedback. Since allowing the encoder to know the entire source sequence before the transmission will not decrease the reliability function, converse bounds for a classical fully accessible source pertain. We extend Berlin et al.'s converse bound \cite{Berlin} for Burnashev's reliability function to JSCC. For fully accessible sources, we show that the converse is achievable by a variable-length joint source-channel code with block encoding
-- the MaxEJS code \cite{Naghshvar2}. For a source whose symbols arrive at the encoder with an infinite arriving rate (symbols per channel use) as the source length goes to infinity, we show that the converse is achievable by the buffer-then-transmit code that buffers the arriving symbols during the symbol arriving period and implements a block encoding scheme that achieves the JSCC reliability function for a fully accessible source after the arriving period. For example, a classical fully accessible source 
has an infinite symbol arriving rate because its symbols arrive all at once. Yet, this buffer-then-transmit code fails to achieve the JSCC reliability function for streaming if the source symbols arrive at the encoder with a finite arriving rate of symbols per channel use. For streaming symbols with an arriving rate greater than $\frac{1}{\bunderline{H}}\left(H(P_{Y}^*) - \log\frac{1}{p_{\max}}\right)$, we show that preceding any code with block encoding that achieves the JSCC reliability function for a fully accessible source by our instantaneous encoding phase will make it achieve the block encoding error exponent as if the encoder knew the entire source sequence before the transmission. Here $\bunderline{H}$ is a lower bound on the information in the streaming source and is equal to the source entropy rate if the source is information stable, $H(P_Y^*)$ is the entropy of the channel output distribution induced by the capacity-achieving channel input distribution, and $p_{\max}$ is the maximum channel transition probability. Thus, surprisingly, the JSCC reliability function for streaming is equal to that for a fully accessible source. Furthermore, we show via simulations that the reliability function gives a surprisingly good approximation to the delay-reliability tradeoffs attained by the JSCC reliability function-achieving codes in the ultra-short blocklength regime. 

The above discussion highlights the existence of a sequence of codes with instantaneous encoding indexed by the length of the source sequence $k$ that achieves the JSCC reliability function as $k\rightarrow\infty$. However, in the remote tracking and control scenarios, a single code that can choose to decode any $k$ symbols of a streaming source at any time $t$ with an error probability that decays exponentially with the decoding delay (i.e., an anytime code \cite{Sahai}) is desired. To this end, we design the \emph{instantaneous small-enough difference (SED) code}. The instantaneous SED code is similar to the instantaneous encoding phase except that it continues the transmissions after the symbol arriving period, drops the randomization step, and specifies the group partitioning rule to the instantaneous SED rule. The instantaneous SED code is also similar to the instantaneous encoding scheme in our previous work \cite{Nian} designed for transmitting a streaming source with random symbol arriving times unknown to the decoder, except that \cite{Nian} used an instantaneous \emph{smallest-difference} rule. The instantaneous smallest-difference rule minimizes the difference between the group priors and the capacity-achieving probabilities, whereas the instantaneous SED rule only drives their difference small enough. The instantaneous SED rule reduces to Naghshvar et al.'s \cite{Naghshvar2} SED rule if the source is fully accessible before the transmission. In contrast to the instantaneous encoding phase followed by a block encoding scheme, the instantaneous SED code only has one phase, namely, it follows the same transmission strategy at each time. For transmitting i.i.d. Bernoulli$\left(\frac{1}{2}\right)$ bits that arrive at the encoder at consecutive times over a BSC($0.05$), simulations of the instantaneous SED code show that the error probability of decoding the first $k=[4\colon4\colon16]$ bits at times $t\in [4,64]$, $t\geq k$, decreases exponentially with an anytime reliability $\alpha\simeq 0.172$, outperforming the theoretical anytime reliability of Lalitha et al's anytime code \cite{Lalitha}. This implies that the binary instantaneous SED code can be used to stabilize an unstable linear system with bounded noise \cite{Sahai}--\cite{Lalitha}. Although the achievability of a positive anytime reliability is evidenced by the simulation, it is difficult to prove this analytically since one cannot leverage the submartingales and the bounds on the expected decoding time of a block encoding scheme in \cite{Burnashev}, \cite{Naghshvar2}. Nevertheless, we show that a sequence of instantaneous SED codes indexed by the length of the symbol sequence $k$ achieves the JSCC reliability function for streaming over a Gallager-symmetric \cite[p. 94]{Gallager} binary-input DMC. This result is based on our finding that, after dropping the randomization step, the instantaneous encoding phase continues to achieve the JSCC reliability function when followed by a reliability function-achieving block encoding scheme, but at a cost of increasing the lower bound on the symbol arriving rate to $\frac{1}{\log\frac{1}{p_{S,\max}}}\left(\log\frac{1}{p_{\min}}-\log\frac{1}{p_{\max}}\right)$. Here, $p_{S,\max}$ is the maximum symbol arriving probability and $p_{\min}$ is the minimum channel transition probability.

Since the size of the evolving source alphabet grows exponentially in time $t$, the complexities of the instantaneous encoding phase and the instantaneous SED code are exponential in time $t$. For the source symbols that are equiprobably distributed, we design low-complexity algorithms for both codes that we term \emph{type-based} codes. The complexity reduction is achieved by judiciously partitioning the evolving source alphabet into \emph{types}. The cardinality of the partition is $O(t)$, i.e., it is exponentially smaller than the size of the source alphabet. The type partitioning enables the encoder and the decoder to update the priors and the posteriors of the source sequences as well as to partition source sequences in terms of types rather than individual sequences. Since the prior and the posterior updates have a linear complexity in the number of types, and the type-based group partitioning rule has a log-linear complexity in the number of types due to type sorting, our type-based codes only have a log-linear complexity $O(t\log t)$. Although Antonini et al.'s block encoding scheme for BSCs \cite{Antonini} attains a reduction in complexity also by grouping message sequences, the types in Antonini et al.'s scheme \cite{Antonini} are generated all at once by grouping the message sequences that have the same Hamming distance to the received channel outputs, while the types in our type-based codes evolve with the arrival of source symbols. 

For the transmission over a degenerate DMC, i.e., a DMC whose transition matrix contains a zero, we propose a code with instantaneous encoding that achieves zero error for all rates asymptotically below Shannon's JSCC limit. While feedback codes in most prior literature \cite{Yamamoto}--\cite{Yang}, \cite{Truong} are designed for non-degenerate DMCs, i.e., a DMC whose transition probability matrix has all positive entries, Burnashev \cite[Sec. 6]{Burnashev} constructed a channel code for degenerate DMCs that achieves zero error for all rates asymptotically below the channel capacity. Our code extends Burnashev's code \cite[Sec. 6]{Burnashev} to JSCC and to the streaming source. Similar to \cite{Burnashev}--\cite{Caire}, \cite{Truong}, our code is divided into blocks, and each block consists of a communication phase and a confirmation phase. Burnashev's \cite[Sec. 6]{Burnashev} communication phases use a block encoding scheme that can transmit reliably for all rates below the channel capacity. The communication phase in the first block of our scheme uses a code with instantaneous encoding that can transmit reliably for all rates below Shannon's JSCC limit; our $\ell$-th communication phase transmits the uncompressed source sequence to avoid compression errors, and uses random coding to establish an analyzable probability distribution of the decoding time. Our confirmation phase is the same as that of Burnashev's code \cite[Sec. 6]{Burnashev}: the encoder repeatedly transmits a pre-selected symbol that never leads to channel output $y$ if the decoder's estimate at the end of the communication phase is wrong, and transmits another symbol that can lead to $y$ if the estimate is correct. The confirmation phases rely on the degenerate nature of the channel to ensure zero error: receiving a $y$ secures an error-free estimate of the source.

The rest of the paper is organized as follows. In Section~\ref{Sec_II_A}, we formulate the problem and define the variable-length joint source-channel codes with instantaneous encoding. In Section~\ref{belief_phase}, we present the instantaneous encoding phase. In Section~\ref{Sec_reliability2}, we show the JSCC reliability function for streaming. In Section~\ref{Sec_instantaneous_SED}, we present the instantaneous SED code. In Section~\ref{Sec_practical}, we present the type-based codes with log-linear complexity. In Section~\ref{Sec_simulation}, we display the simulations of the instantaneous encoding phase, the instantaneous SED code, and their corresponding type-based codes. In Section~\ref{Sec_degenerate}, we present the zero-error code for degenerate DMCs. 

A part of this work is presented at the 2022 IEEE International Symposium on Information Theory \cite{Guo}. The conference version does not contain Sections~\ref{Sec_instantaneous_SED}--\ref{Sec_degenerate} or any proofs.

\textit{Notation:} $\log(\cdot)$ is the natural logarithm. Notation $X \leftarrow Y$ reads ``replace $X$ by $Y$''.  For any positive integer $q$, we denote $[q] \triangleq \{1,2,\dots,q\}$. We denote by $[q]^k$ the set of all $q$-ary sequences of length equal to $k$. For a possibly infinite sequence $x=\{x_1,x_2,\dots\}$, we write $x^n=\{x_1,x_2,\dots,x_n\}$ to denote the vector of its first $n$ elements, and we write $\{x_n\}_{n=n_1}^{n_2} = \{x_{n_1},x_{n_1+1},\dots,x_{n_2}\}$ to denote the vector formed by its $n_1, n_1+1,\dots, n_2$-th elements.  For a sequence of random variables $X_k$, $k=1,2,\dots$ and a real number $a\in\mathbb R$, we write $X_k\xrightarrow{\mathrm{i.p.}} a$ to denote that $X_k$ converges to $a$ in probability, i.e., $\lim_{k\rightarrow\infty}\mathbb P[|X_k-a|\geq \epsilon] = 0, ~\forall \epsilon >0$. For any set $\mathcal A$, we denote by $\mathbbm{1}_{\mathcal A}(x)$ an indicator function that is equal to $1$ if and only if $x\in\mathcal A$. For two positive functions $f,g\colon \mathbb Z_+\rightarrow \mathbb R_+$, we write $f(k)=o(g(k))$ to denote $\lim_{k\rightarrow \infty}\frac{f(k)}{g(k)}=0$; we write $f(k)=O(g(k))$ to denote $\limsup_{k\rightarrow\infty}\frac{f(k)}{g(k)}<\infty$; we write $f(k)=\Omega(g(k))$ to denote $\liminf_{k\rightarrow\infty}\frac{f(k)}{g(k)}>0$.

\section{Problem statement}\label{Sec_II_A}
Consider the setup in Fig.~\ref{problem2}. We formally define the discrete source that streams into the encoder as follows. 
\begin{figure}[h!]
\centering
\includegraphics[trim = 23mm 222mm 58mm 40mm, clip, width=9cm]{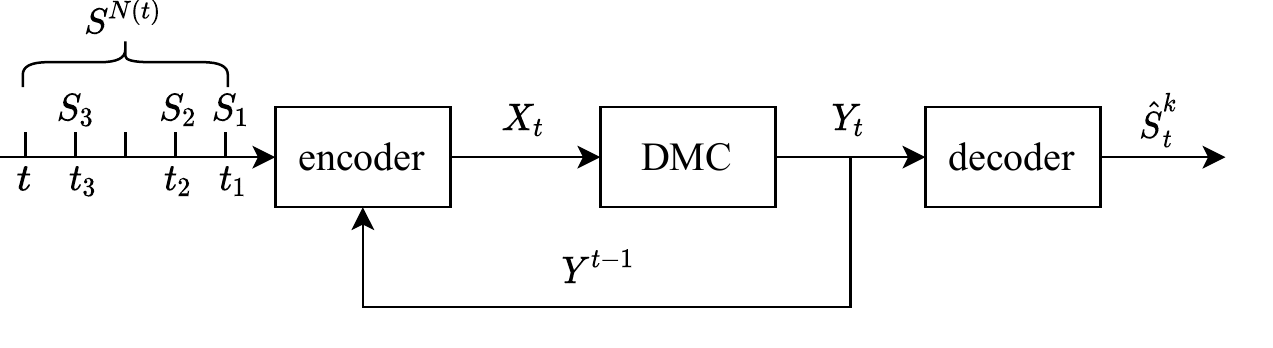}
\caption{Real-time feedback communication system with a streaming source.}
\label{problem2}
\end{figure}
\begin{definition}[A $(q,\{t_n\}_{n=1}^\infty)$ discrete streaming source (DSS)]\label{def_dss} We say that a source is a DSS if it emits a sequence of discrete source symbols $S_n\in[q]$, $n=1,2,\dots$ at times $t_1\leq t_2\leq \dots$, where symbol $S_n$ that arrives at the encoder at time $t_n$ is distributed according to the source distribution
\begin{align}\label{bit_prob}
    P_{S_n|S^{n-1}}, n =1,2,\dots
\end{align}
\end{definition}
Throughout, we assume that the entropy rate of the DSS 
    \begin{align}\label{entropy_rate}
        H \triangleq \lim_{n\rightarrow\infty}\frac{H(S^n)}{n}~\text{(nats per symbol)}
    \end{align}
     is well-defined and positive;
the first symbol $S_1$ arrives at the encoder at time $t_1 \triangleq 1$;
both the encoder and the decoder know the symbol alphabet $[q]$, the arrival times $t_1,t_2,\dots$, and the source distribution~\eqref{bit_prob}. The DSS reduces to the classical \emph{discrete source} (DS) that is fully accessible to the encoder before the transmission if
 \begin{align}\label{DS}
     t_n=1,~\forall n=1,2,\dots
 \end{align} 
Fig.~\ref{Fig_full_vs_stream} displays a fully accessible source and a streaming source.
\begin{figure}%
    \centering
    \subfigure[fully accessible: $1=t_1=t_2=\dots$ ]{{\includegraphics[trim = 30mm 240mm 130mm 20mm, clip, width=5cm]{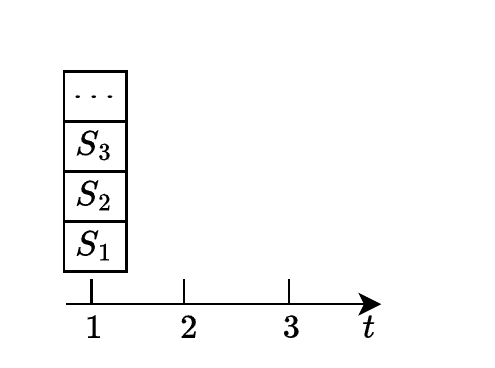}}}%
    \quad
    \subfigure[streaming: $t_1=1, t_2=2,t_3=4, t_4=t_5=6,\dots$]{{\includegraphics[trim = 35mm 242mm 105mm 30mm, clip, width=7.2cm]{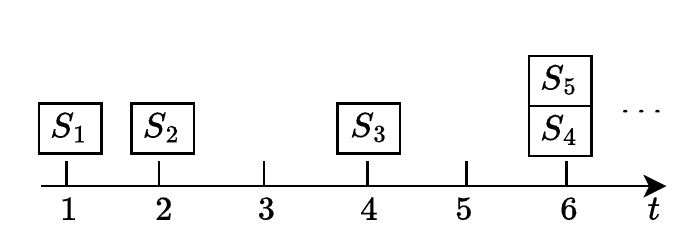}}}%
    \caption{A fully accessible source vs. a streaming source. A fully accessible source emits all symbols at $t=1$. A streaming source emits symbols progressively.}%
    \label{Fig_full_vs_stream}%
\end{figure}
Operationally, symbol $S_n$ represents a data packet. 
We denote the number of symbols that the encoder has received by time $t$ by
\begin{align}\label{Nt}
    N(t)\triangleq \max\{n\colon t_n\leq t, n =1,2,\dots\}.
\end{align}
Given a DSS (Definition~\ref{def_dss}) with symbol arriving times $t_1,t_2,\dots$, we denote its \emph{symbol arriving rate} by, assuming that the limit exists
\begin{align}\label{freq}
     f\triangleq \lim_{n\rightarrow\infty} \frac{n}{t_n}~\text{(symbols per unit time)}\in[0,\infty].
\end{align}
The symbol arriving rate $f=\infty$ implies that the source symbols arrive at the encoder so frequently that the number of channel uses increases slower than the source length. For example, the DS \eqref{DS} has $f=\infty$. The symbol arriving rate $f<\infty$ implies that the number of channel uses goes to infinity as the source length goes to infinity. For example, if one source symbol arrives at the encoder every $\lambda\geq 1$ channel uses, $\lambda\in \mathbb Z_+$, i.e., 
\begin{align}\label{periodic}
    t_n = \lambda(n-1)+1,
\end{align}
then 
\begin{align}\label{freq_lambda}
    f =\frac{1}{\lambda}.
\end{align}

We assume that the channel is a DMC with a single-letter transition probability distribution $P_{Y|X}\colon \mathcal X\rightarrow \mathcal Y$.
\begin{definition}[Non-degenerate and degenerate DMCs]\label{def_channel}
A DMC is non-degenerate if it satisfies
\begin{align}\label{non-degenerate}
    P_{Y|X}(y|x)>0, \forall x\in\mathcal X, y\in\mathcal Y.
\end{align}
A DMC is degenerate if there exist $y\in\mathcal Y$, $x\in\mathcal X$, $x'\in\mathcal X$, such that
\begin{subequations}\label{degenerate}
\begin{align}\label{degenerate_a}
    &P_{Y|X}(y|x)  >0,\\ \label{degenerate_b}
    &P_{Y|X}(y|x') =0.
\end{align}
\end{subequations}
\end{definition}
A non-degenerate DMC is considered in \cite{Burnashev}--\cite{Yang}, e.g., a BSC. A degenerate DMC is considered in \cite[Sec.6]{Burnashev}, e.g., a BEC. Fig.~\ref{Fig_DMC} display examples of DMCs. We denote the capacity of the DMC by
\begin{align}\label{capacity}
        C \triangleq \max_{P_{X}}I(X;Y),
    \end{align}
and we denote the maximum Kullback–Leibler (KL) divergence between its transition probabilities by
\begin{align}\label{C1}
        C_1 \triangleq \max_{x,x'\in\mathcal X}D(P_{Y|X=x}||P_{Y|X=x'}).
\end{align}
Assumption \eqref{non-degenerate} posits that $C_1$ \eqref{C1} is finite. 

A DMC is \emph{symmetric} (Gallager-symmetric \cite[p. 94]{Gallager}) if the columns in its channel transition probability matrix can be partitioned so that within each partition, all rows are permutations of each other, and all columns are permutations of each other. 

\begin{figure}%
    \centering
    \subfigure[non-degenerate ]{{\includegraphics[trim = 40mm 80mm 40mm 50mm, clip, width=3cm]{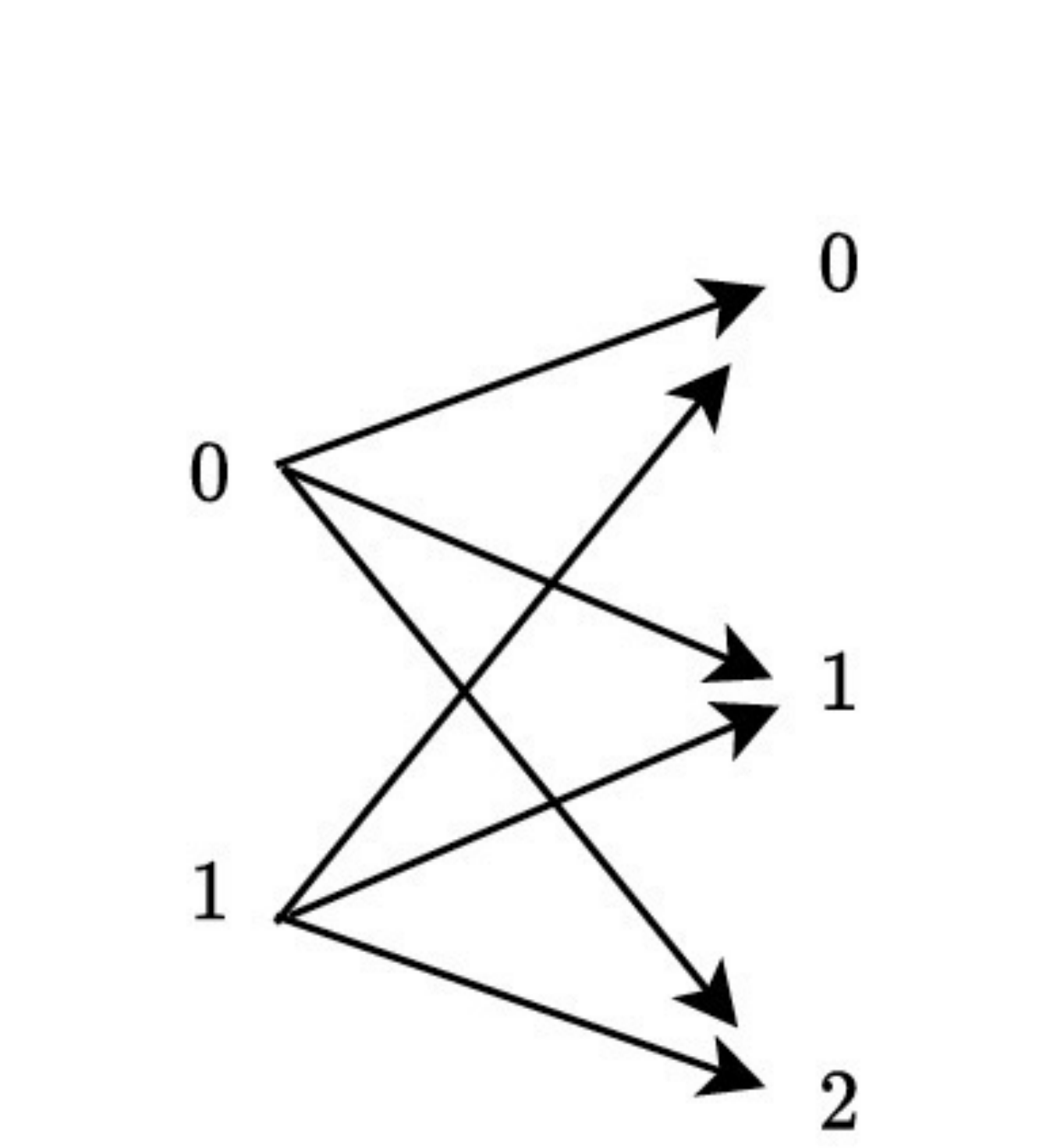}}}%
    \subfigure[degenerate]{{\includegraphics[trim = 40mm 80mm 40mm 50mm, clip, width=3cm]{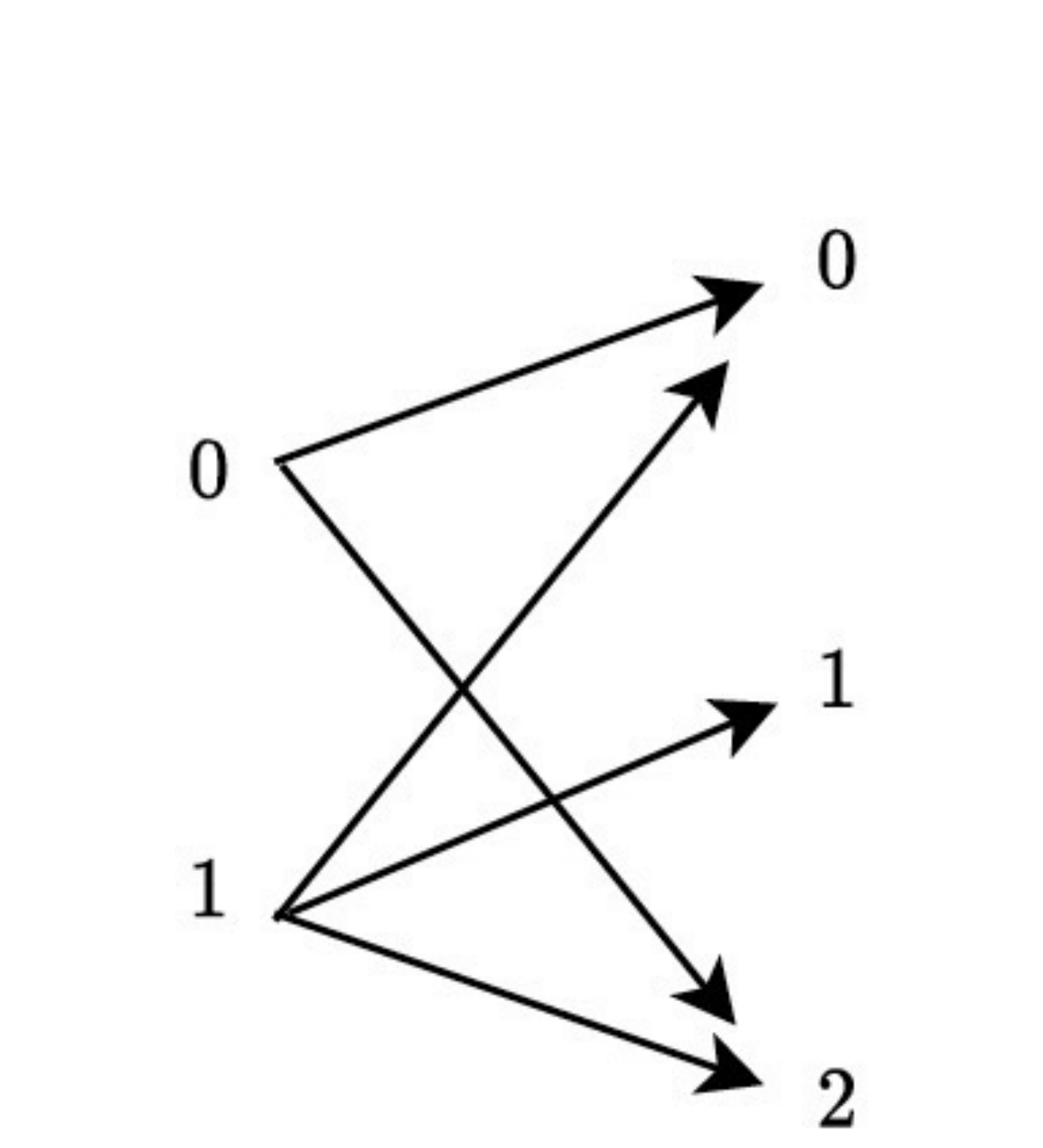}}}%
    \subfigure[neither non-degenerate nor degenerate]{{\includegraphics[trim = 40mm 80mm 40mm 50mm, clip, width=3cm]{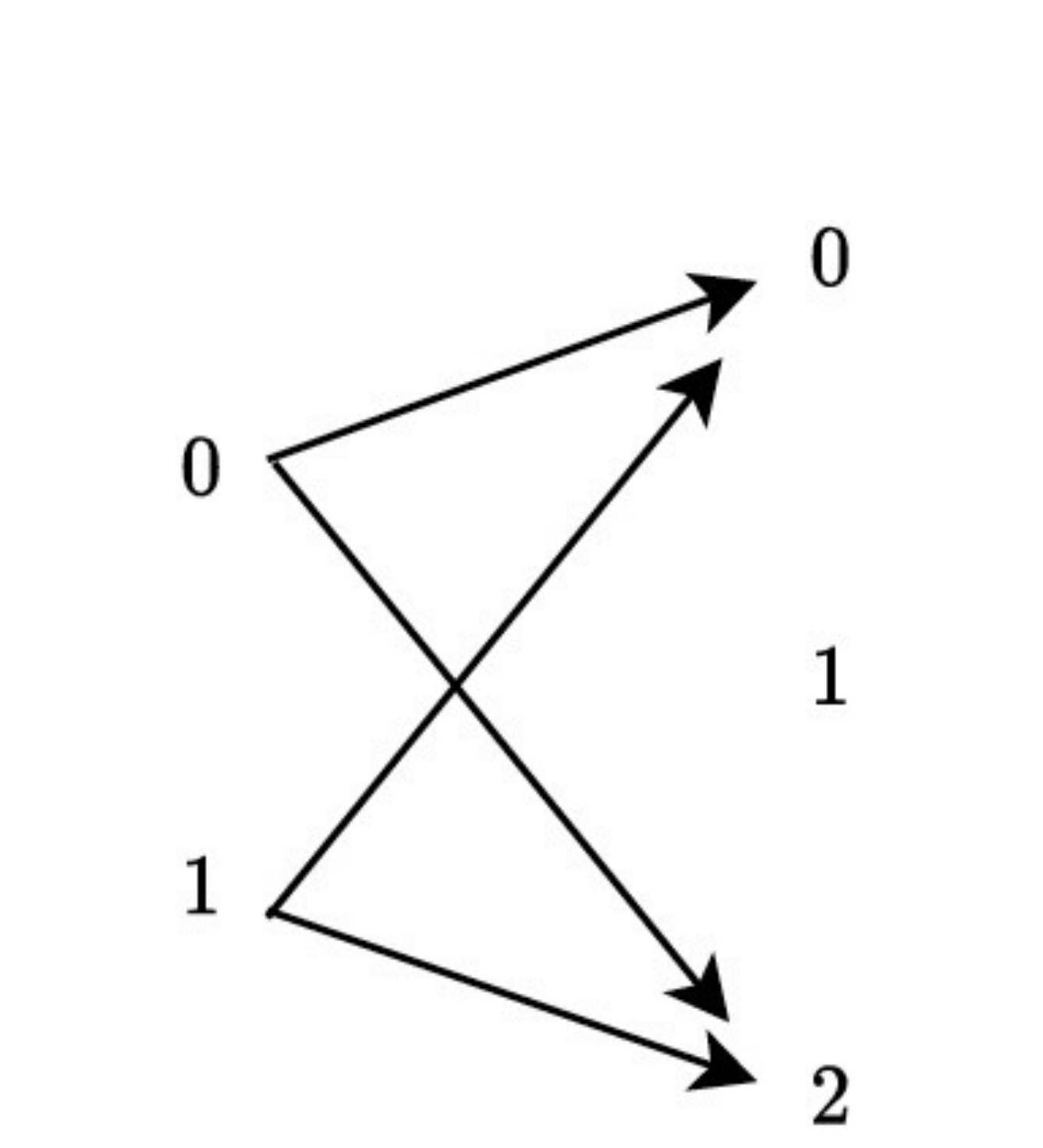}}}%
    \caption{A DMC $P_{Y|X}\colon \{0,1\}\rightarrow \{0,1,2\}$. An arrow from channel input $x\in\{0,1\}$ to channel output $y\in\{0,1,2\}$ signifies $P_{Y|X}(y|x)>0$. Channel (a) is a non-degenerate DMC that satisfies \eqref{non-degenerate}. Channel (b) is a degenerate DMC that satisfies \eqref{degenerate} with $y=1$, $x=1$, $x'=0$. Channel (c) does not satisfy \eqref{non-degenerate}--\eqref{degenerate} since $y=1$ is not reachable.}%
    \label{Fig_DMC}%
\end{figure}

Throughout, we measure the symbol arriving rate \eqref{freq} with a unit time equal to a channel use.

We proceed to define the codes that we use to transmit a DSS over a DMC with feedback. All the codes in this paper are \emph{variable-length joint source-channel codes with feedback}. We distinguish two classes of codes, one is called a code with instantaneous encoding, and the other is called a code with block encoding. Next, we define the code with instantaneous encoding designed to recover the first $k$ symbols of a DSS at rate $R$ symbols per channel use and error probability $\epsilon$.
\begin{definition}[A $(k,R, \epsilon)$ code with instantaneous encoding]\label{def2} Fix a $(q,\{t_n\}_{n=1}^\infty)$ DSS and fix a DMC with a single-letter transition probability distribution $P_{Y|X}\colon \mathcal X\rightarrow\mathcal Y$. An $(k,R, \epsilon)$ code with instantaneous encoding consists of:\\
1. a sequence of (possibly randomized) encoding functions $\mathsf f_t\colon [q]^{N(t)}\times \mathcal Y^{t-1}\rightarrow \mathcal X$, $t=1,2,\dots$ that the encoder uses to form the channel input
\begin{align}\label{enc}
    X_t \triangleq \mathsf f_t\left(S^{N(t)}, Y^{t-1}\right);
\end{align}
2. a sequence of decoding functions $\mathsf g_{t}\colon \mathcal Y^t \rightarrow [q]^k$, $t=1,2,\dots$ that the decoder uses to form the estimate
\begin{align}\label{Skt_g}
    \hat S^k_t\triangleq \mathsf g_t(Y^t);
\end{align}
3. a stopping time $\eta_k$ adapted to the filtration generated by the channel output $Y_1,Y_2,\dots$ that determines when the transmission stops and that satisfies
\begin{align}\label{time_constraint}
    & \frac{k}{\mathbb E[\eta_k]} \geq R~\text{(symbols per channel use)},\\ \label{error_constraint}
    & \mathbb P[\hat S^k_{\eta_k}\neq  S^k]\leq  \epsilon.
\end{align}
\end{definition} 

For any rate $R>0$, the minimum error probability achievable by rate-$R$ codes with instantaneous encoding and message length $k$ is given by
\begin{equation}
\begin{aligned}
  \epsilon^*(k, R) \triangleq \inf\{\epsilon\colon\exists~&\left(k, R, \epsilon\right)~\text{code}\\ &\text{with instantaneous encoding}\}.
\end{aligned}
\end{equation}

For transmitting a DSS over a non-degenerate DMC with noiseless feedback via a code with instantaneous encoding, we define the \emph{JSCC reliability function for streaming} as
\begin{align}\label{reliabilityfunc2}
     E(R) \triangleq \lim_{k\rightarrow\infty}\frac{R}{k}\log\frac{1}{\epsilon^*(k,R)}.
\end{align}

If a DSS satisfies \eqref{DS}, i.e., a DS, a code with instantaneous encoding (i.e., causal code) in Definition~\ref{def2} reduces to a code with block encoding (i.e., non-causal code), and the JSCC reliability function for streaming \eqref{reliabilityfunc2} reduces to the JSCC reliability function for a fully accessible source.

We use $E(R)$ \eqref{reliabilityfunc2} to quantify the fundamental delay-reliability tradeoff achieved by codes with instantaneous encoding. The reliability function is a classical performance metric that can be used to approximate that tradeoff as $\epsilon \simeq e^{-\frac{k}{R}E(R)}$. Although this approximation ignores the subexponential terms, it still sheds light on the finite-blocklength performance as our numerical simulations in Section~\ref{Sec_simulation} demonstrate.

Similar to classical codes with block encoding, a $(k,R,\epsilon)$ code with instantaneous encoding in Definition~\ref{def2} is designed to recover only the first $k$ symbols of a DSS, and $E(R)$ \eqref{reliabilityfunc2} is achieved by a sequence of codes with instantaneous encoding indexed by the length of the symbol sequence $k$ as $k\rightarrow\infty$. We proceed to define a code with instantaneous encoding that decodes the first $k$ symbols at a time $t\geq t_k$ with an error probability that decays exponentially with delay $t-t_k$, for all $k$ and $t$. Because the decoding time and the number of symbols to decode can be chosen on the fly, this code is referred to as an \emph{anytime} code and can be used to stabilize an unstable linear system with bounded noise over a noisy channel with feedback \cite{Sahai}. We formally define anytime codes as follows.

\begin{definition}[A $(\kappa, \alpha)$ anytime code]\label{def_anytime} Fix a $(q,\{t_n\}_{n=1}^{\infty})$ DSS and fix a DMC with a single-letter transition probability distribution $P_{Y|X}\colon \mathcal X\rightarrow\mathcal Y$. A $(\kappa,\alpha)$ anytime code consists of:\\
1. a sequence of (possibly randomized) encoding functions defined in Definition~\ref{def2}-1;\\
2. a sequence of decoding functions $\mathsf g_{t,k}\colon \mathcal Y^t\rightarrow [q]^k$ indexed both by the decoding time $t$ and the length of the decoded symbol sequence $k$ that the decoder uses to form an estimate $\hat S^k_t\triangleq \mathsf g_{t,k}(Y^t)$ of the first $k$ symbols at time $t$.\\
For all $k=1,2,\dots$, $t=1,2,\dots$, $t\geq t_k$, the error probability of decoding the first $k$ symbols at time $t$ must satisfy
\begin{align}\label{anytime_error}
    \mathbb P[\hat S^k_t \neq S^k] \leq \kappa e^{-\alpha(t-t_k)}
\end{align}
for some $\kappa,\alpha\in\mathbb R_+$.
\end{definition}
The exponentially decaying rate $\alpha$ of the error probability in \eqref{anytime_error} is referred to as the anytime reliability. While Sahai and Mitter's anytime code in \cite[Definition 3.1]{Sahai} is defined to transmit a DSS that emits source symbols one by one at consecutive times, Definition~\ref{def_anytime} slightly extends \cite[Definition 3.1]{Sahai} to a general DSS in Definition~\ref{def_dss}.

In this paper, we aim to find $E(R)$ \eqref{reliabilityfunc2}, the codes with instantaneous encoding that achieve $E(R)$, and an anytime code.

\section{Instantaneous encoding phase}\label{belief_phase}
With the aim of transmitting the first $k$ source symbols of a DSS, we present our instantaneous encoding phase, which specifies the encoding functions $\{\mathsf f_t\}_{t=1}^{t_k}$ in Definition~\ref{def2}.
We fix a DMC with a single-letter transition probability distribution $P_{Y|X}\colon \mathcal X\rightarrow\mathcal Y$ and capacity-achieving distribution $P_{X}^*$, and we fix a $(q,\{t_n\}_{n=1}^\infty)$ DSS with distribution \eqref{bit_prob}. We denote the following functions of the channel outputs,
\begin{align}\label{rho2}
    &\rho_i(Y^t) \triangleq P_{S^{N(t)}|Y^t}(i|Y^t),\\
    &\theta_i(Y^{t-1}) \triangleq P_{S^{N(t)}|Y^{t-1}}(i|Y^{t-1}),\\\label{pi2}
    &\pi_{x}(Y^{t-1}) \triangleq \sum_{i\in\mathcal G_x(Y^{t-1})}\theta_i(Y^{t-1}),
\end{align}
where we refer to $\rho_i(Y^t)$ and $\theta_i(Y^t)$ as the posterior and the prior of source sequence $i\in[q]^{N(t)}$, respectively; we refer to $\pi_{x}(Y^{t-1})$ as the prior of the group $\mathcal G_x(Y^{{t-1}})$ corresponding to channel input $x\in\mathcal X$ that we specify in \eqref{SD} below. The probability distributions $P_{S^{N(t)}|Y^t}$ and $P_{S^{N(t)}|Y^{t-1}}$ are determined by the code below.

\emph{Algorithm}: The instantaneous encoding phase operates during times $t=1,2,\dots,t_k$.

At each time $t$, the encoder and the decoder first update the priors $\theta_i(y^{t-1})$ for all $i\in [q]^{N(t)}$. At symbol arriving times $t = t_n$, $n=1,2,\dots,k$ the prior $\theta_i(y^{t-1})$, $i\in [q]^{N(t)}$ is updated using the posterior $\rho_{i^{N(t-1)}}(y^{t-1})$ and the source distribution \eqref{bit_prob}, i.e.,
\begin{align}\label{post_prior}
\theta_i(y^{t-1}) = P_{S^{N(t)}|S^{N(t-1)}}\left(i|i^{N(t-1)}\right)\rho_{i^{N(t-1)}}(y^{t-1}),
\end{align}
where $i^{N(t-1)}$ is the length-$N(t-1)$ prefix of sequence $i$.
At times in-between arrivals, i.e., at $t\in (t_n,t_{n+1})$, $n=1,2,\dots,k-1$, the prior $\theta_i(y^{t-1})$ is equal to the posterior $\rho_i(y^{t-1})$ for all $i\in[q]^{N(t)}$, i.e.,
\begin{align}\label{theta_eq_rho}
  \theta_i(y^{t-1}) = \rho_i(y^{t-1}).   
\end{align}

At each time $t$, once the priors are updated, the encoder and the decoder partition the message alphabet $[q]^{N(t)}$ into $|\mathcal X|$ disjoint groups $\{\mathcal G_x(y^{t-1})\}_{x\in\mathcal X}$ such that for all $x\in\mathcal X$,
\begin{equation}\label{SD}
\begin{aligned}
   \pi_{x}(y^{t-1}) - P_{X}^*(x) \leq \min_{i\in\mathcal G_x(y^{t-1})}\theta_i(y^{t-1}).
\end{aligned}
\end{equation}
The partitioning rule \eqref{SD} ensures that the group priors $\{\pi_{ x}(y^{t-1})\}_{x\in\mathcal X}$  are close enough to the capacity-achieving distribution $\{P_{X}^*(x)\}_{x\in\mathcal X}$. There always exists a partition $\{\mathcal G_x(y^{t-1})\}_{x\in\mathcal X}$ of $[q]^{N(t)}$ that satisfies the partitioning rule \eqref{SD}, since the partition given by the \emph{greedy heuristic} algorithm \cite{Korf} satisfies it, see the algorithm and the proof in Appendix~\ref{SD_exists}.

Using the partition $\{\mathcal G_x(y^{t-1})\}_{x\in\mathcal X}$, the encoder and the decoder construct two sets by comparing the group priors $\{\pi_{x}(y^{t-1})\}_{x\in\mathcal X}$ with the capacity-achieving distribution $\{P_{X}^*(x)\}_{x\in\mathcal X}$:
\begin{align}\label{XbarXunder}
    &\underline{\mathcal X}(y^{t-1}) \triangleq \{x\in\mathcal X:\pi_{x}(y^{t-1})\leq P_{X}^*(x)\},\\\label{XbarXunder2}
    &\overline {\mathcal X}(y^{t-1}) \triangleq \{x\in\mathcal X: \pi_{x}(y^{t-1})> P_{X}^*(x)\}.
\end{align}

Then, the encoder and the decoder determine a set of probabilities $\{p_{\overline x\rightarrow \underline{x}}\}_{\overline x\in \overline{\mathcal X}(y^{t-1}), \underline{x}\in\underline{\mathcal X}(y^{t-1})}$ for randomizing the channel input, such that for all $\overline x\in\overline{\mathcal X}(y^{t-1})$, $\underline{x}\in\underline{\mathcal X}(y^{t-1})$, it holds that
\begin{align}\label{pxy}
    &\pi_{\overline x}(y^{t-1}) - \sum_{\underline x\in\underline{\mathcal X}(y^{t-1})}p_{\overline x\rightarrow \underline x} = P_{X}^*(\overline x),\\ \label{pxy2}
    &\pi_{\underline x}(y^{t-1}) + \sum_{\overline x\in\overline{\mathcal X}(y^{t-1})}p_{\overline x\rightarrow \underline x} = P_{X}^*(\underline x).
\end{align}
The set $\{p_{\overline x\rightarrow \underline{x}}\}_{\overline x\in \overline{\mathcal X}(y^{t-1}), \underline{x}\in\underline{\mathcal X}(y^{t-1})}$ can be determined by the algorithm in Appendix~\ref{set_pxx}.

The output of the encoder is formed via randomization as follows. The encoder first determines the group that contains the sequence $S^{N(t)}$ it received so far:
\begin{align}\label{Zt1}
    Z_t  \triangleq \sum_{x\in\mathcal X}x\mathbbm{1}_{\mathcal G_x(y^{t-1})}\left(S^{N(t)} \right).
\end{align}
Then, the encoder outputs $X_t$ according to
\begin{align}\nonumber
    &P_{X_t|Z_t,Y^{t-1}}(x|z,y^{t-1})\\\label{PXZY}
    =& \begin{cases}
    \frac{P_{X}^*(z)}{\pi_{z}(y^{t-1})}, &\text{if}~x=z, z\in\overline{\mathcal X}(y^{t-1}),\\
    \frac{p_{z\rightarrow x}}{\pi_{z}(y^{t-1})}, &\text{if}~x \in\underline{\mathcal X}(y^{t-1}), z\in\overline{\mathcal X}(y^{t-1})\\
    \mathbbm{1}_{\{z\}}(x),~ &\text{if}~ z\in\underline{\mathcal X}(y^{t-1}),\\
    0,~&\text{otherwise}.
    \end{cases}
\end{align}

The decoder also knows the randomization distribution $P_{X_t|Z_t,Y^{t-1}}$ \eqref{PXZY}, since it knows group priors $\{\pi_x(y^{t-1})\}_{x\in\mathcal X}$ \eqref{SD}, sets $\overline {\mathcal X}(y^{t-1})$ and $\underline{\mathcal X}(y^{t-1})$ \eqref{XbarXunder}--\eqref{XbarXunder2}, and probabilities $\{p_{\overline x\rightarrow \underline{x}}\}_{\overline x\in \overline{\mathcal X}(y^{t-1}), \underline{x}\in\underline{\mathcal X}(y^{t-1})}$ \eqref{pxy}--\eqref{pxy2}.
Due to \eqref{XbarXunder}--\eqref{PXZY}, the channel input distribution at time $t=1,2,\dots,t_k$, is equal to the capacity-achieving channel input distribution, i.e., for all $y^{t-1}\in\mathcal Y^{t-1}$,
\begin{align}\label{PXYPX*}
    P_{X_t|Y^{t-1}}(x|y^{t-1}) = P_{X}^*(x).
\end{align}
See the proof of \eqref{PXYPX*} in Appendix~\ref{pf_PXYPX*}. Fig.~\ref{Fig_rand} below provides an example of group partitioning and channel input randomization.
\begin{figure}[h!]
\centering
\includegraphics[trim = 50mm 232mm 55mm 33mm, clip, width=13cm]{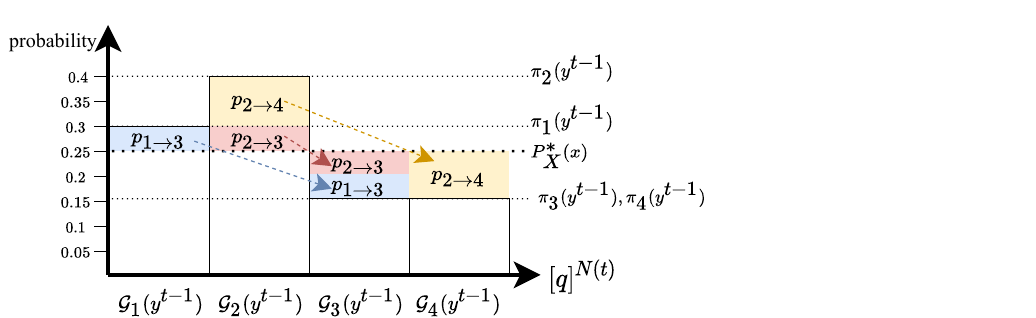}
\caption{An example of group partitioning and channel input randomization for a DMC with uniform capacity-achieving distribution $P_{X}^*(x)=0.25$, $\mathcal X=[4]$. The horizontal axis represents a partition of $4$ groups. The vertical axis represents the prior probabilities of the groups. The source alphabet $[q]^{N(t)}$ is partitioned into $\{\mathcal G_{x}(y^{t-1})\}_{x\in[4]}$ such that the partitioning rule \eqref{SD} is satisfied. Groups $\mathcal G_{x}(y^{t-1}), x\in\{1,2\}$ constitute $\overline{\mathcal X}(y^{t-1})$ \eqref{XbarXunder2} and groups $\mathcal G_{x}(y^{t-1}), x\in\{3,4\}$ constitute $\underline{\mathcal X}(y^{t-1})$ \eqref{XbarXunder}. The probabilities $\{p_{\overline x\rightarrow \underline{x}}\}_{\overline x\in \{1,2\}, \underline{x}\in\{3,4\}}$ \eqref{pxy}--\eqref{pxy2} used to randomize transmitted group indices are colored. The randomization matches the probability of transmitting group index $x\in[4]$ to $P_{X}^*(x)$.} 
\label{Fig_rand}
\end{figure}

Upon receiving the channel output $Y_t=y_t$ at time $t$, the encoder and the decoder update the posteriors $\rho_i(y^t)$ for all possible sequences of source symbols $i\in[q]^{N(t)}$ using the prior $\theta_i(y^{t-1})$, the channel output $y_t$, and the randomization probability \eqref{PXZY}, i.e.,
\begin{align}\label{prior_post}
    &\rho_i(y^{t})=\frac{\sum_{x\in\mathcal X} P_{Y|X}(y_t|x)P_{X_t|Z_t,Y^{t-1}}(x|z(i),y^{t-1}) }{P_{Y}^*(y_t)}\theta_i(y^{t-1}),
\end{align}
where $z(i)$ is the index of the group that contains sequence $i$, i.e., it is equal to the right side of \eqref{Zt1} with $S^{N(t)}\leftarrow i$; $P_{Y}^*$ is the channel output distribution induced by the capacity-achieving distribution $P_{X}^*$; \eqref{prior_post} holds due to \eqref{PXYPX*} and the Markov chain $Y_t-X_t - (Z_t,Y^{t-1})-S^{N(t)}$. 

We conclude the presentation of the instantaneous encoding phase with several remarks.

The randomization \eqref{XbarXunder}--\eqref{PXZY} of the instantaneous encoding phase is only used for analysis: Theorem~\ref{thm_2} in Section~\ref{Sec_reliability2} continues to hold if the randomization step \eqref{XbarXunder}--\eqref{PXZY} is dropped  and the deterministic group index $Z_t$ \eqref{Zt1} is transmitted, but at a cost of imposing assumptions on the DSS that are stricter than assumptions $(\mathrm{a})$--$(\mathrm{b})$ in Theorem~\ref{thm_2}. See Remark~\ref{rmk_drop_random} in Section~\ref{Sec_reliability2} for details.
From the perspective of encoding, the randomization \eqref{PXZY} turns the encoding function $\mathsf f_t$ into a stochastic kernel $P_{X_t|S^{N(t)}, Y^{t-1}}$. From the perspective of the channel, the randomization $P_{X_t|Z_t, Y^{t-1}}$ \eqref{PXZY} together with the DMC $P_{Y|X}$ can be viewed as a cascaded DMC with channel input $(Z_t, Y^{t-1})$. The randomness in \eqref{Zt1} is not common with the decoder as it only needs to know the distribution $P_{X_t|Z_t, Y^{t-1}}$ to update posterior $\rho_i(y^t)$ in \eqref{prior_post}.

The complexity of the instantaneous encoding phase is $O\left(q^{N(t)}\log q^{N(t)}\right)$ if the classical greedy heuristic algorithm (Appendix~\ref{SD_exists}) is used for group partitioning \eqref{SD}. For equiprobably distributed source symbols, we design an efficient algorithm that reduces the complexity down to $O(t\log t)$ in Section~\ref{Sec_practical}. 

\section{JSCC reliability function}\label{Sec_reliability2}
In this section, we show the JSCC reliability function for streaming $E(R)$ \eqref{reliabilityfunc2} using the instantaneous encoding phase introduced in Section~\ref{belief_phase}. For brevity, we denote the maximum and the minimum channel transition probabilities of a DMC $P_{Y|X}\colon \mathcal X\rightarrow\mathcal Y$ by
\begin{align}\label{P_max}
    p_{\max} \triangleq \max_{x\in\mathcal X,y\in\mathcal Y} P_{Y|X}(y|x),\\\label{def_lambda}
    p_{\min} \triangleq \min_{x\in\mathcal X,y\in\mathcal Y} P_{Y|X}(y|x),
\end{align}
and we denote the maximum symbol arriving probability of the DSS \eqref{bit_prob} by
\begin{align}\label{PSMAX}
    p_{S,\max} \triangleq \max_{n\in\mathbb N, s\in[q], s'\in[q]^{n-1}}P_{S_n|S^{n-1}}(s|s').
\end{align}

\begin{theorem}\label{thm_2}
Fix a non-degenerate DMC with capacity $C$ \eqref{capacity}, maximum KL divergence $C_1$ \eqref{C1}, and maximum channel transition probability $p_{\max}$ \eqref{P_max}. Fix a $(q,\{t_n\}_{n=1}^{\infty})$ DSS with entropy rate $H>0$ \eqref{entropy_rate} and symbol arriving rate $f$ \eqref{freq}. If the DSS has $f<\infty$ \eqref{freq}, then we further assume that
\begin{itemize}
    \item[$(\mathrm{a})$] the information in the DSS is asymptotically lower bounded as
    \begin{align}\label{assump_a}
        \lim_{n\rightarrow\infty}\mathbb P\left[\frac{1}{n}\log\frac{1}{P_{S^n}(S^n)}\geq \bunderline{H} \right]=1
    \end{align}
    for some $\bunderline{H}\in(0,\infty)$;
    \item[$(\mathrm{b})$] the symbol arriving rate is large enough:
    \begin{align}\label{assump_b}
     f>\frac{1}{\bunderline{H}}\left(H(P_{Y}^*) - \log\frac{1}{p_{\max}}\right).
    \end{align}
\end{itemize}
Then, the JSCC reliability function for streaming \eqref{reliabilityfunc2} is equal to
\begin{align}\label{converse2}
    E(R) = C_1\left(1-\frac{H}{C}R\right),~ 0< R <\frac{C}{H}.
\end{align}
\end{theorem}
\begin{proof}[Proof sketch]
The converse proof is in Appendix~\ref{pf_converse}: allowing the encoder to know the entire source sequence before the transmission will not reduce the JSCC reliability function, therefore converse bounds for a (fully accessible) DS apply. Namely, we lower bound the expected decoding time for any code with block encoding to attain a target error probability using Fano's inequality and a binary hypothesis test. This extends Berlin et al.'s \cite{Berlin} converse bound on Burnashev's reliability function applicable to the channel coding setting to the JSCC setting.

The achievability proof is in Appendix~\ref{pf_achieve}: for any (fully accessible) DS, the JSCC reliability function \eqref{converse2} is achievable by the MaxEJS code \cite[Sec. IV-C]{Naghshvar2}, and is achievable by the SED code \cite[Sec. V-B]{Naghshvar2} if the channel is a symmetric binary-input DMC (Appendix~\ref{pf_achieve_A}). 

For any DSS with $f=\infty$, including the DS \eqref{DS}, the buffer-then-transmit code for $k$ source symbols that achieves $E(R)$ \eqref{converse2} operates as follows. It waits until the $k$-th symbol arrives at time $t_k$, and at times $t\geq t_k+1$, applies a JSCC reliability function-achieving code with block encoding for $k$ symbols $S^k$ of a (fully accessible) DS with prior $P_{S^k}$ \eqref{bit_prob} (e.g., the MaxEJS code or the SED code \cite{Naghshvar2}). The buffer-then-transmit code achieves (see details in Appendix~\ref{pf_achieve_B})
\begin{align}\label{NC_reliability} 
    E(R) \geq C_1\left(1 - \left(\frac{H}{C}+\frac{1}{f}\right)R\right),
\end{align}
which reduces to $E(R)$ \eqref{converse2} for $f=\infty$. Indeed, $f=\infty$ means that the arrival time $t_k$ is negligible compared to the blocklength. The buffer-then-transmit code fails to achieve $E(R)$ \eqref{converse2} if $f<\infty$.

For any DSS with $f<\infty$ that satisfies the assumptions $(\mathrm{a})$--$(\mathrm{b})$ in Theorem~\ref{thm_2}, the code with instantaneous encoding for $k$ source symbols that achieves $E(R)$ \eqref{converse2} implements the instantaneous encoding phase (Section~\ref{belief_phase}) at times $t=1,2,\dots,t_k$ and operates as a JSCC reliability function-achieving code with block encoding for $k$ symbols $S^k$ of a (fully accessible) DS with prior $P_{S^k|Y^{t_k}}$ at times $t\geq t_k+1$, where $Y_1,\dots, Y_{t_k}$ are the channel outputs generated in the instantaneous encoding phase. For example, we can insert the instantaneous encoding phase before the MaxEJS code (or the SED code for symmetric binary-input DMCs). See Appendix~\ref{pf_achieve_C}.
\end{proof}

Assumption $(\mathrm{a})$ holds with $\bunderline{H} = H$ for any information stable source since such sources satisfy $\frac{1}{n}\log \frac{1}{P_{S^n}(S^n)}\xrightarrow{\mathrm{i.p.}} H$ \cite{KD}. For example, $\bunderline{H}=H(S)$ if the source emits i.i.d. symbols. Assumption $(\mathrm{b})$ in Theorem~\ref{thm_2} implies
\begin{align}\label{freq_large}
     f\geq \frac{C}{H}
\end{align}
 since $H(Y|X)\geq \log\frac{1}{p_{\max}}$ and $H\geq \bunderline{H}$. The symbol arriving rate constraint \eqref{freq_large} ensures that all coding rates $R<\frac{C}{H}$ are achievable. Otherwise, if \eqref{freq_large} is not satisfied and the DSS has $p_{S,\max}<1$,
 the rate region achievable by any code with instantaneous encoding is limited to $R\leq f$. The limitation arises because decoding $S^k$ before the final arrival time $t_k$ results in a non-vanishing error probability (Appendix~\ref{decoding_after_tk}). For example, if the DSS emits i.i.d. symbols with entropy rate $H=1$ nat per symbol arriving at the encoder every $1000$ channel uses, and the DMC has capacity $C=1$ nat per channel use, then the achievable rate is limited by $\frac{1}{1000}$ symbols per channel use, which is far less than Shannon's JSCC limit $\frac{C}{H} = 1$ symbol per channel use. While \eqref{freq_large} gives a converse bound on the symbol arriving rate and \eqref{assump_b} guarantees achievability of $E(R)$ \eqref{converse2}, the existence of a critical symbol arriving rate $f_{\mathrm{cr}}$ such that for all $f>f_{\mathrm{cr}}$, $E(R)$ \eqref{converse2} is achievable, and for all $f<f_{\mathrm{cr}}$, $E(R)$ \eqref{converse2} is not achievable, remains open. While $E(R)$ \eqref{converse2} is not a function of $f$, it is conceivable that for $f<f_{\mathrm{cr}}$, the reliability function \eqref{reliabilityfunc2} will depend on $f$. This is reminiscent of the channel reliability function for transmitting over a DMC without feedback via a fixed-length block code, which is known only for rates greater than a critical value where its converse bound (sphere-packing exponent \cite{SGB}) coincides with its achievability bound (random-coding exponent \cite{Gallager_e}). 

Since the (fully accessible) DS \eqref{DS} is a special DSS, Theorem~\ref{thm_2} gives the JSCC reliability function \eqref{converse2} for a fully accessible source. It generalizes Burnashev's reliability function \cite{Burnashev} to the classical JSCC context, and generalizes Truong and Tan's excess-distortion reliability function \cite{Truong} at zero distortion to the DS with memory and to all rates $R<\frac{C}{H}$. 

Remarkably, Theorem~\ref{thm_2} establishes that the JSCC reliability function for a streaming source (satisfying assumptions $(\mathrm{a})$--$(\mathrm{b})$) is equal to that for a fully accessible source. This is surprising as this means that revealing source symbols only causally to the encoder has no detrimental effect on the reliability function.

While the instantaneous encoding phase in Section~\ref{belief_phase} achieves $E(R)$ \eqref{converse2}, in fact, any coding strategy during the symbol arriving period that satisfies
\begin{align}\label{pre_coding_condition}
    \lim_{k\rightarrow\infty} \frac{I(S^k; Y^{t_k})}{t_k} = C
\end{align}
achieves $E(R)$ \eqref{converse2} when followed by a JSCC reliability function-achieving code with block encoding.  This is because \eqref{achieve_sub1}--\eqref{achieve_sub2} in the achievability proof in Appendix~\ref{pf_achieve_C} always hold for such a coding strategy.
For equiprobably distributed $q$-ary source symbols that arrive at the encoder one by one at consecutive times $t=1,2,\dots,k$ and a symmetric $q$-input DMC, uncoded transmission during the symbol arriving period $t=1,2,\dots,k$ satisfies \eqref{pre_coding_condition} and thus constitutes an appropriate instantaneous encoding phase for that scenario. If $q=2$, this corresponds to the systematic transmission phase in \cite{Antonini}. Furthermore, even if the instantaneous encoding phase in Section~\ref{belief_phase} drops the randomization \eqref{XbarXunder}--\eqref{PXZY} and transmits $Z_t$ \eqref{Zt1} as the channel input, it continues to satisfy the sufficient condition \eqref{pre_coding_condition} under a more conservative condition than \eqref{assump_b} (see Remark~\ref{rmk_drop_random} below).

\begin{remark}\label{rmk_drop_random}
Fix a non-degenerate DMC with the maximum and the minimum channel transition probabilities $p_{\max}$ and $p_{\min}$, and fix a $(q,\{t_n\}_{n=1}^{\infty})$ DSS with maximum symbol arriving probability $p_{S,\max}<1$ and symbol arriving rate $f<\infty$. If the DSS satisfies
\begin{itemize}
    \item [$(\mathrm{b}^\prime)$] the symbol arriving rate is large enough:
    \begin{align}\label{freq_new}
        f>\frac{1}{\log\frac{1}{p_{S,\max}}}\left(\log\frac{1}{p_{\min}}-\log\frac{1}{p_{\max}}\right),
    \end{align}
\end{itemize}
then the instantaneous encoding phase in Section~\ref{belief_phase} that transmits the non-randomized $Z_t$ \eqref{Zt1} as the channel input at each time $t=1,2,\dots,t_k$ satisfies \eqref{pre_coding_condition}, which means that it achieves $E(R)$ \eqref{converse2}, the JSCC reliability function for streaming, when followed by a JSCC reliability function-achieving code with block encoding.
\end{remark}
\begin{proof}[Proof sketch] We show that under assumption $(\mathrm{b}^\prime)$, all source priors $\theta_i(y^{t-1})$, $i\in[q]^{N(t)}$, converge pointwise to zero in $t$ during the symbol arriving period $t\in [1,t_k]$ as $k\rightarrow\infty$. The convergent source priors and the partitioning rule \eqref{SD} imply that the group priors converge pointwise to the capacity-achieving distribution $P_X^*$. Since the encoder transmits a group index without randomization as the channel input, the channel input distribution converges to the capacity-achieving distribution, yielding \eqref{pre_coding_condition}. See Appendix~\ref{pf_drop_random} for details.
\end{proof}
Note that the result of Remark~\ref{rmk_drop_random} does not require assumption $(\mathrm{a})$ since $P_{S^n}(s^n)\leq (p_{S,\max})^n$, $\forall s^n\in[q]^n$, already implies that it holds with $\bunderline{H}\leftarrow \log\frac{1}{p_{S,\max}}$.
 
Since $\bunderline{H}\geq \log\frac{1}{p_{S,\max}}$ and $\log\frac{1}{p_{\min}}\geq H(P_Y^*)$, assumption $(\mathrm{b}^\prime)$ is stricter than assumption $(\mathrm{b})$. The increase of the threshold is because 1) the channel output distribution $P_{Y}^*$ in \eqref{theta_ub2} is replaced by $P_{Y_t|Y^{t-1}}(\cdot|\cdot)\geq p_{\min}$ \eqref{sp_ub}; 2) in the proof of Remark~\ref{rmk_drop_random}, we show that all the source priors converge \emph{pointwise} to zero \eqref{as_convergence} during the symbol arriving period as $k\rightarrow\infty$ using the upper bound $P_{S_n|S^{n-1}}(\cdot|\cdot)\leq p_{S,\max}$, whereas in Theorem~\ref{thm_2}, we only need that the the source prior of the true symbol sequence converges \emph{in probability} to zero \eqref{PSNT0_m}.

\section{Instantaneous SED code}\label{Sec_instantaneous_SED}
While the JSCC reliability function-achieving codes with instantaneous encoding in Section~\ref{belief_phase} are designed to transmit the first $k$ symbols of a DSS, and a sequence of such codes indexed by the source length $k$ achieves $E(R)$ \eqref{converse2} as $k\rightarrow\infty$, we now show an anytime code (Definition~\ref{def_anytime}) termed the instantaneous SED code. In Section~\ref{Sec_alg_anytime}, we present the algorithm of the instantaneous SED code for a symmetric binary-input DMC. In Section~\ref{Sec_pos_reliability}, we show by simulations that the instantaneous SED code empirically achieves a positive anytime reliability, and thus can be used to stabilize an unstable linear system with bounded noise over a noisy channel. In Section~\ref{Sec_SED_ER}, we show that if the instantaneous SED code is restricted to transmit the first $k$ symbols of a DSS, a sequence of instantaneous SED codes indexed by the length of the symbol sequence $k$ also achieves $E(R)$ \eqref{converse2} for streaming over a symmetric binary-input DMC.

\subsection{Algorithm of the instantaneous SED code}\label{Sec_alg_anytime}
The instantaneous SED code is almost the same as the instantaneous encoding phase in Section~\ref{belief_phase}, expect that 1) it particularizes the partitioning rule \eqref{SD} to the instantaneous SED rule in \eqref{pXinequal}--\eqref{Gxinequal} below; 2) its encoder does not randomize the channel input and transmits $Z_t$ \eqref{Zt1} at each time $t$; 3) it continues to operate after the symbol arriving period. Fixing a symmetric binary-input DMC $P_{Y|X}\colon\mathcal \{0,1\}\rightarrow \mathcal Y$ and fixing a $(q,\{t_n\}_{n=1}^{\infty})$ DSS, we present the algorithm of the instantaneous SED code.

\emph{Algorithm:} The instantaneous SED code operates at times $t=1,2,\dots$

At each time $t$, the encoder and the decoder first update the priors $\theta_i\left(y^{t-1}\right)$ for all possible sequences $i\in[q]^{N(t)}$ that the source could have emitted by time $t$. If $t=t_n$, $n=1,2,\dots$, the prior is updated using \eqref{post_prior}; otherwise, the prior is equal to the posterior \eqref{theta_eq_rho}.

Once the priors are updated, the encoder and the decoder partition the source alphabet $[q]^{N(t)}$ into $2$ disjoint groups $\{\mathcal G_x\}_{x\in\{0,1\}}$ according to the \emph{instantaneous SED rule}, which says the following: if $x,x'\in\mathcal \{0,1\}$ satisfy
\begin{align}\label{pXinequal}
 \pi_{x}(y^{t-1}) \geq   \pi_{x'}(y^{t-1}),
\end{align}
then they must also satisfy
\begin{align}\label{Gxinequal}
 \pi_{x}(y^{t-1}) - \pi_{x'}(y^{t-1}) \leq \min_{i\in\mathcal G_x(y^{t-1})}\theta_i(y^{t-1}).
\end{align}
There always exists a partition $\{\mathcal G_x(y^{t-1})\}_{x\in\mathcal \{0,1\}}$ that satisfies the instantaneous SED rule \eqref{pXinequal}--\eqref{Gxinequal} since the partition that attains the smallest difference $|\pi_0(y^{t-1})-\pi_1(y^{t-1})|$  satisfies it \cite[Appendix~III-E]{Naghshvar2}.

Once the source alphabet is partitioned, the encoder transmits the index $Z_t$ \eqref{Zt1} of the group that contains the true source sequence $S^{N(t)}$ as the channel input.

Upon receiving the channel output $Y_t = y_t$ at time $t$, the encoder and the decoder update the posteriors $\rho_i\left(y^{t}\right)$ for all $i\in[q]^{N(t)}$ using the priors $\theta_i\left(y^{t-1}\right)$ and the channel output $y_t$, i.e.,
\begin{align}\label{prior_post2}
    \rho_i\left(y^{t-1}\right) = \frac{P_{Y|X}(y_t|z(i))}{\sum_{x\in\mathcal X}P_{Y|X}(y|x)\pi_x(y^{t-1})}\theta_i(y^{t-1}),
\end{align}
where $z(i)$ is the index of the group that contains sequence $i$, i.e., it is equal to the right side of \eqref{Zt1} with $S^{N(t)}\leftarrow i$.

The maximum a posteriori (MAP) decoder estimates the first $k$ symbols at time $t$ as
    \begin{align}\label{hatSNT}
        \hat S^k_t \triangleq \argmax_{i\in [q]^k}P_{S^k|Y^t}(i|Y^t).
    \end{align}

We conclude the presentation of the algorithm with several remarks.

We call the group partitioning rule in \eqref{pXinequal}--\eqref{Gxinequal} instantaneous small-enough difference (SED) rule since it reduces to Naghshvar et al.'s SED rule \cite{Naghshvar2} if the source is fully accessible to the encoder before the transmission. The rule ensures that the difference between a group prior $\pi_x(y^{t-1})$ and its corresponding capacity-achieving probability $P_X^*(x)=\frac{1}{2}$, $x\in\{0,1\}$ is bounded by the source prior on the right side of \eqref{Gxinequal}. 

\begin{remark}
  Even though the algorithm of the instantaneous SED code is presented for a DSS with deterministic symbol arriving times, it can be used to transmit a DSS with \emph{random} symbol arriving times. In that case, the number of symbols $N(t)$ that have arrived by time $t$ is a random variable, and the decoder only knows the symbol arriving distribution $\{P_{S^{N(t)}|S^{N(t-1)}}\}_{t=1}^{\infty}$ rather than the exact symbol arriving times. The instantaneous SED code can be used to transmit such a streaming source as long as the encoder and the decoder keep updating the source priors, partitioning the groups, and updating the posteriors at times $t=1,2,\dots$ for \emph{all possible} source sequences that can arrive at the encoder by time $t$, see our work \cite{Nian} for the algorithm with the instantaneous SED rule replaced by an instantaneous smallest-difference rule. 
  
  Designing a code with instantaneous encoding whose decoder knows neither the symbol arriving times nor the symbol arriving distribution remains an open problem. In this case, the decoder needs to learn the symbol arriving distribution online using the past symbol arriving times. Online learning of the distribution of source symbols is also an interesting research direction.
\end{remark}

\subsection{Instantaneous SED code is an anytime code}\label{Sec_pos_reliability}
We first provide numerical evidence showing that the instantaneous SED code is an anytime code: it empirically attains an error probability that decreases exponentially as \eqref{anytime_error}. We then determine which unstable scalar linear systems can be stabilized by the instantaneous SED code.

In Fig.~\ref{Fig_infinite}, we display the error probability \eqref{anytime_error} of the instantaneous SED code, where the $y$-axis corresponds to the error probability of decoding the length-$k$ prefix of a DSS at time $t$ \eqref{anytime_error}. At each time $t$, we generate a Bernoulli$\left(\frac{1}{2}\right)$ source bit and a realization of a BSC($0.05$), run these experiments for $10^5$ trials, and obtain the error probability \eqref{anytime_error} by dividing the total number of errors by the total number of trials. To reduce the implementation complexity, we simulate the type-based version of the instantaneous SED code in Section~\ref{Type_based_SED}, which has a log-linear complexity. The type-based version is an approximation of the exact instantaneous SED code since it uses an approximating instantaneous SED rule and an approximating decoding rule to mimic the instantaneous SED rule \eqref{pXinequal}--\eqref{Gxinequal} and the MAP decoder \eqref{hatSNT}, respectively, however, it performs remarkably close to the original instantaneous SED code. See Section~\ref{Type_based_SED} for details.
\begin{figure}[h!]
\centering
\includegraphics[trim = 0mm 0mm 0mm 0mm, clip, width=8.5cm]{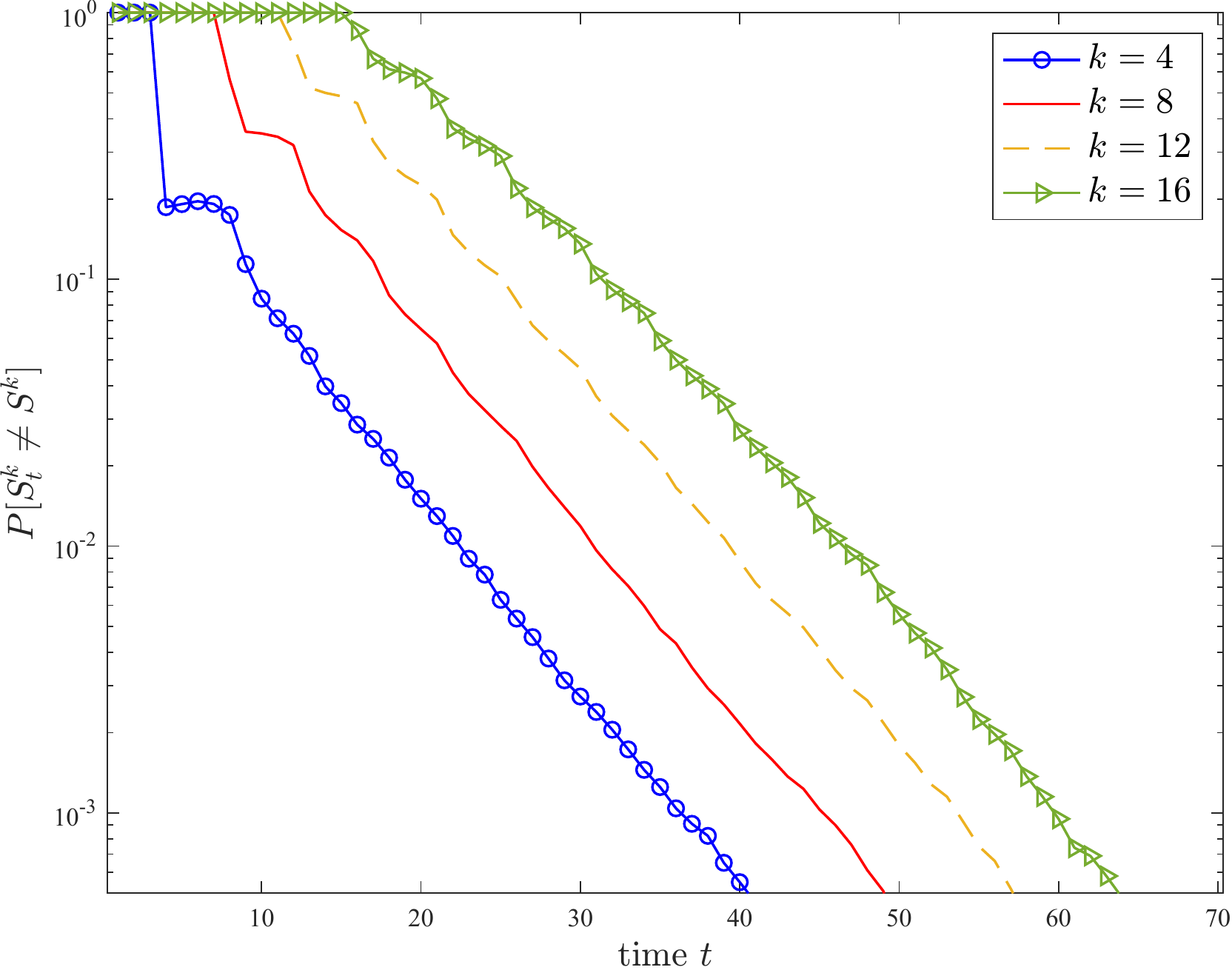}
\caption{The error probability $\mathbb P[\hat S^k_t\neq S^k]$ of decoding the first $k$ symbols of a DSS at time $t$ achieved by the type-based instantaneous SED code (Section~\ref{Type_based_SED}). The DSS emits a Bernoulli$\left(\frac{1}{2}\right)$ bit at times $t=1,2,\dots$. The channel is a BSC($0.05$).}
\label{Fig_infinite}
\end{figure}
The slope of the curves corresponds to the anytime reliability $\alpha$ \eqref{anytime_error} of the instantaneous SED code. The anytime reliability for the source and the channel in Fig.~\ref{Fig_infinite} is approximately equal to $\alpha\approx 0.172$.
The simulation results in Fig.~\ref{Fig_infinite} align with our expectation: the error probability decays exponentially with delay $t-k$ \eqref{anytime_error}, implying that the instantaneous SED code is an anytime code.

We proceed to display the unstable scalar linear system that can be stabilized by the instantaneous SED code. Consider the scalar linear system in Fig.~\ref{Fig_anytime_sys}, $Z_{t+1}=\lambda Z_t + U_t + W_t$,
where $\lambda>1$, $Z_t$ is the real-valued state, $U_t$ is the real-valued control signal, $|W_t|\leq \frac{\Omega}{2}$ is the bounded noise, and the initial state is $Z_1\triangleq 0$. At time $t$, the observer uses the observed states $Z^t$ as well as the past channel feedback $Y^{t-1}$ to form a channel input $X_t$; the controller uses the received channel outputs $Y^t$ to form a control signal $U_t$. For a $(q,\{t_n\}_{n=1}^\infty)$ DSS that emits source symbols one by one at consecutive times $t_n=n$, $n=1,2,\dots$, the anytime rate of a $(\kappa,\alpha)$ anytime code in Definition~\ref{def_anytime} is defined as $R_{\mathrm{any}} = \log q$ nats per channel use, e.g., for the DSS in Fig.~\ref{Fig_infinite}, $R_{\mathrm{any}}=\log 2$; the $\alpha$-anytime capacity $C_{\mathrm{any}}(\alpha)$ is defined as the least upper bound on the anytime rates $R_{\mathrm{any}}$ such that the anytime reliability $\alpha$ is achievable \cite{Sahai}. For such a DSS, Sahai and Mitter \cite[Lemma~4.1 in Sec.~IV-D]{Sahai} showed that the unstable scalar linear system with bounded noise in Fig.~\ref{Fig_anytime_sys} can be stabilized so that $\eta$-th moment $\mathbb E[|Z_t|^{\eta}]$ stays finite at all times, provided that 
    $C_{\mathrm{any}}(\alpha) > \log \lambda$,
    $\alpha>\eta\log \lambda$.
Thus, the instantaneous SED code can be used to stabilize the $\eta$-th moment of the unstable scalar linear system in Fig.~\ref{Fig_anytime_sys} over a BSC(0.05) for any coefficient
\begin{align}
    \lambda &<e^{\min\{R_{\mathrm{any}},\frac{\alpha}{\eta}\}}\\
          &= \min\left\{2, e^{\frac{0.172}{\eta}}\right\}.
\end{align}
E.g., if $\eta=2$, then $\lambda<1.09$. In comparison, the theoretical results in \cite[Corollary~1, Fig.~2]{Lalitha} with $n=1$ show that, for the control over a BSC($0.05$) in Fig.~\ref{Fig_anytime_sys}, Lalitha et al.'s anytime code is only guaranteed to stabilize the $\eta$-th moment of a linear system with $\lambda = 1$.

The control scheme \cite[Sec.~IV]{Sahai} that stabilizes the system in Fig.~\ref{Fig_anytime_sys} employs an anytime code and operates as follows. At each time $t$, the observer computes an $R_{\mathrm{any}}$-nat virtual control signal $\bar U_t$ and acts as an anytime encoder to transmit $\bar U_t$ as the $t$-th symbol of a DSS over a noisy channel with feedback. Here, $\bar U_t$ controls a virtual state $\bar Z_{t+1} = \lambda \bar Z_t + W_t  + \bar U_t$, and is equal to the negative of the $R_{\mathrm{any}}$-nat quantization of $\lambda \bar Z_t$. It ensures the boundedness of $\bar Z_{t+1}$. Upon receiving the channel output, the controller acts as an anytime decoder to refresh its estimate $\hat{\bar U}^t_t$ of $\bar U^t$ and forms a control signal $U_t$ that compensates the past estimation errors of the virtual control signals as if the plant $\{Z_s\}_{s=1}^{t+1}$ was controlled by $\hat{\bar U}^t_t$ heretofore. As a result of applying $U_t$, the actual state $Z_{t+1}$ is forced close to the bounded virtual state $\bar Z_{t+1}$ with the difference $|Z_{t+1}-\bar Z_{t+1}|$ governed by the difference between $\bar U^t_t$ and $\hat {\bar U}^t_t$. The exponentially decaying with $t-k$ error probability of decoding $\bar U^k_t$ achieved by the anytime code together with the bounded $\bar Z_{t+1}$ ensures a finite $\mathbb E[|Z_{t+1}|^{\eta}]$. In fact, the full feedback channel in Fig.~\ref{Fig_anytime_sys} can be replaced by a channel that only feeds the control signal from the controller to the observer, since $Z_t, Z_{t-1}, U_{t-1}$ suffice to compute $W_{t-1}$ and thereby to compute $\bar U_t$ at each time $t$.

\begin{figure}[h!]
\centering
\includegraphics[trim = 20mm 230mm 30mm 20mm, clip, width=9cm]{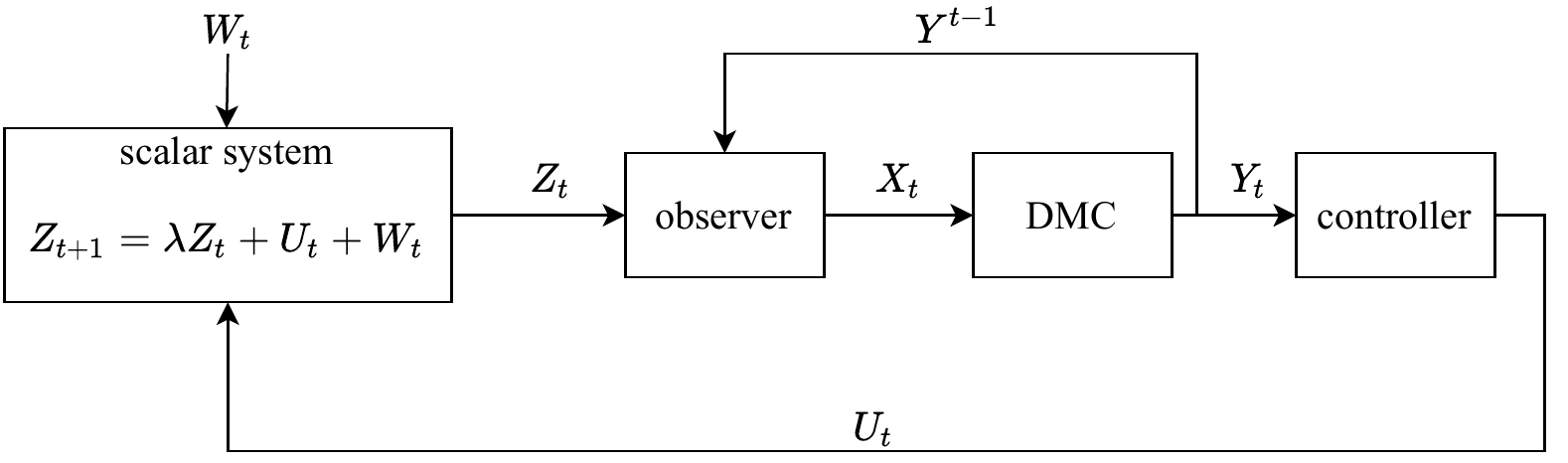}
\caption{A scalar linear system controlled over a noisy channel with noiseless feedback.}
\label{Fig_anytime_sys}
\end{figure}


As verified by the simulations in Fig.~\ref{Fig_infinite}, the instantaneous SED code achieves a positive anytime reliability, however, it is difficult to extend our analysis for $E(R)$ \eqref{converse2} to show that the instantaneous SED code satisfies \eqref{anytime_error} analytically. The submartingales in \cite{Burnashev}\cite{Naghshvar2} used to compute the upper bound on the expected decoding time for a block encoding scheme to attain a target error probability fail to hold if the encoder keeps incorporating newly arrived symbols after time $t_k$. Therefore, we cannot directly use Lemma~\ref{lemma_s1} in Appendix~\ref{pf_achieve_A} to upper bound the expected decoding time, and different tools are needed for the analysis of anytime reliability.

\subsection{Instantaneous SED code achieves $E(R)$}\label{Sec_SED_ER}
We first restrict the instantaneous SED code in Section~\ref{Sec_alg_anytime} to transmit only the first $k$ source symbols of a DSS, and we form a sequence of instantaneous SED codes indexed by the length of the symbol sequence $k$. We then show that the code sequence achieves the JSCC reliability function \eqref{converse2} for streaming over a symmetric binary-input DMC as $k\rightarrow\infty$. 

We restrict the instantaneous SED code in Section~\ref{Sec_alg_anytime} to transmit the first $k$ symbols of a $(q,\{t_n\}_{n=1}^{\infty})$ DSS as follows.
\begin{itemize}
\item[1)] The alphabet $[q]^{N(t)}$ that contains all possible sequences that could have arrived by time $t$ is replaced by the alphabet $[q]^{\min\{N(t),k\}}$ that stops evolving and reduces to $[q]^k$ after all $k$ symbols arrive at time $t_k$. As a consequence, for $t\geq t_k+1$ and all $i\in[q]^k$, the priors $\theta_i(y^{t-1})$ are equal to the corresponding posteriors $\rho_i(y^{t-1})$, the encoder and the decoder partition $[q]^k$ to obtain $\{\mathcal G_x(y^{t-1})\}_{x\in\{0,1\}}$, the encoder transmits the index of the group \eqref{Zt1} that contains $S^k$, and only the posteriors $\rho_i(y^t)$ are updated.
\item [2)] The transmission is stopped and the MAP estimate \eqref{hatSNT} of $S^k$ is produced at the stopping time
\begin{align}\label{eta_k_sub}
    \eta_k \triangleq \min\left\{t\colon \max_{i\in[q]^k} P_{S^k|Y^t}(i|Y^t) \geq 1- \epsilon\right\},  \epsilon\in(0,1).
\end{align}
\end{itemize}

The MAP decoder \eqref{hatSNT} together with the stopping rule \eqref{eta_k_sub} ensures the error constraint in \eqref{error_constraint}, since the MAP decoder \eqref{hatSNT} implies $\mathbb P[\hat S^k_{\eta_k}= S^k] = \mathbb E\left[\mathbb E\left[\mathbbm{1}_{\left\{\hat S^k_{\eta_k}\right\}}(S^k)\middle | Y^{\eta_k}\right] \right]=\mathbb E\left[\max_{i\in[q]^k} P_{S^k|Y^{\eta_k}}(i|Y^{\eta_k})\right]$, which is lower bounded by $1-\epsilon$ due to the stopping time \eqref{eta_k_sub}.

\begin{theorem}\label{Prop_SED}
Fix a non-degenerate symmetric binary-input DMC and a $(q,\{t_n\}_{n=1}^{\infty})$ DSS satisfying assumption $(\mathrm{b}^\prime)$ in Remark~\ref{rmk_drop_random}. The sequence of instantaneous SED codes for transmitting the first $k$ symbols of the DSS achieves $E(R)$ \eqref{converse2} as $k\rightarrow\infty$.
\end{theorem}
\begin{proof}
First, we observe that after the symbol arriving period $t\geq t_k+1$, the instantaneous SED code reduces to the SED code \cite[Sec. V-B]{Naghshvar2} because the instantaneous SED rule \eqref{pXinequal}--\eqref{Gxinequal} reduces to the SED rule \cite[Eq. (50)]{Naghshvar2} if all $k$ source symbols are fully accessible \eqref{DS} to the encoder. The SED code \cite{Naghshvar2} achieves the JSCC reliability function \eqref{converse2} for transmitting a fully accessible source over a non-degenerate symmetric binary-input DMC (Appendix~\ref{pf_achieve_A}).

Second, we observe that during the symbol arriving period $t=1,2\dots,t_k$, the instantaneous SED code corresponds to dropping the randomization step of the instantaneous encoding phase in Section~\ref{belief_phase}. This is because for a symmetric binary-input DMC, \eqref{pXinequal} implies $\pi_{x'}(y^{t-1})\leq P_{X}^*(x')= \frac{1}{2}$, thus any partition $\{\mathcal G_x(y^{t-1})\}_{x\in\mathcal \{0,1\}}$ that satisfies the instantaneous SED rule \eqref{pXinequal}--\eqref{Gxinequal} also satisfies the partitioning rule in \eqref{SD}. Therefore, Remark~\ref{rmk_drop_random} implies that the instantaneous SED code at times $t=1,2,\dots,t_k$ satisfies the sufficient condition \eqref{pre_coding_condition} under assumption $(\mathrm{b}^\prime)$.

As we have discussed in the proof sketch of Theorem~\ref{thm_2}, a JSCC reliability function-achieving code with instantaneous encoding can be obtained by preceding a JSCC reliability function-achieving code with block encoding by an instantaneous encoding phase that satisfies \eqref{pre_coding_condition}. The two observations above imply that the instantaneous SED code achieves $E(R)$ \eqref{converse2} in the setting of Theorem~\ref{Prop_SED}.
\end{proof}

\section{Low-complexity codes with instantaneous encoding}\label{Sec_practical}
We present the type-based algorithms for the instantaneous encoding phase in Section~\ref{belief_phase}, for the instantaneous SED code as an anytime code in Section~\ref{Sec_alg_anytime}, and for the instantaneous SED code restricted to transmit $k$ symbols only in Section~\ref{Sec_SED_ER}. The type-based instantaneous encoding phase is the exact phase in Section~\ref{belief_phase}, whereas the type-based instantaneous SED codes are approximations of the original codes in Sections~\ref{Sec_alg_anytime}~and~\ref{Sec_SED_ER}. All the type-based codes that we are about to see have a log-linear complexity $O(t\log t)$ in time $t$.

We assume that the source symbols of the DSS are equiprobably distributed, i.e., the source distribution \eqref{bit_prob} satisfies
\begin{align}\label{channel-source}
    P_{S_n|S^{n-1}}(a|b) = \frac{1}{q},
\end{align}
for all $a\in[q]$, $b\in[q]^{n-1}$, $n=1,2,\dots$

In our type-based codes, the evolving source alphabet is judiciously divided into disjoint sets that we call \emph{types}, so that the source sequences in each type share the same prior and the same posterior. Here, the same prior is guaranteed by the equiprobably distributed symbols \eqref{channel-source}, and the same posterior is guaranteed by moving a whole type to a group during the group partitioning process (see step (iii) below). As a consequence of classifying source sequences into types, the prior update, the group partitioning, and the posterior update can be implemented in terms of types rather than individual source sequences, which results in an exponential reduction of complexity.


We denote by $\mathcal S_1,\mathcal S_2,\dots$ a sequence of types. We slightly abuse the notation to denote by $\theta_{\mathcal S_j}(Y^{t-1})$ and $\rho_{\mathcal S_j}(Y^t)$ the prior and the posterior of a single source sequence in type $\mathcal S_j$ at time $t$ rather than the prior and the posterior of the whole type. We fix a $(q,\{t_n\}_{n=1}^{\infty})$ DSS that satisfies \eqref{channel-source} and fix a DMC with a single-letter transition probability $P_{Y|X}\colon \mathcal X\rightarrow\mathcal Y$.

\subsection{Type-based instantaneous encoding phase}\label{type-based_encoding_phase}
The type-based instantaneous encoding phase operates at times $t=1,2,\dots, t_k$, where $k$ is the number of source symbols of a DSS that we aim to transmit.

(i) \emph{Type update:} At each time $t$, the algorithm first updates the types. At $t=1$, the algorithm is initialized with one type $\mathcal S_1\triangleq [q]^{N(1)}$. At $t=t_n$, $n=2,\dots,k$, the algorithm updates all the existing types by appending every sequence in $[q]^{N(t)-N(t-1)}$ to every sequence in the type. After the update, the length of the source sequences in each type is equal to $N(t)$; the cardinality of each type is multiplied by $q^{N(t)-N(t-1)}$; the total number of types remains unchanged. At $t\neq t_n$, $n=1,2,\dots,k$, the algorithm does not update the types.

(ii) \emph{Prior update:} Once the types are updated, the algorithm proceeds to update the prior of the source sequences in each existing type. The prior $\theta_{\mathcal S_j}(y^{t-1})$, $j=1,2,\dots$ of the source sequences in type $\mathcal S_j$ is fully determined by \eqref{post_prior} with $\theta_i(y^{t-1})\leftarrow \theta_{\mathcal S_j}(y^{t-1})$, $P_{S^{N(t)}|S^{N(t-1)}}(\cdot|\cdot)\leftarrow \left(\frac{1}{q}\right)^{N(t)-N(t-1)}$, and $\rho_{i^{N(t-1)}}(y^{t-1})\leftarrow \rho_{\mathcal S_j}(y^{t-1})$. If the types are not updated, the priors are equal to the posteriors, i.e., $\theta_{\mathcal S_j}(y^{t-1})\leftarrow \rho_{\mathcal S_j}(y^{t-1})$, $j=1,2,\dots$

(iii) \emph{Group partitioning:} Using all the existing types and their priors, the algorithm determines a partition $\{\mathcal G_x(y^{t-1})\}_{x\in\mathcal X}$ that satisfies the partitioning rule \eqref{SD} via a \emph{type-based greedy heuristic algorithm}. It operates as follows.  It initializes all the groups $\{\mathcal G_x(y^{t-1})\}_{x\in\mathcal X}$ by empty sets and initializes the group priors $\{\pi_x(y^{t-1})\}_{x\in\mathcal X}$ by zeros.
It forms a queue by sorting all the existing types according to priors $\theta_{\mathcal S_j}(y^{t-1})$, $j=1,2,\dots$ in a descending order. It moves the types in the queue one by one to one of the groups $\{\mathcal G_x(y^{t-1})\}_{x\in\mathcal X}$. Before each move, it first determines a group $\mathcal G_{x^*}(y^{t-1})$ whose current prior $\pi_{x^*}(y^{t-1})$ has the largest gap to the corresponding capacity-achieving probability $P_{X}^*(x^*)$, 
\begin{align}\label{xstar_max}
    x^*\triangleq \arg\max_{x\in\mathcal X} P_{X}^*(x)-\pi_x(y^{t-1}).
\end{align}
Suppose the first type in the sorted queue, i.e., the type whose sequences have the largest prior, is $\mathcal S_j$. It then proceeds to determine the number of sequences that are moved from type $\mathcal S_j$ to group $\mathcal G_{x^*}(y^{t-1})$ by calculating
\begin{align}\label{ceil_n}
    n \triangleq \left\lceil \frac{P_{X}^*(x^*)-\pi_{x^*}(y^{t-1})}{\theta_{\mathcal S_j}(y^{t-1})} \right\rceil.
\end{align}
If $n\geq |\mathcal S_j|$, then it moves the whole type $\mathcal S_j$ to group $\mathcal G_{x^*}(y^{t-1})$; otherwise, it splits $\mathcal S_j$ into two types by keeping the smallest or the largest $n$ consecutive\footnote{This step ensures that all sequences in a type are consecutive. Thus, as we will discuss in the last paragraph in Section~\ref{type-based_encoding_phase}, it is sufficient to store two sequences, one with the smallest and one with the largest lexicographic orders, in a type to fully specify that type.} (in lexicographic order) sequences in $\mathcal S_j$ and transferring the rest into a new type, and it moves type $\mathcal S_j$ to group $\mathcal G_{x^*}(y^{t-1})$ and moves the new type to the beginning of the queue. It updates the prior $\pi_{x^*}(y^{t-1})$ after each move.

(iv) \emph{Randomization:} The type-based instantaneous encoding algorithm implements the randomization in \eqref{XbarXunder}--\eqref{PXZY} with respect to a partition $\{\mathcal G_{x}(y^{t-1})\}_{x\in\mathcal X}$. 

(v) \emph{Posterior update:} Upon receiving the channel output $Y_t=y_t$, the algorithm updates the posterior of the source sequences in each existing type. The posterior $\rho_{\mathcal S_j}(y^t)$, $j=1,2,\dots$ of the source sequences in type $\mathcal S_j$ is fully determined by \eqref{prior_post} with $\rho_i(y^t)\leftarrow \rho_{\mathcal S_j}(y^t)$, $\theta_i(y^{t-1})\leftarrow \theta_{\mathcal S_j}(y^{t-1})$.

Using \eqref{ceil_n} and Appendix~\ref{SD_exists}, we conclude that the type-based greedy heuristic algorithm ensures \eqref{SD}. 

We show that the complexity of the type-based instantaneous encoding phase is log-linear $O(t\log t)$ at times $t=1,2,\dots, t_k$. We first show that the number of types grows linearly, i.e., $O(t)$. Since the type update in step (i) does not add new types, the number of types increases only due to the split of types during group partitioning in step (iii). At most $|\mathcal X|$ types are split at each time. This is because the ceiling in \eqref{ceil_n} ensures that the group that receives the $n$ sequences from a split type will have a group prior no smaller than the corresponding capacity-achieving probability, thus the group will no longer be the solution to the maximization problem \eqref{xstar_max} and will not cause the split of other types. We proceed to analyze the complexity of each step of the algorithm. Step (i) (type update) has a linear complexity in the number of types, i.e., $O(t)$. This is because the methods of updating and splitting a type in steps (i) and (iii) ensure that the sequences in any type are consecutive, thus it is sufficient to store the starting and the ending sequences in each type to fully specify all the sequences in that type. As a result, updating a type is equivalent to updating the starting and the ending sequences of that type. Step (ii) (prior update) and step (v) (posterior update) have a linear complexity in the number of types, i.e., $O(t)$. Step (iii) (group partitioning) has a log-linear complexity in the number of types due to type sorting, i.e., $O(t\log t)$. This is because the average complexity of sorting a sequence of numbers is log-linear in the size of the sequence \cite{Cormen}. Step (iv) (randomization) has complexity $O(1)$ due to determining $\{p_{\overline x\rightarrow \underline{x}}\}_{\overline x\in \overline{\mathcal X}(y^{t-1}), \underline{x}\in\underline{\mathcal X}(y^{t-1})}$ in \eqref{pxy}--\eqref{pxy2}.

\subsection{Type-based instantaneous SED codes}\label{Type_based_SED}
We present type-based codes for the anytime instantaneous SED code in Section~\ref{Sec_alg_anytime} and for the instantaneous SED code restricted to transmit $k$ symbols in Section~\ref{Sec_SED_ER}, respectively.

The type-based anytime instantaneous SED code for a symmetric binary-input DMC operates at times $t=1,2,\dots$: 

(i$^\prime$) \emph{Type update:} At each time $t$, the algorithm updates types as in step (i) with $k=\infty$.

(ii$^\prime$) \emph{Prior update:} The algorithm updates the prior of the source sequences in each existing type as in step (ii) with $k=\infty$.

(iii$^\prime$) \emph{Group partitioning:} Using all the existing types and their priors, the algorithm determines a partition $\{\mathcal G_x(y^{t-1})\}_{x\in\{0,1\}}$ using an \emph{approximating} instantaneous SED rule that mimics the exact rule in \eqref{pXinequal}--\eqref{Gxinequal} as follows. It forms a queue by sorting all the existing types according to priors $\theta_{\mathcal{S}_j}(y^{t-1})$, $j=1,2,\dots$ in a descending order. It moves the types in the queue one by one to $\mathcal G_0(y^{t-1})$ until $\pi_0(y^{t-1})\geq P_X^*(0)= 0.5$ for the first time. Suppose the last type moved to $\mathcal G_0(y^{t-1})$ is $\mathcal S_j$. To make the group priors more even, it then calculates the number of sequences $n$ to be moved away from $\mathcal S_j$ as
\begin{subequations}
\begin{align}\nonumber
n~&\triangleq \argmin_{n\in \{\bunderline{n},\bar n\}}\Big| \left(\pi_0(y^{t-1}) -n\theta_{\mathcal S_j}(y^{t-1})\right)\\\label{nstar}
&- \left(\pi_1(y^{t-1}) + n\theta_{\mathcal S_j}(y^{t-1})\right)\Big|,\\\label{under_n}
\underline{n}~&\triangleq \left\lfloor\frac{\pi_0(y^{t-1})-0.5}{\theta_{\mathcal S_j}(y^{t-1})}\right\rfloor,\\\label{upper_n}
\bar{n}~& \triangleq \left\lceil\frac{\pi_0(y^{t-1})-0.5}{\theta_{\mathcal S_j}(y^{t-1})}\right\rceil.
\end{align}
\end{subequations}
It splits $\mathcal S_j$ into two types by transferring the first or the last $n$ \eqref{nstar} lexicographically ordered sequences in $\mathcal S_j$ to a new type. It moves the new type and all the remaining types in the queue to $\mathcal G_1(y^{t-1})$. 

(iv$^\prime$) The randomization step in (iv) is dropped.

(v$^\prime$) \emph{Posterior update:} The algorithm updates the posteriors of the source sequences in each existing type. The posterior $\rho_{\mathcal S_j}(y^t)$, $j=1,2,\dots$, is fully determined by \eqref{prior_post2} with $\rho_i(y^t)\leftarrow \rho_{\mathcal S_j}(y^t)$, $\theta_i(y^{t-1})\leftarrow \theta_{\mathcal S_j}(y^{t-1})$.

(vi$^\prime$) \emph{Decoding at time $t$:} To decode the first $k$ symbols at time $t$, where $k$ can be any integer that satisfies $t_k\leq t$, the algorithm first finds the type whose source sequences have the largest posterior. Then, it searches for the most probable length-$k$ prefix in that type by relying on the fact that sequences in the same type share the same posterior; thus, the prefix shared by the maximum number of sequences is the most probable one. Namely, the algorithm extracts the length-$k$ prefixes of the starting and the ending sequences, denoted by $i_{\mathrm{start}}^k$ and $i_{\mathrm{end}}^k$, respectively. If $i_{\mathrm{start}}^k= i_{\mathrm{end}}^k$ (Fig.~\ref{Fig_type_prefix}-a), then the decoder outputs $\hat S^k_t = i_{\mathrm{start}}^k$. If $i_{\mathrm{start}}^k$ and $i_{\mathrm{end}}^k$ are not lexicographically consecutive (Fig.~\ref{Fig_type_prefix}-b), then the decoder outputs a length-$k$ prefix in between the two prefixes. If $i_{\mathrm{start}}^k$ and $i_{\mathrm{end}}^k$ are lexicographically consecutive (Fig.~\ref{Fig_type_prefix}-c), then the algorithm computes the number of sequences in the type that have prefix $i_{\mathrm{start}}^k$ and the number of sequences in the type that have prefix $i_{\mathrm{end}}^k$ using the last $N(t)-k$ symbols of the starting and the ending sequences; the decoder outputs the prefix that is shared by more source sequences.

\begin{figure}[h!]
\centering
\includegraphics[trim = 25mm 229mm 25mm 20mm, clip, width=9.5cm]{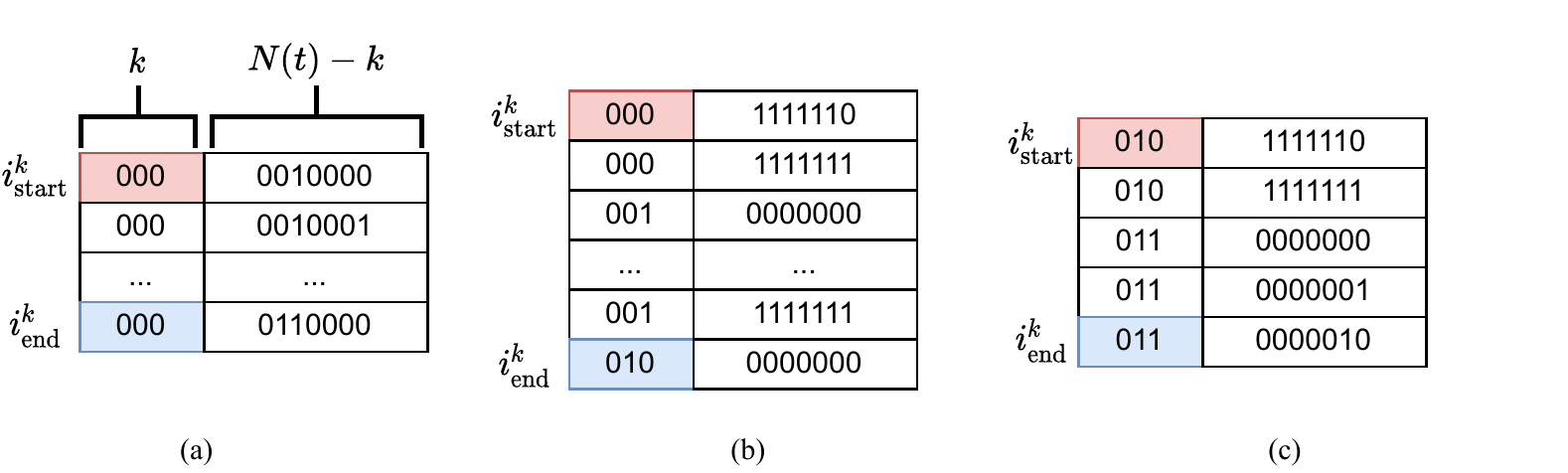}
\caption{Tables $(\mathrm{a})$, $(\mathrm{b})$, $(\mathrm{c})$ represent three types at time $t$. Each row represents a source sequence in the type. The first row and the last row in each type represent the starting sequence and the ending sequence in that type, respectively. The first column represents the length-$k$ prefix of sequences in the type. The source sequences in a type are lexicographically consecutive due to the methods of updating and splitting a type in steps~(i$^\prime$)~and~(iii$^\prime$). In $(\mathrm{a})$, since $i_{\mathrm{start}}^k = i_{\mathrm{end}}^k = 000$, the most probable sequence is $000$. In $(\mathrm{b})$, since $i_{\mathrm{start}}^k=000$ and $i_{\mathrm{end}}^k = 010$ are not lexicographically consecutive, the most probable prefix is $001$. In $(\mathrm{c})$, since $i_{\mathrm{start}}^k=010$ and $i_{\mathrm{end}}^k = 011$ are lexicographically consecutive, the number of sequences with prefix $i_{\mathrm{start}}^k$ can be computed by subtracting $1111110$, the last $N(t)-k$ symbols of the starting sequence, from $1111111$ and adding $1$; the number of sequences with prefix $i_{\mathrm{end}}^k$ is equal to the last $N(t)-k$ symbols of the ending sequence plus $1$. Since $(\mathrm{c})$ contains more sequences with prefix $011$, this is the most probable prefix.}
\label{Fig_type_prefix}
\end{figure}

We proceed to show that the complexity of the type-based anytime instantaneous SED code is $O(t\log t)$. Similar to the type-based instantaneous encoding phase in Section~\ref{type-based_encoding_phase}, the number of types grows linearly with time $t$ since the number of types increases only if a type is split in step~(iii$^\prime$), and at most $1$ type is split at each time $t$.  The complexities of steps~(i$^\prime$), (ii$^\prime$), (v$^\prime$) are all linear in the number of types $O(t)$ due to the discussion at the end of Section~\ref{type-based_encoding_phase}. The complexity of step~(iii$^\prime$) is log-linear in the number of types $O(t\log t)$ due to sorting the types. Since the sequences in a type are lexicographically consecutive due to the updating and the splitting methods in steps~(i$^\prime$)~and~(iii$^\prime$), it suffices to use the starting and the ending sequences in a type to determine the most probable prefix in that type. Thus, the complexity of step~(vi$^\prime$) is linear in the number of types due to searching for the type whose sequences have the largest posterior.

Restricting the type-based anytime instantaneous SED code described above to transmit only the first $k$ symbols of a DSS is equivalent to implementing steps (i), (ii), (iii$^\prime$), (v) one by one, and performing decoding as follows.

(vi$^{\prime\prime}$) \emph{Decoding and stopping:} If there exists a type $\mathcal S_j$ that satisfies $\rho_{\mathcal S_j}(y^t)\geq 1-\epsilon$ and contains a source sequence of length $k$, then the decoder stops and outputs a sequence in that type as the estimate $\hat S^k_{\eta_k}$. 

The complexity of the type-based instantaneous SED code for transmitting $k$ symbols remains log-linear, $O(t\log t)$, since the complexity of step (vi$^{\prime\prime}$) is $O(t)$ due to searching for the type that satisfies the requirements. 

While the type-based instantaneous encoding phase in Section~\ref{type-based_encoding_phase} is the exact algorithm of the instantaneous encoding phase in Section~\ref{belief_phase}, the type-based anytime instantaneous SED code and the type-based instantaneous SED code for transmitting $k$ symbols are \emph{approximations} of the original algorithms in Sections~\ref{Sec_alg_anytime}~and~\ref{Sec_SED_ER} due to two reasons below:

First, in step (iii$^\prime$) (group partitioning), we use the approximating instantaneous SED rule to mimic the exact rule in \eqref{pXinequal}--\eqref{Gxinequal}. The minimum of the objective function in \eqref{nstar} is equal to the difference $|\pi_0(y^{t-1})-\pi_1(y^{t-1})|$ between the group priors of the partition $\{\mathcal G_x(y^{t-1})\}_{x\in\{0,1\}}$ obtained by the approximating rule in step (iii$^\prime$). The difference is upper bounded as (Appendix~\ref{approximate_SED})
\begin{align}\label{pi0pi1_2}
    |\pi_0(y^{t-1}) - \pi_1(y^{t-1})| \leq \theta_{\mathcal S_j}(y^{t-1}),
\end{align}
where $\mathcal S_j$ is the last type moved to $\mathcal G_0(y^{t-1})$ so that its group prior exceeds $0.5$ for the first time. If $\pi_0(y^{t-1}) \geq \pi_1(y^{t-1})$, \eqref{pi0pi1_2} recovers \eqref{Gxinequal} since $\theta_{\mathcal S_j}(y^{t-1})$ is the smallest prior in $\mathcal G_0(y^{t-1})$, thus the approximating instantaneous SED rule recovers the exact rule. If $\pi_0(y^{t-1}) <\pi_1(y^{t-1})$, $\theta_{\mathcal S_j}(y^{t-1})$ on the right side of \eqref{pi0pi1_2} is the largest prior in $\mathcal G_1(y^{t-1})$, violating the right side of \eqref{Gxinequal}.

We use the approximating algorithm of the instantaneous SED rule \eqref{pXinequal}--\eqref{Gxinequal} since it is unclear how to implement the exact instantaneous SED rule with polynomial complexity. In the worst case, the complexity of the latter is as high as double exponential $O\left(2^{q^{N(t)}}\right)$ due to solving a minimization problem via an exhaustive search \cite[Algorithm~1]{Naghshvar2}. An exact algorithm for the SED rule with exponential complexity in the source length is given by \cite[Algorithm~2]{Naghshvar2}.

Second, in step (vi$^{\prime}$) (decoding at time $t$) of the type-based anytime instantaneous SED code, we only find the most likely length-$k$ prefix in the type that achieves $\max_j \rho_{\mathcal S_j}(y^t)$, yet it is possible that this prefix is not the one that has the globally largest posterior \eqref{hatSNT}. To search for the most probable length-$k$ prefix, one needs to compute the posteriors for all $q^k$ prefixes of length $k$ using $O(t)$ types, resulting in an exponential complexity $O(q^kt)$ in the length of the prefix $k$, whereas the complexity of step (vi$^{\prime}$) is only $O(t)$ independent of $k$.

Although the type-based instantaneous SED code is an approximation, as we are about to see in Fig.~\ref{Fig_type_based} Section~\ref{Sec_simulation}, it is almost as good as the exact code.

\section{Simulations}\label{Sec_simulation}

Fig.~\ref{Fig_R} shows the performance of our new instantaneous encoding schemes. Namely, we fix an error probability $\epsilon=10^{-6}$, a BSC($0.05$), and a DSS that emits i.i.d. Bernoulli$\left(\frac{1}{2}\right)$ bits one by one at consecutive times. We display the rate $R_k \triangleq \frac{k}{\mathbb E[\eta_k]}$ as a function of source length $k$ empirically attained by the instantaneous encoding phase followed by the SED code \cite[Algorithm~2]{Naghshvar2} and the instantaneous SED code in Section~\ref{Sec_SED_ER}, and we compare achievable rates to that of the SED code for a fully accessible source, as well as to that of a buffer-then-transmit code that implements the SED code during the block encoding phase. We also plot the rate $R_k$ obtained from the reliability function approximation \eqref{reliabilityfunc2}:
\begin{align}\label{solve_R}
    E(R_k) \simeq \frac{R_k}{k}\log\frac{1}{\epsilon}.
\end{align} 
Due to the discussions in the proof sketch of Theorem~\ref{thm_2}, the instantaneous encoding phase followed either by the MaxEJS code or by the SED code achieves the JSCC reliability function for streaming \eqref{converse2}. For the simulations in Fig.~\ref{Fig_R}, we choose the SED code since it applies to a BSC and its complexity, exponential in the source length, is lower than the double-exponential complexity of the MaxEJS code. To obtain the empirical rate in Fig.~\ref{Fig_R}, at each source length $k$, we run the experiments for every code for $10^5$ trials, and we obtain the denominator $\mathbb E[\eta_k]$ of the empirical rate by averaging the stopping times in all the experiments.

We observe from Fig.~\ref{Fig_R} that the achievable rate of the instantaneous encoding phase followed by the SED code is significantly larger than that of the buffer-then-transmit code, and approaches that of the SED code as $k$ increases even though the SED encoder knows the entire source sequence before the transmission. The instantaneous SED code demonstrates an even better performance: it is essentially as good as the SED code.
The rate obtained from reliability function approximation \eqref{solve_R} is remarkably close to the empirical achievable rates of our codes with instantaneous encoding even for very short source length $k\simeq 16$. For example, at $k=16$, the rate obtained from approximation \eqref{solve_R} is $0.58$ (symbols per channel use) and the empirical rate of the instantaneous SED code is $0.59$ (symbols per channel use). This means that the reliability function \eqref{reliabilityfunc2}, an inherently asymptotic notion, accurately reflects the delay-reliability tradeoffs attained by the JSCC reliability function-achieving codes in the ultra-short blocklength regime. The achievable rate corresponding to the buffer-then-transmit code is limited by \eqref{NC_reliability}.
\begin{figure}[h!]
\centering
\includegraphics[trim = 0mm 0mm 0mm 0mm, clip, width=8.5cm]{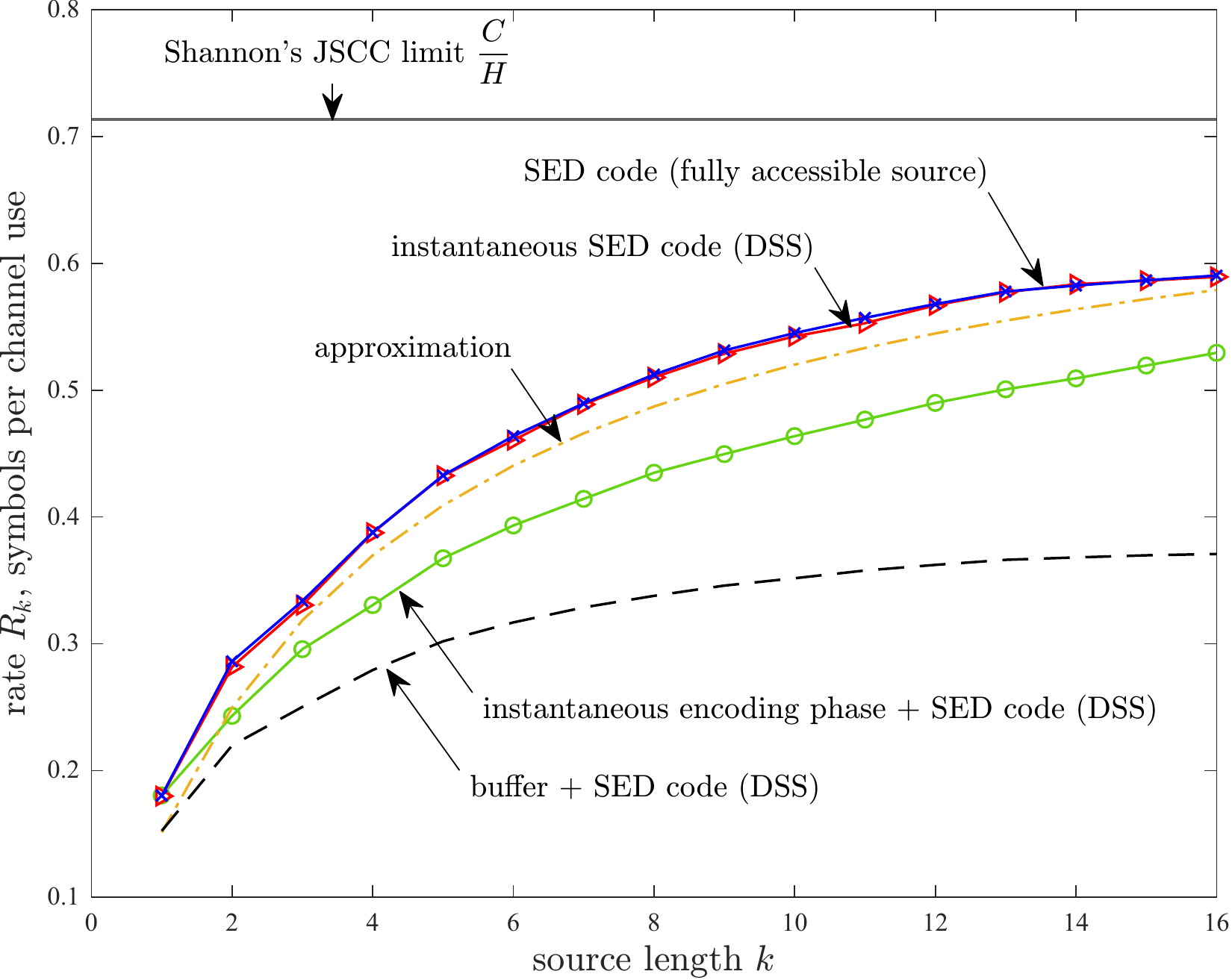}
\caption{Rate $R_k$ (symbols per channel use) vs. source length $k$. The error probability is constrained by $\epsilon = 10^{-6}$ \eqref{error_constraint}. The DMC is a BSC($0.05$). Naghshvar et al.'s SED code \cite[Algorithm~2]{Naghshvar2} operates on a fully accessible block $S^k$ of independent Bernoulli$\left(\frac{1}{2}\right)$ bits. The instantaneous encoding phase followed by the SED code, the instantaneous SED code, and the buffer-then-transmit code operate on $k$ i.i.d. Bernoulli$\left(\frac{1}{2}\right)$ source bits emitted one by one at times $t=1,2,\dots,k$.
The curves are displayed for the range of $k$'s where the complexities of the SED code and the instantaneous SED code are not prohibitive.}
\label{Fig_R}
\end{figure}

Fig.~\ref{Fig_type_based} shows the performance of the type-based instantaneous SED code. We fix an error probability $\epsilon=10^{-6}$ \eqref{error_constraint}, a BSC($p$) with $p=0.05, 0.03, 0.01$, and a DSS that emits i.i.d. Bernoulli$\left(\frac{1}{2}\right)$ bits one by one at consecutive times. We plot rate $R_k=\frac{k}{\mathbb E[\eta_k]}$ as a function of source length $k$ empirically achieved by the instantaneous SED code in Section~\ref{Sec_SED_ER} and its corresponding type-based code in Section~\ref{Type_based_SED}, as well as the rate obtained from the reliability function approximation \eqref{solve_R}. At each source length $k$, we run the experiments using the same method as in Fig.~\ref{Fig_R}. The rate gap between the instantaneous SED code and the type-based instantaneous SED code is negligible, meaning that the type-based instantaneous SED code with only log-linear complexity is a good approximation to the exact code in Section~\ref{Sec_SED_ER}. Furthermore, it is interesting to see that even though the DSS has symbol arriving rate $f=1$ symbol per channel use, which is far less than that required in assumption $(\mathrm{b}')$, the achievable rates of the instantaneous SED code stay very close to the rates obtained from the reliability function approximation. This suggests that assumption $(\mathrm{b}')$ on the symbol arriving rate, sufficient for the instantaneous SED code to achieve $E(R)$, could be conservative.
\begin{figure}[h!]
\centering
\includegraphics[trim = 0mm 0mm 0mm 0mm, clip, width=8.5cm]{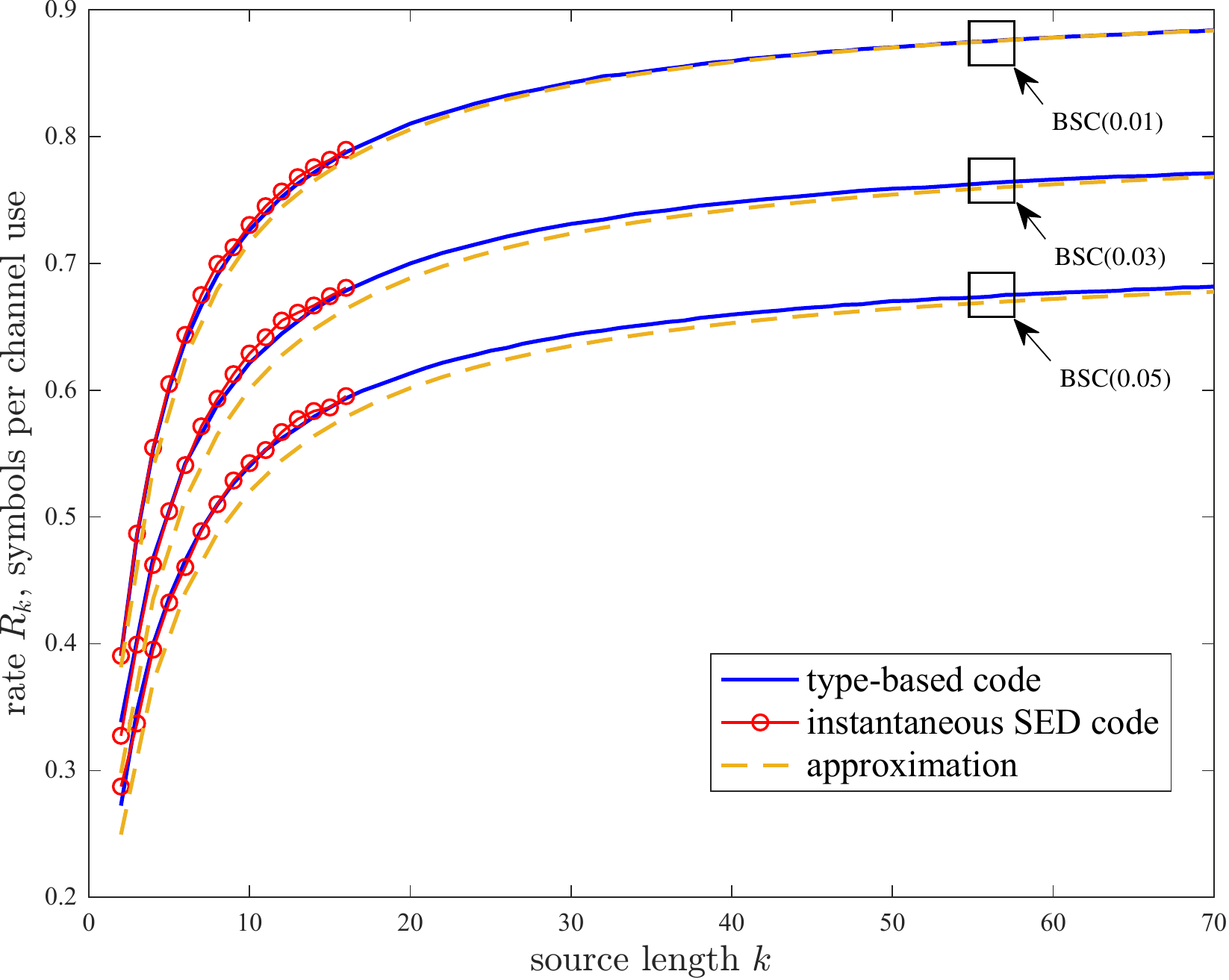}
\caption{Rate $R_k$ (symbols per channel use) vs. source length $k$. The error probability is constrained by $\epsilon = 10^{-6}$ \eqref{error_constraint}. The type-base instantaneous SED code in Section~\ref{Type_based_SED} and the instantaneous SED code in Section~\ref{Sec_SED_ER} operate on $k$ i.i.d. Bernoulli$\left(\frac{1}{2}\right)$ source bits emitted one by one at times $t=1,2,\dots,k$.}
\label{Fig_type_based}
\end{figure}

\section{Streaming over a degenerate DMC with zero error}\label{Sec_degenerate}
In this section, we propose a code with instantaneous encoding for a degenerate DMC \eqref{degenerate} that achieves zero decoding error at any rate asymptotically below $\frac{C}{H}$. Here, our code does not exactly follow Definition~\ref{def2} since it generalizes the code with instantaneous encoding in Definition~\ref{def2} by allowing \emph{common randomness} $U\in\mathcal U$, which is a random variable that is revealed to the encoder and the decoder before the transmission. With common randomness $U$, the encoder $\mathsf f_t$ \eqref{enc} can use $U$ to form $X_t$, and the decoder $\mathsf g_t$ \eqref{Skt_g} can use $U$ to decide the stopping time $\eta_k$ and the estimate $\hat S^k_{\eta_k}$. We refer to such a code as a $\langle k, R, \epsilon \rangle$ \emph{code with instantaneous encoding and common randomness} if it achieves rate $R$ \eqref{time_constraint} and error probability $\epsilon$ \eqref{error_constraint} for transmitting $k$ symbols of a DSS. Common randomness is widely used to specify a random codebook in the scenario where multiple constraints on expectations of quantities that depend on the codebook must be satisfied simultaneously and where Shannon's probabilistic method is not sufficient to claim the existence of a deterministic codebook satisfying all constraints, e.g., \cite{Yamamoto}\cite{PPV}\cite{Truong}\cite{V}. Since for a fixed $k$, we seek to satisfy two constraints, on the rate and on the error probability, the cardinality of $\mathcal U$ can be restricted as $|\mathcal U|\leq 2$ (Appendix~\ref{Cara}). 

Theorem~\ref{thm_ze}, stated next, establishes the existence of zero-error codes for the transmission over a degenerate DMC at any rate asymptotically below $\frac{C}{H}$.
\begin{theorem}\label{thm_ze}
Fix a degenerate DMC with capacity $C$ \eqref{capacity}, fix a $(q,\{t_n\}_{n=1}^{\infty})$ DSS with entropy entropy rate $H>0$ \eqref{entropy_rate} satisfying assumptions $(\mathrm{a})$--$(\mathrm{b})$ in Theorem~\ref{thm_2}, and fix any $R<\frac{C}{H}$. There exists a sequence of $\langle k,R_k, 0\rangle $ codes with instantaneous encoding and common randomness that satisfies
\begin{align}\label{RkR}
    \lim_{k\rightarrow\infty} R_k = R.
\end{align}
\end{theorem}
\begin{proof}[Proof sketch]
Our zero-error code for degenerate DMCs extends Burnashev's scheme \cite[Sec. 6]{Burnashev} to JSCC and to streaming sources: to achieve Shannon's JSCC limit $\frac{C}{H}$, a Shannon limit-achieving code is used in the first communication phase to compress the source; to transmit streaming sources, we combine an instantaneous encoding phase that satisfies \eqref{pre_coding_condition} with a Shannon limit-achieving block encoding scheme to form a Shannon limit-achieving instantaneous encoding scheme. To achieve zero error, we employ confirmation phases similar to those in Burnashev's scheme \cite{Burnashev}. We say that a $\langle k, R, \epsilon_k \rangle$ code with instantaneous encoding and common randomness achieves Shannon's JSCC limit $\frac{C}{H}$ if for all $R<\frac{C}{H}$, a sequence of such codes indexed by $k$ satisfies $\epsilon_k\rightarrow 0$ as $k\rightarrow \infty$. Our zero-error code includes such Shannon limit-achieving codes as a building block. Note that in contrast to the discussions in Sections~\ref{Sec_reliability2}--\ref{Sec_instantaneous_SED} focused on the exponential rate of decay of $\epsilon_k$ to $0$ \eqref{reliabilityfunc2} over non-degenerate DMCs, here merely having $\epsilon_k$ decrease to $0$ suffices. The following argument shows the existence of such codes for the class of channels that includes both non-degenerate and degenerate DMCs. 

We employ the joint source-channel code in \cite[Theorem~2]{V} due to the simplicity of the error analysis it affords. The code in \cite[Theorem~2]{V} is a $\langle k,R,\epsilon_k\rangle$ Shannon limit-achieving code with block encoding and common randomness because its expected decoding time to attain error probability $\epsilon$ is upper bounded as \eqref{lambda_k_ub0} in Appendix~\ref{pf_achieve_A} with $C_1\leftarrow C$ \cite[Eq. (16)]{V}, implying that it achieves a positive error exponent that is equal to \eqref{converse2} with $C_1\leftarrow C$ for all $R<\frac{C}{H}$. The block encoding scheme in \cite[Theorem~2]{V} is a stop-feedback code, meaning that the encoder uses channel feedback only to decide whether to stop the transmission but not to form channel inputs. If the DSS has an infinite symbol arriving rate $f=\infty$ \eqref{freq}, a buffer-then-transmit code using the block encoding scheme in \cite[Theorem~2]{V} achieves the Shannon limit since it achieves the same error exponent as the code in \cite[Theorem~2]{V}. To see this, one can simply invoke \eqref{lambda_k_ub0} in Lemma~\ref{lemma_s1} with $C_1\leftarrow C$ and follow the proofs in Appendix~\ref{pf_achieve_B}. By the same token, if the DSS has a finite symbol arriving rate $f<\infty$ \eqref{freq}, a code implementing an instantaneous encoding phase that satisfies \eqref{pre_coding_condition} followed by the block encoding scheme in \cite[Theorem~2]{V} for $k$ source symbols with prior $P_{S^k|Y^{t_k}}$ achieves the Shannon limit with the same error exponent as the code in \cite[Theorem~2]{V}.

Our zero-error code with instantaneous encoding and common randomness for transmitting $k$ symbols over a degenerate DMC operates as follows (details in Appendix~\ref{zero-error_A}). Similar to \cite{Burnashev}--\cite{Caire}, \cite{Truong}, our code is divided into blocks. Each block contains a communication phase and a confirmation phase. In the first block, the communication phase uses a $\langle k,R,\epsilon_k\rangle$ Shannon limit-achieving code with instantaneous encoding and common randomness. The confirmation phase selects two symbols $x$ \eqref{degenerate_a} and $x'$ \eqref{degenerate_b} as the channel inputs (i.e., $x'$ never leads to channel output $y$); the encoder repeatedly transmits $x$ if the decoder's estimate of the source sequence at the end of the communication phase is correct, and transmits $x'$ otherwise. If the decoder receives a $y$ in the confirmation phase, meaning that the encoder communicated its knowledge that the decoder's estimate is correct with zero error, then it outputs its estimate, otherwise, the next block is transmitted. The $\ell$-th block, $\ell\geq 2$, differs from the first block in that it does not compress the source to avoid errors due to an atypical source realization and in that it uses random coding whereas the first block can employ any Shannon-limit achieving code.

We proceed to discuss the error and the rate achievable by our code (details in Appendix~\ref{zero-error_B}).

Our code achieves zero error by employing confirmation phases that rely on the degenerate nature of the channel: receiving a $y$ in the confirmation phase guarantees a correct estimate. 

Our code achieves all rates asymptotically below $\frac{C}{H}$ because 1) the first block employs a Shannon limit-achieving code in the communication phase, 2) the length of the confirmation phase is made negligible compared to the length of the communication phase as the source length $k\rightarrow\infty$, meaning that the length of the first block asymptotically equals the length of its communication phase, and 3) subsequent blocks asymptotically do not incur a penalty on rate, as we discuss next. Since the length of each block is comparable to the length of the first block, it is enough to show that the expected number of blocks $T_k$ transmitted after the first block converges to zero. 
The refreshing of random codebook for all uncompressed source sequences in every block after the first block ensures that the channel output vectors in these subsequent blocks are i.i.d. and are independent of the channel outputs in the first block. Conditioned on $T_k>0$, the i.i.d. vectors give rise to a geometric distribution of $T_k$ with failure probability converging to $0$, which implies $\mathbb E[T_k] \rightarrow 0$ as $k\rightarrow\infty$. 
\end{proof}

A stop-feedback code with block encoding that retransmits blocks with the overall rate asymptotically equal to the rate of the first block is also used by Forney
\cite[p. 213]{Forney} for deriving a lower bound on the reliability function of a DMC.

\section{Conclusion}
In this paper, we have derived the reliability function for transmitting a discrete streaming source over a DMC with feedback using variable-length joint source-channel coding with instantaneous encoding under regularity conditions (Theorem~\ref{thm_2}). Since a classical fully accessible DS is a special DSS (see \eqref{DS}), Theorem~\ref{thm_2} extends Burnashev's reliability function to the classical JSCC scenario with block encoding, as well as to a streaming scenario. The most surprising observation in this paper is that the JSCC reliability function for a streaming source is equal to that for a fully accessible source. A naive buffer-then-transmit code that idles the transmission during the symbol arriving period does not achieve the JSCC reliability function for a non-trivial streaming source (see \eqref{NC_reliability}). To achieve the JSCC reliability function for such sources, we have proposed a novel instantaneous encoding phase (Section~\ref{belief_phase}). We have shown that preceding a JSCC reliability function-achieving code with block encoding, e.g., the MaxEJS code or the SED code \cite{Naghshvar2}, by our instantaneous encoding phase (Section~\ref{belief_phase})  will make it overcome the detrimental effect due to the streaming nature of the source and make it achieve the same error exponent as if the encoder knew the entire source sequence before the transmission. The instantaneous encoding phase (Section~\ref{belief_phase}) achieves the JSCC reliability function because it satisfies the sufficient condition \eqref{pre_coding_condition} on the statistics of the encoder outputs during the symbol arriving period, for example, the instantaneous encoding phase continues to achieve the sufficient condition \eqref{pre_coding_condition} after it drops the randomization step, but at a cost of increasing the threshold for the symbol arriving rate (Remark~\ref{rmk_drop_random}). While our JSCC reliability function-achieving codes are designed to transmit $k$ symbols of a streaming source and stop, we have also designed an instantaneous SED code (Section~\ref{Sec_instantaneous_SED}) that can choose the decoding time and the number of symbols to decode on the fly. It empirically attains a positive anytime reliability (Fig.~\ref{Fig_infinite}), thus it can be used to stabilize an unstable scalar linear system with a bounded noise over a noisy channel. A sequence of such codes indexed by the source length to decode also achieves the JSCC reliability function for streaming in the limit of large source length (Theorem~\ref{thm_ze}). For practical implementations, we have designed type-based log-linear complexity algorithms for the instantaneous encoding phase and the instantaneous SED code (Section~\ref{Sec_practical}) that apply to streaming sources with equiprobable symbols. While the codes that achieve the JSCC reliability function are designed for non-degenerate DMCs, we have also designed zero-error codes with instantaneous encoding for degenerate DMCs (Section~\ref{Sec_degenerate}), extending Burnashev's zero-error channel code to the JSCC and to the streaming scenarios. 

Future research directions include the following. First, it would be interesting to find the JSCC reliability function for a wider class of channels and streaming sources (e.g., sources without a valid $f$, with a small $f$, or with an infinite alphabet). Second, it would be interesting to extend $E(R)$ to lossy JSCC by inserting an appropriate instantaneous encoding phase before a lossy JSCC reliability function-achieving block encoding scheme. Third, it is practically important to design instantaneous encoding schemes using limited or noisy feedback. Finally, it would be interesting to find codes that can learn the streaming source distribution on the fly.


\section*{Acknowledgement}
Insightful comments from Dr. Oron Sabag are gratefully acknowledged.

\appendices

\section{A partition that satisfies \eqref{SD}}\label{SD_exists}
For any $t=1,2,\dots$ and any $y^{t-1}\in\mathcal Y^{t-1}$, we show that the greedy heuristic algorithm \cite{Korf} yields a partition  $\{\mathcal G_x(y^{t-1})\}_{x\in\mathcal X}$ that satisfies the partitioning rule \eqref{SD}. 

The greedy heuristic algorithm operates as follows. At time $t$, it initializes all the groups $\{\mathcal G_x(y^{t-1})\}_{x\in\mathcal X}$ by empty sets and initializes all the group priors $\{\pi_x(y^{t-1})\}_{x\in\mathcal X}$ by zeros. It sorts all the source sequences in $[q]^{N(t)}$ according to their priors $\theta_i(y^{t-1})$, $i\in[q]^{N(t)}$ in a descending manner. Starting from the sequence with the largest prior, it moves the sequence in the sorted list to the group $\mathcal G_{x^*}(y^{t-1})$ whose current group prior has the largest gap to the corresponding capacity-achieving probability, i.e.,
\begin{align}\label{max_problem}
    x^*\triangleq \arg\max_{x\in\mathcal X}P_{X}^*(x) - \pi_x(y^{t-1}).
\end{align}
The group prior $\pi_{x^*}(y^{t-1})$ is updated after each move.
The partitioning process repeats until all the source sequences have been classified.

We show that the resulting partition $\{\mathcal G_x(y^{t-1})\}_{x\in\mathcal X}$ satisfies \eqref{SD}. We first notice that the maximization problem on the right side of \eqref{max_problem} must be strictly larger than zero before all source sequences have been classified. If moving sequence $i$ to group $\mathcal G_{x^*}(y^{t-1})$ leads to 
\begin{align}\label{pi<PX}
    \pi_{x^*}(y^{t-1}) - P_{X}^*(x^*)<0,
\end{align}
then \eqref{SD} is obviously satisfied. If moving sequence $i$ to group $\mathcal G_{x^*}(y^{t-1})$ leads to
\begin{align}\label{pi_x_P_X}
    \pi_{x^*}(y^{t-1}) - P_{X}^*(x^*)\geq 0,
\end{align}
then \eqref{SD} is satisfied since \eqref{pi<PX} holds before the move, and sequence $i$ has the smallest prior in $\mathcal G_{x^*}(y^{t-1})$ after the move. Furthermore, if the group $\mathcal G_{x^*}(y^{t-1})$ satisfies \eqref{pi_x_P_X}, it will no longer be the solution to the maximization problem in \eqref{max_problem} and thus will no longer accept new sequences. This means that \eqref{SD} holds for all $x\in\mathcal X$ at the end of the greedy heuristic partitioning.

\section{An algorithm to determine $\{p_{\overline x\rightarrow \underline{x}}\}_{\overline x\in \overline{\mathcal X}(y^{t-1}), \underline{x}\in\underline{\mathcal X}(y^{t-1})}$}\label{set_pxx}
We design Algorithm~\ref{algo_pxx} to determine a set of probabilities $\{p_{\overline x\rightarrow \underline{x}}\}_{\overline x\in \overline{\mathcal X}(y^{t-1}), \underline{x}\in\underline{\mathcal X}(y^{t-1})}$ that satisfies \eqref{pxy}--\eqref{pxy2}. We denote by $\bar{\mathcal X}(1)$ the first element in set $\bar{\mathcal X} (y^{t-1})$. The order of the elements in $\bar{\mathcal X} (y^{t-1})$ is irrelevant.

For every group $\mathcal G_{\underline{x}}(y^{t-1})$ with $\underline{x}\in\underline{\mathcal{X}}(y^{t-1})$, Algorithm~\ref{algo_pxx} goes through groups $\mathcal G_{\bar{x}}(y^{t-1})$ with $\bar x \in \bar{\mathcal X}(y^{t-1})$ to transfer probability $p_{\bar x\rightarrow\underline x}$ to $\mathcal G_{\underline{x}}(y^{t-1})$. The amount of probability $p_{\bar x\rightarrow\underline x}$ to transfer from $\mathcal G_{\bar x}(y^{t-1})$ to $\mathcal G_{\underline x}(y^{t-1})$ is the smallest of $\hat \pi_{\bar x}(y^{t-1})-P_X^*(\bar x)$ and $P_X^*(\underline x) - \hat \pi_{\underline x}(y^{t-1})$. After the update, if the new prior $\hat \pi_{\bar x}(y^{t-1})$ (or $\hat \pi_{\underline x}(y^{t-1})$) is equal to its target value $P_{X}^*(\bar x)$ (or $P_{X}^*(\underline x)$), the corresponding group will be removed from the set $\bar {\mathcal X}(y^{t-1})$ (or $\underline {\mathcal X}(y^{t-1})$). In this way, $\{p_{\overline x\rightarrow \underline{x}}\}_{\overline x\in \overline{\mathcal X}(y^{t-1}), \underline{x}\in\underline{\mathcal X}(y^{t-1})}$ are determined. At the end of the algorithm, $\hat \pi_{\bar x}(y^{t-1})$ and $\hat \pi_{\underline x}(y^{t-1})$ indeed represent the left sides of \eqref{pxy} and \eqref{pxy2}, respectively. 

We show that \eqref{pxy}--\eqref{pxy2} hold: \eqref{pxy2} holds by lines 4, 6, 8, 9--10 of Algorithm~\ref{algo_pxx} and the fact that line 9 must be satisfied during the while loop since
\begin{align}\label{sumsum_eq}
    \sum_{\underline{x}\in\underline{\mathcal X}(y^{t-1})}P_X^*(\underline x) - \pi_{\underline x}(y^{t-1}) = \sum_{\bar{x}\in\bar{\mathcal X}(y^{t-1})}\pi_{\bar x}(y^{t-1}) - P_X^*(\bar x)
\end{align}
ensures that there are enough probabilities $p_{\bar x\rightarrow\underline x}$ to transfer from groups in $\bar{\mathcal X}(y^{t-1})$ to $\mathcal G_{\underline x}(y^{t-1})$;
\eqref{pxy} holds by lines 6, 7, 11--12 of Algorithm~\ref{algo_pxx} and the facts that 1) \eqref{pxy2} holds, i.e., $\hat \pi_{\underline x}(y^{t-1}) = P_X^*(\underline x)$, 2) the minimum in line 6 ensures that $\hat \pi_{\bar x}(y^{t-1})\geq P_X^*(\bar x)$ at the end of the algorithm, 3) \eqref{sumsum_eq} implies that $\hat \pi_{\bar x}(y^{t-1})>P_X^*(\bar x)$ is impossible, otherwise $\sum_{x\in\mathcal X}\hat \pi_x(y^{t-1})>\sum_{x\in\mathcal X}P_X^*(x)=1$. 

\begin{algorithm}[t]
\caption{Determine $\{p_{\overline x\rightarrow \underline{x}}\}_{\overline x\in \overline{\mathcal X}(y^{t-1}), \underline{x}\in\underline{\mathcal X}(y^{t-1})}$ that satisfies \eqref{pxy}--\eqref{pxy2} }\label{algo_pxx}
\KwData{$\pi_x(y^{t-1})_{x\in\mathcal X}$,~$\bar {\mathcal X}(y^{t-1})$, ~$\underline {\mathcal X}(y^{t-1})$}
\KwResult{$\{p_{\overline x\rightarrow \underline{x}}\}_{\overline x\in \overline{\mathcal X}(y^{t-1}), \underline{x}\in\underline{\mathcal X}(y^{t-1})}$}
{$\{p_{\overline x\rightarrow \underline{x}}\}_{\overline x\in \overline{\mathcal X}(y^{t-1}), \underline{x}\in\underline{\mathcal X}(y^{t-1})} = 0;$}\\
{$\hat \pi_x(y^{t-1})\gets \pi_x(y^{t-1}),\forall x\in\mathcal X;$}\\
\For{$\underline{x}\in \underline{\mathcal X}(y^{t-1})$}{
\While{$\hat \pi_{\underline{x}}(y^{t-1})<P_{X}^*(\underline{x})$}{
$\bar x \gets \bar {\mathcal X}(1);$\\
$p_{\bar x\rightarrow\underline x} \gets \min\{\hat \pi_{\bar x}(y^{t-1})-P_X^*(\bar x), P_X^*(\underline x) - \hat \pi_{\underline x}(y^{t-1})\};$\\
$\hat \pi_{\bar x}(y^{t-1}) \gets \hat \pi_{\bar x}(y^{t-1}) - p_{\bar x\rightarrow\underline x};$\\
$\hat \pi_{\underline x}(y^{t-1}) \gets \hat \pi_{\underline x}(y^{t-1}) + p_{\underline x\rightarrow\underline x};$\\
\If{$\hat \pi_{\underline x}(y^{t-1}) = P_{X}^*(\underline x)$}{$\underline{\mathcal X}(y^{t-1})\gets \underline{\mathcal X}(y^{t-1})\setminus \underline x;$}
\If{$\hat \pi_{\bar x}(y^{t-1}) = P_{X}^*(\bar x)$}{$\bar{\mathcal X}(y^{t-1})\gets \bar{\mathcal X}(y^{t-1})\setminus \bar x;$}
}}
\end{algorithm}

\section{Channel input distribution is equal to the capacity-achieving distribution }\label{pf_PXYPX*}
We show that \eqref{PXYPX*} holds, i.e., the channel input distribution is equal to the capacity-achieving distribution. For any $x\in\mathcal X$ and $y^{t-1}\in\mathcal Y^{t-1}$, we expand right side of \eqref{PXYPX*} as
\begin{subequations}\label{PXYexpand}
\begin{align}\nonumber
    & P_{X_t|Y^{t-1}}(x|y^{t-1})\\ \label{expand_a}
    =& \sum_{z\in\mathcal X}P_{X_t|Z_t, Y^{t-1}}(x|z,y^{t-1})P_{Z_t|Y^{t-1}}(z|y^{t-1})\\ \label{expand_b}
    = & \sum_{z\in\mathcal X}P_{X_t|Z_t, Y^{t-1}}(x|z,y^{t-1})\pi_{z}(y^{t-1})\\\label{expand_c}
    = & P_{X_t|Z_t, Y^{t-1}}(x|x,y^{t-1})\pi_{x}(y^{t-1})\\ \label{expand_d}
    + & \sum_{z\neq x}P_{X_t|Z_t, Y^{t-1}}(x|z,y^{t-1})\pi_{z}(y^{t-1}),
\end{align}
\end{subequations}
where \eqref{expand_a} holds by the law of total probability and \eqref{expand_b} holds by the definition of $Z_t$ in \eqref{Zt1}.

By the randomization distribution in \eqref{PXZY}, if $x\in \overline {\mathcal X}(y^{t-1})$, then \eqref{expand_c} is equal to $P_{X}^*(x)$ and \eqref{expand_d} is equal to $0$, and if $x\in \underline {\mathcal X}(y^{t-1})$, then \eqref{expand_c} is equal to $\pi_{x}(y^{t-1})$
and \eqref{expand_d} is equal to 
\begin{align}\label{xunder2}
    \sum_{z\in \overline {\mathcal X}(y^{t-1})}\frac{p_{z\rightarrow x}}{\pi_{x}(y^{t-1})}\pi_{x}(y^{t-1}) = P_{X}^*(x) - \pi_{x}(y^{t-1}),
\end{align}
where \eqref{xunder2} uses \eqref{pxy2}.

\section{Converse proof of Theorem~\ref{thm_2}}\label{pf_converse}
\subsection{Converse proof}
Inspired by Berlin et al.'s converse proof \cite{Berlin} for Burnashev's reliability function, we provide a converse bound on the JSCC reliability function for a fully accessible source by lower bounding the expected stopping time of an arbitrary code with block encoding using the error probability at the stopping time. The converse bound continues to apply for the JSCC reliability function for streaming, since given a DMC, every code with instantaneous encoding for transmitting the first $k$ symbols of a $(q,\{t_n\}_{n=1}^{\infty})$ DSS in Definition~\ref{def2} is a special code with block encoding for transmitting the first $k$ symbols $S^k\in[q]^k$ of a DS \eqref{DS}.

We consider $k$ symbols $S^k\in[q]^k$ of a DS with source distribution $P_{S^k}$, and we fix a non-degenerate DMC with a single-letter transition probability $P_{Y|X}\colon \mathcal X\rightarrow \mathcal Y$. We fix an arbitrary code with block encoding with a stopping time $\eta_k$ for transmitting $S^k$ over the non-degenerate DMC with feedback. We assume that the decoder is a MAP decoder \eqref{hatSNT}, since given any encoding function and any stopping time in Definition~\ref{def2}, the MAP decoder \eqref{hatSNT} achieves the minimum error probability \eqref{error_constraint}. For brevity, we denote the error probability of a MAP decoder given channel outputs $y^t\in\mathcal Y^t$ by
\begin{align}
    P_{e}(y^t)\triangleq 1-\max_{s\in[q]^k}P_{S^k|Y^t}(s|y^t),
\end{align}
and we denote the error probability of a MAP decoder at the stopping time $\eta_k$ by
\begin{align}\label{MAP_Pe}
    P_e \triangleq \mathbb E[P_e(Y^{\eta_k})].
\end{align}
We define stopping time $\tau_{\delta}$ as
\begin{align}\label{def_taudelta}
    \tau_{\delta}\triangleq \min\{t\colon P_e(y^t)\leq \delta~\text{or}~t=\eta_k\}.
\end{align}

To obtain the converse bound on the JSCC reliability function for a fully accessible source, we establish a lower bound on the expected decoding time $\mathbb E[\eta_k]$ using the error probability $P_e$ and and source distribution $P_{S^k}$. To this end, we lower bound $\mathbb E[\tau_{\delta}]$ and $\mathbb E[\eta_k - \tau_{\delta}]$, respectively. The lower bound on $\mathbb E[\tau_{\delta}]$ is stated below.
\begin{lemma}[Modified Lemma~2 in \cite{Berlin}]\label{Berlin1} Consider $k$ symbols $S^k\in [q]^k$ of a DS with source distribution $P_{S^k}$ \eqref{bit_prob} and fix a non-degenerate DMC with capacity $C$ \eqref{capacity}. For any $\delta\in\left(0,\frac{1}{2}\right]$, it holds that
\begin{align}\label{Etau}
    \mathbb E[\tau_{\delta}] \geq \frac{H(S^k)}{C}\left(1 - \left(\delta+\frac{P_e}{\delta}\right)\frac{\log q^k}{H(S^k)}\right) - \frac{h(\delta)}{C}.
\end{align}
\end{lemma}
\begin{proof}
Appendix~\ref{pf_Berlin1}.
\end{proof}
The lower bound on $\mathbb E[\eta_k - \tau_{\delta}]$ is stated below.
\begin{lemma}[Modified Eq. (17) \cite{Berlin}]\label{Berlin2} Consider $k$ symbols $S^k\in [q]^k$ of a DS with source distribution $P_{S^k}$ and fix a non-degenerate DMC with transition probability $P_{Y|X}\colon \mathcal X\rightarrow\mathcal Y$ and maximum KL divergence $C_1$ \eqref{C1}. For any $\delta\in\left(0,\frac{1}{2}\right]$, it holds that
\begin{align}\label{Berlin2_eq}
    &\mathbb E[\eta_k - \tau_{\delta}]\geq\\ \nonumber
    &\frac{\log \frac{1}{P_e} - \log 4 + \log(\min\left\{p_{\min}\delta, 1-\max_{s\in[q]^k}P_{S^k}(s)\right\})}{C_1},
\end{align}
where $p_{\min}$ in \eqref{Berlin2_eq} is the minimum channel transition probability \eqref{def_lambda}.
\end{lemma}
\begin{proof}
Appendix~\ref{pf_Berlin2}.
\end{proof}
Summing up the right sides of \eqref{Etau} and \eqref{Berlin2_eq}, we obtain the following lower bound on the expected decoding time $\eta_k$ of an arbitrary code with block encoding:
\begin{align}\nonumber
    &\mathbb E[\eta_k] \geq \frac{H(S^k)}{C}\left(1 - \left(\delta+\frac{P_e}{\delta}\right)\frac{\log q^k}{H(S^k)}\right) + \frac{\log\frac{1}{P_e}}{C_1} + \\\label{lb_etak}
    &\frac{- \log 4 + \log(\min\left\{p_{\min}\delta, 1-\max_{s\in[q]^k}P_{S^k}(s)\right\})}{C_1}-  \frac{h(\delta)}{C}.
\end{align}
The asymptotic performance of the lower bound \eqref{lb_etak} relies on two properties of the DS in Lemma~\ref{lemma_asyp}, stated next.
\begin{lemma}\label{lemma_asyp}
Consider a DS with a well-defined and positive entropy rate $H$ \eqref{entropy_rate} and a finite single-letter alphabet $[q]$. Then,
\begin{align}\label{infty_entropy}
&\lim_{k\rightarrow\infty}\frac{\log q^k}{H(S^k)} = \frac{\log q}{H}<\infty,\\\label{lim_Psk}
    &\liminf_{k\rightarrow\infty} \left(1-\max_{s\in[q]^k}P_{S^k}(s)\right) >0.
\end{align}
\end{lemma}
\begin{proof}
The proof of \eqref{lim_Psk} is in Appendix~\ref{pf_DS_prop}.
\end{proof}
Plugging \eqref{infty_entropy}--\eqref{lim_Psk} and $\delta = -\frac{1}{\log P_e}$ into the right side of \eqref{lb_etak}, we obtain
\begin{align}\label{lb_etak_final}
  \mathbb E[\eta_k] \geq \left( \frac{H(S^k)}{C}+  \frac{\log\frac{1}{P_e}}{C_1}\right)(1-o(1)),
\end{align}
where $o(1)$ in \eqref{lb_etak_final} is a positive term that converges to $0$ as both $P_e\rightarrow 0$ and $k\rightarrow\infty$. Rearranging terms of \eqref{lb_etak_final}, we conclude that $E(R)$ is upper bounded by the right side of \eqref{converse2}.
 Similar to \cite{Burnashev}, \cite[Eq. (5)]{Berlin}, \cite[Proposition~9]{Truong}, here we need not consider the case where $P_e$ does not converge to zero since this means $E(R)=0$.
 
\subsection{Proof of Lemma~\ref{Berlin1}}\label{pf_Berlin1}
We follow Berlin et al.'s notations \cite{Berlin}: we denote by $\mathcal H(S^k|Y^t)$ a random variable that satisfies $\mathcal H(S^k|y^t) = H(S^k|Y^t=y^t)$. Note that $\mathbb E[\mathcal H(S^k|Y^t)] = H(S^k|Y^t)$. If any step below has already been proved in \cite{Berlin}, we avoid repeated reasoning by referring to the proof in \cite{Berlin}. Compared to Berlin et al.'s proof in \cite[Sec. IV]{Berlin}, the proof below does not assume that the source is equiprobably distributed -- it keeps the generic form of the source prior $P_{S^k}$.

The sequence $\{\mathcal H(S^k|Y^t)+tC\}_{t=0,1,\dots}$ is a submartingale (\cite[Lemma 2]{Berlin}) with respect to the filtration generated by the channel outputs. Using Doob's optional stopping theorem \cite{Williams}, the initial state of the submartingale $\{\mathcal H(S^k|Y^t)+tC\}_{t=0,1,\dots}$ is upper bounded as
\begin{align}\label{HSK}
    H(S^k)  &\leq H(S^k|Y^{\tau_{\delta}}) + \mathbb E[\tau_{\delta}]C.
\end{align}
For $\delta\in\left(0,\frac{1}{2}\right]$, the conditional entropy on the right side of \eqref{HSK} is upper bounded as
\begin{align}\label{EH}
      H(S^k|Y^{\tau_{\delta}})\leq  h(\delta) + \left(\delta+\frac{P_e}{\delta}\right)\log q^k
\end{align}
using Fano's inequality (in the same manner as in \cite[Eq. (14)--(16)]{Berlin}). Plugging \eqref{EH} to the right side of \eqref{HSK} and rearranging terms, we obtain \eqref{Etau}.

\subsection{Proof of Lemma~\ref{Berlin2}}\label{pf_Berlin2}
We obtain the lower bound on $\mathbb E[\eta_k - \tau_{\delta}]$ in Lemma~\ref{Berlin2} by constructing a binary hypothesis test performed over a non-degenerate DMC with feedback.
We first state a lower bound on the error probability of such a test. Consider a binary hypothesis test $(H_0, H_1)$ performed over a DMC with feedback via a variable-length code with block encoding. The encoder sends a sequence of symbols $X_1,X_2,\dots$, over the given DMC with feedback, such that at the stopping time $T$, if $H_0$ is true, then the channel output vector $Y^T$ is distributed according to $Q_{H_0}$, otherwise, the channel output vector $Y^T$ is distributed according to $Q_{H_1}$. At the stopping time, the decoder uses the decoding function $\hat W\colon \mathcal Y^T\rightarrow \{H_0, H_1\}$ to form a decoded hypothesis. We denote the set of channel outputs $y^T$ that leads to decoded hypothesis $H_i$, $i\in\{0,1\}$ by
\begin{align}
    \mathcal Y_{H_i} \triangleq \{y^T\in\mathcal Y^T\colon \hat W(y^T) = H_i\}, i\in\{0,1\}.
\end{align}
We denote by $p_{H_i}$, $i\in\{0,1\}$, the prior probability of hypothesis $H_i$, $i\in\{0,1\}$ before the transmission. We denote the error probability of the binary hypothesis test at the stopping time $T$ by
\begin{align}
    P_b \triangleq p_{H_0}Q_{H_0}(\mathcal Y_{H_1}) + p_{H_1}Q_{H_1}(\mathcal Y_{H_0}).
\end{align}
\begin{lemma}[Lemma 1 in \cite{Berlin}]\label{lemma_binary}
Consider a binary hypothesis test with hypotheses $H_0$ and $H_1$ performed over a non-degenerate DMC with feedback that has maximum KL divergence $C_1$ \eqref{C1} via a variable-length code with block encoding. The error probability of the binary hypothesis test $P_b$ at stopping time $T$ is lower bounded as
\begin{align}\label{Lemma1_Berlin}
    P_b \geq \frac{\min\{p_{H_0}, p_{H_1}\}}{4}e^{-C_1 \mathbb E[T]},
\end{align}
where $p_{H_0}$ and $p_{H_1}$ are the prior probabilities of the hypotheses\footnote{Notice that in the proof of \cite[Lemma~1]{Berlin}, Berlin et al. invoked \cite[Proposition~1]{Berlin}, which relies on $\mathbb E[T]<\infty$. Yet, \eqref{Lemma1_Berlin} also trivially holds for $\mathbb E[T]=\infty$.}.
\end{lemma}

We employ the same hypothesis test as that in \cite[Section V]{Berlin}. Compared to Berlin et al.'s proof \cite[Sec. V]{Berlin}, the proof below lower bounds the priors of the hypotheses differently since we consider a generic source distribution whereas Berlin at al.'s \cite{Berlin} considered equiprobable source symbols.

The binary hypothesis test (cf. \cite[Section V]{Berlin}) starts at time $\tau_{\delta}+1$ and operates as follows. Given any $Y^{\tau_{\delta}}$, we partition the alphabet $[q]^k$ into two sets $\mathcal G(Y^{\tau_{\delta}})$ and $[q]^k\setminus\mathcal G(Y^{\tau_{\delta}})$ (we will specify $\mathcal G$ in the sequel). The two hypotheses are $H_0\colon S^k\in \mathcal G(Y^{\tau_{\delta}})$ and $H_1\colon S^k\in [q]^k\setminus\mathcal G(Y^{\tau_{\delta}})$. At the stopping time $\eta_k$, the MAP decoder outputs the estimate of the source $\hat S^k_{\eta_k}$ using the channel outputs $Y^{\eta_k}$. If the estimate satisfies $\hat S^k_{\eta_k} \in \mathcal G(Y^{\tau_{\delta}})$, then we declare $H_0$, otherwise, we declare $H_1$. The error probability of decoding $S^k$ is lower bounded by the error probability of the binary hypothesis test (\cite[the second paragraph below Prop. 2]{Berlin}), i.e., given any $t\geq 0$, $y^t\in\mathcal Y^t$,
\begin{equation}\label{two_error_prob}
\begin{aligned}
   & \mathbb P[\hat S^k_{\eta_k}\neq S^k|Y^{\tau_{\delta}}=y^t] \\
    \geq~&
    \mathbb P[\hat S^k_{\eta_k} \notin \mathcal G(Y^{\tau_{\delta}}), S^k \in \mathcal G(Y^{\tau_{\delta}})|Y^{\tau_{\delta}}=y^t] \\
    +~& \mathbb P[\hat S^k_{\eta_k} \in \mathcal G(Y^{\tau_{\delta}}), S^k \notin \mathcal G(Y^{\tau_{\delta}})|Y^{\tau_{\delta}}=y^t].
\end{aligned}
\end{equation}
We invoke Lemma~\ref{lemma_binary} with $p_{H_0}\leftarrow \mathbb P[H_0|Y^{\tau_{\delta}}=y^t]$, $p_{H_1}\leftarrow \mathbb P[H_1|Y^{\tau_{\delta}}=y^t]$, $\mathbb E[T] \leftarrow \mathbb E[\eta_k - \tau_{\delta}|Y^{\tau_{\delta}}= y^t]$ to further lower bound the left side of \eqref{two_error_prob} and obtain
\begin{align}\nonumber
    &\mathbb P[\hat S^k_{\eta_k}\neq S^k|Y^{\tau_{\delta}}=y^t]\geq\\\label{Eetatau}& \frac{\min\{\mathbb P[H_0|Y^{\tau_{\delta}}=y^t], \mathbb P[H_1|Y^{\tau_{\delta}}=y^t]\}}{4}e^{-C_1\mathbb E[\eta_k - \tau_{\delta}|Y^{\tau_{\delta}}= y^t]}.
\end{align}
To lower bound the minimization function on the right side of \eqref{Eetatau}, we show that alphabet $[q]^k$ can always be partitioned into two groups $\mathcal G(Y^{\tau_{\delta}})$ and $[q^k]\setminus \mathcal G(Y^{\tau_{\delta}})$ such that for all $t\geq 0$, $y^t\in\mathcal Y^t$, the priors of the hypotheses are lower bounded as
\begin{subequations}\label{group_posterior_binary}
\begin{align}
    \mathbb P[H_0|Y^{\tau_{\delta}}=y^t]&\geq \min\left\{p_{\min}\delta,1-\max_{s\in[q]^k}P_{S^k}(s)\right\},\\
    \mathbb P[H_1|Y^{\tau_{\delta}}=y^t]&\geq \min\left\{p_{\min}\delta,1-\max_{s\in[q]^k}P_{S^k}(s)\right\},
\end{align}
\end{subequations}
where $p_{\min}$ is defined in \eqref{def_lambda}. The priors of the hypotheses are both lower bounded by $p_{\min}\delta$ for any $\delta\in\left(0,\frac{1}{2}\right]$ if either event $A_1\triangleq \{\tau_{\delta}\geq 1\}$ or event $A_2\triangleq \{\tau_{\delta}=0, \max_{s\in[q]^k}P_{S^k}(s)\leq 0.5\}$ occurs, see \cite[Section V]{Berlin}. The threshold $0.5$ defining event $A_2$ corresponds to Berlin et al.'s reasoning in \cite[the second case in the third paragraph after Prop. 2]{Berlin}, which says at time $\tau_{\delta}$, if the posteriors\footnote{The source posterior at time $0$ is equal to the source prior $P_{S^k}$.}  of all source sequences in $[q]^k$ are upper bounded by $1-\delta\in[0.5,1]$, then $[q]^k$ can be divided into two groups with both hypotheses priors lower bounded by $p_{\min}\delta$. In Berlin et al.'s \cite{Berlin} channel coding context where the source symbols are equiprobably distributed, the union of the events $A_1\cup A_2$ occurs almost surely, since $P_{S^k}(s) = \frac{1}{q^k}$. Yet, in the JSCC context, it is possible that event $A_3\triangleq\{\tau_{\delta}=0,\max_{s\in[q]^k}P_{S^k}(s)>0.5\}$ occurs. When $A_3$ occurs, we move the sequence $s\in[q]^k$ that attains the maximum in event $A_3$ to $\mathcal G(Y^{\tau_{\delta}})$, and move the remaining sequences to the other group. This group partitioning rule implies \eqref{group_posterior_binary}. Plugging \eqref{group_posterior_binary} into \eqref{Eetatau}, taking an expectation of \eqref{Eetatau} over $Y^{\tau_{\delta}}$, and applying Jensen's inequality to $e^{-x}$ on the right side of \eqref{Eetatau}, we obtain
\begin{align}\label{Pe_lb}
    P_e \geq \frac{\min\{p_{\min}\delta,1-\max_{s\in[q]^k}P_{S^k}(s)\}}{4}e^{-C_1\mathbb E[\eta_k-\tau_{\delta}]}.
\end{align}
Rearranging terms in \eqref{Pe_lb}, we obtain \eqref{Berlin2_eq}.

\subsection{Proof of Lemma~\ref{lemma_asyp}}\label{pf_DS_prop}
We show that \eqref{lim_Psk} holds. We upper bound the entropy rate as
\begin{align}\nonumber
&\lim_{k\rightarrow\infty} \frac{H(S^k)}{k}\\\label{pf_DS3}
    & \leq \liminf_{k\rightarrow \infty} \frac{h\left(\max_{s\in[q]^k}P_{S^k}(s)\right)}{k} + \left(1-\max_{s\in[q]^k}P_{S^k}(s)\right)\log q\\\label{pf_DS4}
    & = \liminf_{k\rightarrow\infty} \left(1-\max_{s\in[q]^k}P_{S^k}(s)\right) \log q,
\end{align}
where \eqref{pf_DS3} holds since fixing the probability of the source sequence that attains $\max_{s\in[q]^k}P_{S^k}(s)$, the equiprobable distribution on the rest of $q^k-1$ sequences maximizes the concave entropy function; \eqref{pf_DS4} holds since the binary entropy function in \eqref{pf_DS3} is bounded between $[0,1]$. Finally, \eqref{lim_Psk} holds since the entropy rate is positive by assumption.

\section{Achievability proof of Theorem~\ref{thm_2}}\label{pf_achieve}
In the achievability proof, we fix a sequence of codes with instantaneous encoding for transmitting the first $k$ symbols of a DSS, $k=1,2,\dots$, over a non-degenerate DMC with feedback, evaluate the asymptotic behavior of the code sequence as $k\rightarrow\infty$, and conclude the achievability of $E(R)$ \eqref{converse2}.
In Appendix~\ref{pf_achieve_A}, we particularize the DSS in Theorem~\ref{thm_2} to the DS \eqref{DS}, and we show that both the MaxEJS code and the SED code \cite{Naghshvar2} achieve $E(R)$ \eqref{converse2}. In Appendix~\ref{pf_achieve_B}, we consider a DSS with $f=\infty$, and we show that $E(R)$ is achievable by a buffer-then-transmit code that idles the transmissions and only buffers the arriving symbols during the symbol arriving period and implements a JSCC reliability function-achieving code with block encoding after the symbol arriving period. In Appendix~\ref{pf_achieve_C}, we consider a DSS with $f<\infty$ that satisfies $(\mathrm{a})$--$(\mathrm{b})$, and we show that $E(R)$ \eqref{converse2} is achievable by a code with instantaneous encoding that implements the instantaneous encoding phase in Section~\ref{belief_phase} during the symbol arriving period and a JSCC reliability function-achieving code with block encoding after the symbol arriving period.

\subsection{A (fully accessible) DS}\label{pf_achieve_A}
We show that both the MaxEJS code for all non-degenerate DMCs \cite[Sec.~IV-C]{Naghshvar2} and the SED code for non-degenerate symmetric binary-input DMCs \cite[Sec.~V-B]{Naghshvar2} achieve $E(R)$ \eqref{converse2} for a DS. We denote a deterministic encoding function at time $t$ by
\begin{align}\label{def_gamma}
    \gamma_t\colon [q]^k\rightarrow \mathcal X,
\end{align}
we denote the vector of the message posteriors at time $t$ by 
\begin{align}
    \boldsymbol{\rho}(Y^{t}) \triangleq [P_{S^k|Y^t}(1|Y^t),P_{S^k|Y^t}(2|Y^t),\dots,P_{S^k|Y^t}(q^k|Y^t)],
\end{align}
and we denote the extrinsic Jensen-Shannon (EJS) divergence \cite{Naghshvar2} at time $t$ by
\begin{align}\nonumber
   &\mathrm{EJS}(\boldsymbol{\rho}(Y^{t-1}), \gamma_{t}) \triangleq \sum_{i=1}^{q^k}P_{S^k|Y^{t-1}}(i|Y^{t-1})\\ \label{EJS} &D\left(P_{Y|X = \gamma_t(i)}\middle|\middle| \sum_{j\neq i} \frac{P_{S^k|Y^{t-1}}(j|Y^{t-1})}{1-P_{S^k|Y^{t-1}}(i|Y^{t-1})}P_{Y|X=\gamma_t(j)}\right).
\end{align}
The MaxEJS code \cite[Section IV.C]{Naghshvar2} sets its encoding function $\gamma_t^*$ at time $t$ by solving the maximization problem:
\begin{align}
    \gamma_t^* \triangleq \arg\max_{\gamma_t\in\mathcal E} \mathrm{EJS}(\boldsymbol{\rho}(Y^{t-1}),\gamma_t),
\end{align}
where $\mathcal E$ is the set of all possible deterministic functions $\gamma_t$ \eqref{def_gamma}.
The SED code \cite{Naghshvar2} corresponds to the instantaneous SED code in Section~\ref{Sec_SED_ER} for a fully accessible source.

Lemma~\ref{lemma_s1}, stated next, will be used to examine whether a code with block encoding achieves the JSCC reliability function for a fully accessible source.
\begin{lemma}\label{lemma_s1}
Consider $k$ symbols $S^k\in[q]^k$ of a DS with prior probability $P_{S^k}$ and fix a non-degenerate DMC with capacity $C$ \eqref{capacity} and the maximum KL divergence $C_1$ \eqref{C1}. A code with block encoding achieves the JSCC reliability function \eqref{converse2} for the fully accessible source if and only if its stopping time $\eta_k$ and its error probability $\epsilon$ \eqref{error_constraint} at the stopping time $\eta_k$ satisfy
\begin{align}\label{lambda_k_ub0}
    \mathbb E[\eta_k] \leq \left(\frac{H(P_{S^k})}{C} + \frac{\log\frac{1}{\epsilon}}{C_1}\right)(1+o(1)),
\end{align}
where $o(1)\rightarrow 0$ as $k\rightarrow\infty$.
\end{lemma}
\begin{proof}
If a code with block encoding satisfies \eqref{lambda_k_ub0}, then it achieves $E(R)$ \eqref{converse2} because plugging \eqref{lambda_k_ub0} into \eqref{reliabilityfunc2} gives \eqref{converse2}. Conversely, if a code with block encoding achieves $E(R)$ \eqref{converse2}, then $\mathbb E[\eta_k]$ is upper bounded by the right side of \eqref{lambda_k_ub0}. This is because any achievability bound on $\mathbb E[\eta_k]$ that is asymptotically larger than the right side of \eqref{lambda_k_ub0} cannot achieve \eqref{converse2}.
\end{proof}

We show that the MaxEJS code and the SED code both satisfy \eqref{lambda_k_ub0}. While \cite[Eq. (32)]{Naghshvar2} in \cite[Theorem 1]{Naghshvar2} is obtained by plugging a uniform prior of the message to the entropy function in \cite[Appendix II, Eq. (71)]{Naghshvar2}, we leave the prior in its generic form and obtain a modified version of \cite[Theorem 1]{Naghshvar2} as follows.

\begin{lemma}[Modified Theorem 1 in \cite{Naghshvar2}]\label{lemma_EJS} Fix a non-degenerate DMC with capacity $C$ \eqref{capacity} and maximum KL divergence $C_1$ \eqref{C1}, and consider $k$ symbols $S^k\in[q]^k$ of a DS with source distribution $P_{S^k}$. If the encoding functions $\gamma_t$, $t=1,\dots,\eta_k$ of a code with block encoding with the MAP decoder \eqref{hatSNT} and the $\epsilon$-stopping rule \eqref{eta_k_sub} satisfy
\begin{align}\label{EJS1}
     &\mathrm{EJS}(\boldsymbol{\rho}(Y^{t-1}), \gamma_{t}) \geq C,\\ \nonumber
     &\mathrm{EJS}(\boldsymbol{\rho}(Y^{t-1}), \gamma_{t}) \geq \left(1 - \frac{1}{1+\max\{\log q^k, \log\frac{1}{ \epsilon}\}}\right) C_1,\\ \label{EJS2}&\text{if}~\max_{i\in[q]^k}\rho_i(Y^{t-1})\geq 1 - \frac{1}{1+\max\{\log q^k, \log\frac{1}{ \epsilon}\}},
\end{align}
then the expected decoding time of the code with block encoding is upper bounded as
\begin{align}\label{eta_k_ub}
    \mathbb E[\eta_k] \leq \frac{H(P_{S^k}) + \log\log\frac{q^k}{ \epsilon}}{C} + \frac{\log\frac{1}{ \epsilon} + 1}{C_1} + \frac{6(4C_2)^2}{CC_1},
\end{align}
where $C_2\triangleq \max_{y\in\mathcal Y}\frac{\max_{x\in\mathcal X}P_{Y|X}(y|x)}{\min_{x\in\mathcal X}P_{Y|X}(y|x)}$.
\end{lemma}
Since the MaxEJS code satisfies \eqref{EJS1}--\eqref{EJS2} for all non-degenerate DMCs by \cite[Proposition 2]{Naghshvar2}, and the SED code \cite[Sec. V-B]{Naghshvar2} satisfies \eqref{EJS1}--\eqref{EJS2} for non-degenerate symmetric binary-input DMCs by \cite[Proposition 4]{Naghshvar2}, we conclude from \eqref{eta_k_ub} and \eqref{infty_entropy} that they satisfy \eqref{lambda_k_ub0}.

\subsection{A DSS with $f=\infty$}\label{pf_achieve_B}
Fixing a $(q,\{t_n\}_{n=1}^{\infty})$ DSS with $f=\infty$, we show that $E(R)$ \eqref{converse2} is achievable by a buffer-then-transmit code that buffers the arriving symbols at times $t=1,\dots,t_k$ and operates as a JSCC reliability function \eqref{converse2}-achieving code with block encoding for $k$ symbols $S^k$ of a (fully accessible) DS with prior $P_{S^k}$ at times $t\geq t_k+1$ (e.g., the MaxEJS code \cite{Naghshvar2}). To this end, we show an achievability (upper) bound on the expected stopping time of the buffer-then-transmit code.

We denote by $\eta_k'$ the stopping time of the buffer-then-transmit code. We denote by $\eta_k$ the stopping time of a code with block encoding that achieves the JSCC reliability function \eqref{converse2} for a fully accessible source, and we denote by $\epsilon_k$ its error probability at $\eta_k$ \eqref{error_constraint}. Since the decoding starts after time $t_k$, we have
\begin{align}\label{eta_k0}
   \eta_k' =  t_k + \eta_k.
\end{align}

We invoke Lemma~\ref{lemma_s1} with $\epsilon\leftarrow \epsilon_k$ to upper bound $\mathbb E[\eta_k]$ on the right side of \eqref{eta_k0} and obtain an achievability bound on the expected decoding time $\mathbb E[\eta_k']$ of the buffer-then-transmit code:
\begin{align}\label{achieve_sub0_0}
   \mathbb E[\eta_k'] \leq \left(\frac{H(P_{S^k})}{C} + \frac{\log \frac{1}{\epsilon_k}}{C_1}\right)(1+o(1)) +  t_k.
\end{align}
For any DSS satisfying assumptions in Theorem~\ref{thm_2}, plugging \eqref{achieve_sub0_0} into \eqref{reliabilityfunc2}, we obtain \eqref{NC_reliability}. Since $f=\infty$, the achievability bound \eqref{NC_reliability} is equal to \eqref{converse2}. 

\subsection{A DSS with $f<\infty$}\label{pf_achieve_C}
Fixing a $(q,\{t_n\}_{n=1}^{\infty})$ DSS with $f<\infty$, we show that $E(R)$ \eqref{converse2} is achievable by a code with instantaneous encoding that implements the instantaneous encoding phase at times $t=1,2,\dots,t_k$ and operates as a JSCC reliability function \eqref{converse2}-achieving code with block encoding for $k$ symbols $S^k$ of a (fully accessible) DS with prior $P_{S^k|Y^{t_k}}$ at times $t\geq t_k+1$, where $Y_1,\dots, Y_{t_k}$ are the channel outputs generated in the instantaneous encoding phase. To this end, we will use Lemmas~\ref{lemma_s2}--\ref{lemma_s4}, stated below, together with Lemma~\ref{lemma_s1} in Appendix~\ref{pf_achieve_A} to obtain an achievability (upper) bound on the expected stopping time of the code. 

We fix an error probability $\epsilon_k$ \eqref{error_constraint}. We denote by $\eta_k$ the stopping time that ensures $\epsilon_k$ of a JSCC reliability function-achieving code with block encoding for $k$ symbols with prior $P_{S^k|Y^{t_k}}$. The stopping time $\eta_k'$ of the code with instantaneous encoding described above is
\begin{align}\label{eta_k'}
   \eta_k' = t_k + \eta_k
\end{align}
and its error probability is $\epsilon_k$.

The \emph{directed information} $I(A^n\rightarrow B^n)$ from a sequence $A^n$ to a sequence $B^n$ is defined as \cite{Massey}
\begin{equation}\label{DI}
    I(A^n\rightarrow B^n) = \sum_{i=1}^n I(A^i;B_i|B^{i-1}).
\end{equation}
The directed information captures the information due to the causal dependence of $B^n$ on $A^n$.

To upper bound the expected decoding time $\mathbb E[\eta_k']$, it suffices to upper bound $\mathbb E[\eta_k]$ \eqref{eta_k'}. Lemmas~\ref{lemma_s2}--\ref{lemma_s4}, stated below, show the behavior of the mutual information $I(S^k;Y^{t_k})$ as $k\rightarrow\infty$ generated by the instantaneous encoding phase in Section~\ref{belief_phase}.
\begin{lemma}\label{lemma_s2}
Fix a $(q,\{t_n\}_{n=1}^{\infty})$ DSS, and fix a non-degenerate DMC with capacity $C$ \eqref{capacity} and the maximum KL divergence $C_1$ \eqref{C1}. The instantaneous encoding phase that operates at times $t=1,2,\dots,t_k$ in Section~\ref{belief_phase} gives rise to
\begin{align}\label{lemma_I_eq}
    I(S^k; Y^{t_k}) = t_kC - I(X^{t_k}\rightarrow Y^{t_k}|S^k).
\end{align}
\end{lemma}
\begin{proof}
Appendix~\ref{pf_lemma_s2}.
\end{proof}
Lemma~\ref{lemma_assump_b}, stated next, displays the implications of assumption $(\mathrm{b})$ in Theorem~\ref{thm_2}. Given a DSS, we extract all the \emph{distinct} symbol arriving times from $t_1\leq t_2\leq \dots$, and we denote the sequence of distinct symbol arriving times by
\begin{align}\label{distinct_time}
    d_1<d_2<\dots
\end{align}
For example, if a DSS emits a source symbol every $\lambda\geq 1$ channel uses \eqref{periodic}, then the symbol arriving times are equal to the distinct symbols arriving times, i.e., $t_n = d_n$, and Lemma~\ref{lemma_assump_b} below trivially holds.
\begin{lemma}\label{lemma_assump_b} Fix a $(q,\{t_n\}_{n=1}^{\infty})$ DSS with $f<\infty$ and $f$ satisfying assumption $(\mathrm{b})$ in Theorem~\ref{thm_2}. Then,
\begin{itemize}
 \item[(i)] The time interval between consecutive symbol arriving times satisfies
 \begin{align}\label{assump_b_0}
     t_{n+1}- t_n = o(n), n=1,2,\dots;
 \end{align}
 \item[(ii)] The DSS has an infinite number of distinct symbol arriving times $d_{n'}$, $n'=1,2,\dots$.
\end{itemize}
\end{lemma}
\begin{proof}
(i) Assumption $(\mathrm{b})$ and $f<\infty$ ensure that $f\in\left(\frac{1}{\bunderline{H}}\left(H(P_Y^*)-\log\frac{1}{p_{\max}}\right),\infty\right)$. Thus, $\{\frac{t_n}{n}\}, n=1,2,\dots$ is a Cauchy sequence, and \eqref{assump_b_0} follows.

(ii) The DSS has an infinite number of distinct symbol arriving times since $0<f<\infty$ implies that there exist two positive functions $g_1,g_2$ with $g_1(n)=\Omega(n)$ and $g_2(n)=O(n)$ such that the symbol arriving time is bounded between $g_1(n)\leq t_n\leq g_2(n)$, and the symbol arriving interval is constrained by \eqref{assump_b_0}.
\end{proof}

Lemma~\ref{lemma_s4}, stated next, shows the asymptotic behavior of $I(X^{t_k}\rightarrow Y^{t_k}|S^k)$ in \eqref{lemma_I_eq}.
\begin{lemma}\label{lemma_s4}
Fix a $(q,\{t_n\}_{n=1}^{\infty})$ DSS that satisfies $(\mathrm{a})$--$(\mathrm{b})$ and $f<\infty$, and fix a non-degenerate DMC with capacity $C$ \eqref{capacity} and the maximum KL divergence $C_1$ \eqref{C1}. The instantaneous encoding phase that operates at times $t=1,2,\dots,t_k$ in Section~\ref{belief_phase} satisfies
\begin{align}\label{lemma_I_eq2}
    I(X^{t_k}\rightarrow Y^{t_k}|S^k) = o(t_k),
\end{align}
where $\lim_{k\rightarrow\infty}\frac{o(t_k)}{t_k} = 0$.
\end{lemma}
\begin{proof}
Appendix~\ref{pf_lemma_s4}.
\end{proof}

Using Lemmas \ref{lemma_s1}, \ref{lemma_s2}--\ref{lemma_s4}, we obtain an achievability bound on the expected decoding time $\mathbb E[\eta_k]$ \eqref{eta_k'}:
\begin{subequations}\label{achieve_sub}
\allowdisplaybreaks
\begin{align}\label{achieve_sub0}
   \mathbb E[\eta_k'] 
   &\leq \left(\frac{H(S^k|Y^{t_k})}{C} + \frac{\log \frac{1}{\epsilon_k}}{C_1}\right)(1+o(1)) +  t_k\\ \label{achieve_sub1}
   &= \left(\frac{H(S^k) - I(S^k; Y^{t_k})}{C} + \frac{\log \frac{1}{\epsilon_k}}{C_1}\right)(1+o(1)) +  t_k\\ \label{achieve_sub2}
   & = \left(\frac{H(S^k)}{C} + \frac{\log \frac{1}{\epsilon_k}}{C_1}\right)(1+o(1))
\end{align}
\end{subequations}
where \eqref{achieve_sub0} holds by upper bounding $\mathbb E[\eta_k]$ in \eqref{eta_k'} using \eqref{lambda_k_ub0} with $P_{S^k}\leftarrow P_{S^k|Y^{t_k}=y^{t_k}}$ and taking an expectation with respect to $Y^{t_k}$; \eqref{achieve_sub1} holds by expanding $H(S^k|Y^{t_k})$ in \eqref{achieve_sub0}; \eqref{achieve_sub2} holds by plugging Lemmas~\ref{lemma_s2}~and~\ref{lemma_s4} into $I(S^k; Y^{t_k})$ in \eqref{achieve_sub1} and using the fact that $\frac{o(t_k)}{H(S^k)}\leq \frac{o(t_k)}{t_k}\frac{1}{fH} = o(1)$, true due to the assumptions that the entropy rate $H$ and the symbol arriving rate $f$ are both positive. Plugging the achievability bound \eqref{achieve_sub} into \eqref{reliabilityfunc2}, we conclude that the code with instantaneous encoding achieves \eqref{converse2}. 

\subsection{Proof of Lemma~\ref{lemma_s2}}\label{pf_lemma_s2}
We first write the mutual information $I(S^k, X^t; Y_t|Y^{t-1})$ in two ways:
\begin{subequations}\label{I_0}
\begin{align}\label{I1}
  I(S^k, X^t; Y_t|Y^{t-1}) &=  I(S^k;Y_t|Y^{t-1}) + I(X^t; Y_t| Y^{t-1}, S^k)\\ \label{I2}
  &= I(X^t; Y_t|Y^{t-1}) + I(S^k; Y_t| Y^{t-1}, X^t),
\end{align}
\end{subequations}
where the second term on the right side of \eqref{I2} is equal to $0$ since $Y_t-(Y^{t-1},X^t) - S^k$ is a Markov chain. Thus,
\begin{align}\label{II1}
    I(S^k;Y_t|Y^{t-1}) = I(X^t; Y_t|Y^{t-1}) -I(X^t; Y_t| Y^{t-1}, S^k).
\end{align}

We expand $I(S^k; Y^{t_k})$ on the left side of \eqref{lemma_I_eq} as
\begin{subequations}
\begin{align}\label{III1}
    I(S^k; Y^{t_k}) & = \sum_{t=1}^{t_k}I(S^k; Y_t|Y^{t-1})\\ \label{III2}
                    & = \sum_{t=1}^{t_k} I(X^t; Y_t|Y^{t-1}) -I(X^t; Y_t| Y^{t-1}, S^k)\\ \label{III4}
                    & = t_kC - I(X^{t_k}\rightarrow Y^{t_k}|S^k),
\end{align}
\end{subequations}
where \eqref{III1} is by the chain rule; \eqref{III2} is by plugging \eqref{II1} into \eqref{III1}; \eqref{III4} is by applying the definition of the directed information \eqref{DI} to the second term of \eqref{III2} and plugging \eqref{PXYPX*} and the fact that $Y_i$, $i=1,\dots,t_k$ are i.i.d. according to $P_{Y}^*$ into the first term of \eqref{III2}. The channel outputs $Y_1,Y_2,\dots$ are independent since $Y_t - X_t - Y^{t-1}$ is a Markov chain and $X_t$ is independent of $Y^{t-1}$ \eqref{PXYPX*}. The channel outputs $Y_1,Y_2,\dots$ are identically distributed according to $P_Y^*$ since $X_1, X_2,\dots$ follow the capacity-achieving distribution $P_X^*$ \eqref{PXYPX*}.

\subsection{Proof of Lemma~\ref{lemma_s4}}\label{pf_lemma_s4}
To show \eqref{lemma_I_eq2}, we first upper bound the conditional directed information in \eqref{lemma_I_eq2} as a sum of conditional entropies, and upper bound each conditional entropy by a function of the source prior $\theta_{S^{N(t)}}(Y^{t-1})$. Then, we show that $\theta_{S^{N(t)}}(Y^{t-1})$ converges in probability to zero in time $t$ for $t\in[1,t_k]$ as $k\rightarrow\infty$. Finally, we show that the convergence of the source prior leads to the convergence of the entropy sequence and conclude \eqref{lemma_I_eq2}.

The conditional directed information in \eqref{lemma_I_eq2} can be upper bounded as
\begin{subequations}\label{DI_0}
\allowdisplaybreaks
\begin{align}\label{DI_1}
    I(X^{t_k}\rightarrow Y^{t_k}|S^k) &= \sum_{t=1}^{t_k}I(X^t; Y_t|Y^{t-1},S^k)\\ \label{DI_2}
    & \leq \sum_{t=1}^{t_k} H(X_t|Y^{t-1},S^k)\\ \label{DI_3}
    & = \sum_{t=1}^{t_k} H(X_t|Z_t,Y^{t-1}),
\end{align}
\end{subequations}
where \eqref{DI_1} is by the chain rule, and \eqref{DI_3} holds since $Z_t$ is a deterministic function of $(Y^{t-1},S^k)$ and $X_t-(Z_t,Y^{t-1}) - S^k$ is a Markov chain. 

We upper bound each term in the sum of \eqref{DI_3} using $\theta_{S^{N(t)}}(Y^{t-1})$. Given that $Z_{t} = z$, $Y^{t-1} = y^{t-1}$, if $z\in\underline{\mathcal{X}}(y^{t-1})$, we use \eqref{PXZY} to conclude
\begin{align}\label{HZYequal0}
    H(X_{t}|Z_{t} = z,Y^{t-1} = y^{t-1}) = 0.
\end{align} 
If $z\in\overline{\mathcal{X}}(y^{t-1})$, we rearrange terms in \eqref{SD} to obtain
\begin{subequations}\label{PIP2}
\begin{align}
    1-\frac{P_{X}^*(z)}{\pi_z(y^{t-1})} & \leq 1 - \frac{P_{X}^*(z)}{P_{X}^*(z) + \min_{i\in\mathcal G_z(y^{t-1})}\theta_i(y^{t-1})}\\
    &\leq  \frac{\min_{i\in\mathcal G_z(y^{t-1})}\theta_i(y^{t-1})}{\min_{x\in\mathcal X}P_{X}^*(x)}.
\end{align}
\end{subequations}
We upper bound $H(X_{t}|Z_{t} = z,Y^{t-1} = y^{t-1})$, $z\in\overline{\mathcal{X}}(y^{t-1})$ by
\begin{subequations}\label{HXZY}
\begin{align}\nonumber
    & H(X_{t}|Z_{t} = z,Y^{t-1} = y^{t-1}) \\\label{HXZY0}
    =~& \frac{P_{X}^*(z)}{\pi_z(y^{t-1})}\log\frac{\pi_z(y^{t-1})}{P_{X}^*(z)} + \sum_{x\in\underline {\mathcal X}(y^{t-1})}\frac{p_{z\rightarrow x}}{\pi_z(y^{t-1})}\log\frac{\pi_z(y^{t-1})}{p_{z\rightarrow x}}\\\label{HXZY1}
    \leq~& \frac{P_{X}^*(z)}{\pi_z(y^{t-1})}\log\frac{\pi_z(y^{t-1})}{P_{X}^*(z)} + \left(1-\frac{P_{X}^*(z)}{\pi_z(y^{t-1})}\right)\log\frac{|\mathcal X|-1}{1-\frac{P_{X}^*(z)}{\pi_z(y^{t-1})}}\\\label{HXZY2}
    =~&\left(1-\frac{P_{X}^*(z)}{\pi_z(y^{t-1})}\right)\log(|\mathcal X|-1) + h\left(1-\frac{P_{X}^*(z)}{\pi_z(y^{t-1})}\right)\\ \nonumber
    \leq~& \frac{\min_{i\in\mathcal G_z(y^{t-1})}\theta_i(y^{t-1})}{\min_{x\in\mathcal X}P_{X}^*(x)}\log(|\mathcal X|-1) \\\label{HXZY3}  +~& 2\sqrt{\frac{\min_{i\in\mathcal G_z(y^{t-1})}\theta_i(y^{t-1})}{\min_{x\in\mathcal X}P_{X}^*(x)}},
\end{align}
\end{subequations}
where \eqref{HXZY0} holds by \eqref{pxy} and \eqref{PXZY}; \eqref{HXZY1} holds since the sum in the second term on the right side of \eqref{HXZY0} is maximized if $p_{z\rightarrow x}$ is equiprobable on $\underline {\mathcal X}(y^{t-1})$, and $|\underline {\mathcal X}(y^{t-1})|\leq |\mathcal X|-1$; \eqref{HXZY2} holds by rearranging terms; \eqref{HXZY3} holds by applying the upper bound $h(p)\leq 2\sqrt{p}$ to the binary entropy function in \eqref{HXZY2} and plugging \eqref{PIP2} into \eqref{HXZY2}. Therefore, each term in \eqref{DI_3} is upper bounded as
\begin{subequations}\label{EHZY0}
\allowdisplaybreaks
\begin{align}\nonumber
   & H(X_{t}|Z_{t}, Y^{t-1}) \\\nonumber
   \leq~& \frac{\log(|\mathcal X|-1)}{\min_{x\in\mathcal X}P_{X}^*(x)}\mathbb E\left[\min_{i\in\mathcal G_{Z_{t}}(Y^{t-1})}\theta_i(Y^{t-1})\right] \\\label{EHZY}
   +~& \frac{2}{\sqrt{\min_{x\in\mathcal X}P_{X}^*(x)}}\mathbb E\left[\sqrt{\min_{i\in\mathcal G_{Z_{t}}(Y^{t_k-1})}\theta_i(Y^{t-1})}\right]\\ \nonumber
   \leq~& \frac{\log(|\mathcal X|-1)}{\min_{x\in\mathcal X}P_{X}^*(x)}\mathbb E\left[\theta_{S^{N(t)}}(Y^{t-1})\right]\\\label{EHZY2}
   +~& \frac{2}{\sqrt{\min_{x\in\mathcal X}P_{X}^*(x)}}\mathbb E\left[\sqrt{\theta_{S^{N(t)}}(Y^{t-1})}\right]\\\label{EHZY3}
    \leq~& \alpha\mathbb E\left[\sqrt{\theta_{S^{N(t)}}(Y^{t-1})}\right],
\end{align}
\end{subequations}
where 
\begin{align}
    \alpha\triangleq \max\left\{\frac{\log(|\mathcal X|-1)}{\min_{x\in\mathcal X}P_{X}^*(x)}, \frac{2}{\sqrt{\min_{x\in\mathcal X}P_{X}^*(x)}}\right\};
\end{align}
\eqref{EHZY} holds by \eqref{HZYequal0} and \eqref{HXZY}; \eqref{EHZY2} holds since $S^{N(t)}\in\mathcal G_{Z_{t}}(Y^{t_k-1})$.

To obtain the asymptotic behavior of $H(X_t|Z_t, Y^{t-1})$ in \eqref{EHZY0}, we proceed to analyze the asymptotic behavior of $\theta_{S^{N(t)}}(Y^{t-1})$. The source prior $\theta_{S^{N(t)}}(Y^{t-1})$ in \eqref{EHZY2} is upper bounded as
\begin{subequations}\label{theta_ub0}
\begin{align}\nonumber
    &\theta_{S^{N(t)}}(Y^{t-1})\\ \label{theta_ub} =~&  P_{S^{N(t)}}(S^{N(t)})\prod_{j=1}^{t-1}\frac{\sum_{x\in\mathcal X} P_{Y|X}(Y_j|x)P_{X_j|Z_j,Y^{j-1}}(x|Z_j, Y^{j-1})}{P_{Y}^*(Y_j)}\\\label{theta_ub2}
    \leq~& P_{S^{N(t)}}(S^{N(t)})\prod_{j=1}^{t-1}\frac{p_{\max}}{P_{Y}^*(Y_j)},
\end{align}
\end{subequations}
where \eqref{theta_ub} holds by \eqref{post_prior} and \eqref{prior_post}; \eqref{theta_ub2} holds since the numerator in the product term of \eqref{theta_ub} is upper bounded by $p_{\max}$ \eqref{P_max}. Given a DSS in Lemma~\ref{lemma_s4} with distinct symbol arriving times $d_{n'}$, $n'=1,2,\dots$ \eqref{distinct_time} ($n'$ is not bounded due to Lemma~\ref{lemma_assump_b} (ii)), we denote the gap between the symbol arriving rate $f$ and the threshold on the right side of \eqref{assump_b} by 
\begin{align}\label{gamma1}
    \gamma \triangleq f - \frac{1}{\bunderline{H}}\left(H(P_Y^*)-\log\frac{1}{p_{\max}}\right)\in(0,\infty).
\end{align}
For any $t\in[d_{n'},d_{n'+1})$, $n'=1,2,\dots$, the source prior $\theta_{S^{N(t)}}(Y^{t-1})$ \eqref{theta_ub0} satisfies
\begin{subequations}\label{PSNT}
\allowdisplaybreaks
\begin{align}\label{PSNT0}
&\mathbb P\left[\frac{1}{t}\log\theta_{S^{N(t)}}(Y^{t-1})\leq -\gamma \bunderline{H} \right]\\\nonumber
    \geq~&  \mathbb P\Bigg[ -\frac{1}{t}\left(\log \frac{1}{P_{S^{N(d_{n'})}}\left(S^{N(d_{n'})}\right)}\right) + \frac{t-1}{t}\log p_{\max}\\\label{PSNT1} &+\frac{1}{t}\sum_{j=1}^{t-1}\log\frac{1}{P_Y^*(Y_j)} \leq -\gamma \bunderline{H} \Bigg]\\\nonumber
    \geq ~& \mathbb P\Bigg[\frac{N(d_{n'})}{d_{n'+1}-1}\left(\frac{1}{N(d_{n'})}\log \frac{1}{P_{S^{N(d_{n'})}}\left(S^{N(d_{n'})}\right)}\right)\\ \nonumber &\geq \log p_{\max}+H(P_Y^*)+\gamma \bunderline{H},\\\label{PSNT2}
    &\frac{t-1}{t}\log p_{\max} +\frac{1}{t}\sum_{j=1}^{t-1}\log\frac{1}{P_Y^*(Y_j)} = \log p_{\max}+H(P_Y^*)\Bigg]\\\nonumber
    \geq ~&\mathbb P\Bigg[\frac{N(d_{n'})}{d_{n'+1}-1}\left(\frac{1}{N(d_{n'})}\log \frac{1}{P_{S^{N(d_{n'})}}\left(S^{N(d_{n'})}\right)}\right)\\\label{PSNT3} & \geq  \log p_{\max}+H(P_Y^*)+\gamma \bunderline{H}\Bigg]\\\label{PSNT4}
      +~&\mathbb P\left[\frac{t-1}{t}\log p_{\max}+ \frac{1}{t}\sum_{j=1}^{t-1}\log\frac{1}{P_Y^*(Y_j)} = \log p_{\max}+H(P_Y^*)\right]\\ \label{PSNT4'}
      -~&1\\\label{PSNT5}
\rightarrow~& 1,
\end{align}
\end{subequations}
as $n'\rightarrow\infty$, where \eqref{PSNT1} holds by plugging \eqref{theta_ub2} into \eqref{PSNT0} and by replacing $N(t)\leftarrow N(d_{n'})$ since $t\in[d_{n'},d_{n'+1})$; \eqref{PSNT2} holds since $t\leq d_{n'+1}-1$ and the event in \eqref{PSNT1} is implied by the events in \eqref{PSNT2}; \eqref{PSNT3}--\eqref{PSNT4'} hold by applying Fr\'{e}chet inequalities \cite{Frechet} to the probability in \eqref{PSNT2}; \eqref{PSNT5} holds since both probabilities in \eqref{PSNT3}--\eqref{PSNT4} converge to $1$ as $n'\rightarrow\infty$: the probability in \eqref{PSNT3} converges to $1$ as $n'\rightarrow\infty$ by Lemma~\ref{lemma_assump_b} (i), the fact that $\liminf_{n'\rightarrow\infty}\frac{N(d_{n'})}{d_{n'}}\geq f$ since $\left\{\frac{N(d_{n'})}{d_{n'}}\right\}_{n'=1}^{\infty}$ is a subsequence of $\left\{\frac{n}{t_n}\right\}_{n=1}^{\infty}$, the lower bound on the symbol arriving rate (assumption $(\mathrm{b})$), the lower bound on the information in $S^{N(d_{n'})}$ (assumption $(\mathrm{a})$), and the fact that $N(d_{n'})\rightarrow\infty$ as $n'\rightarrow\infty$ since $N(d_{n'})\geq n'$; the probability in \eqref{PSNT4} converges to $1$ since the sum over the logarithms of i.i.d. random variables $Y_1,Y_2,\dots$ (they are i.i.d. by the argument below \eqref{III4}) in \eqref{PSNT4} converges to $H(P_{Y}^*)$ by the law of large numbers, and $t\rightarrow\infty$ as $n'\rightarrow\infty$ due to $t\geq d_{n'}$. Rearranging terms in \eqref{PSNT0}, we we conclude that for any $\delta \in(0,1)$, there exists $n_{\delta}\in\mathbb Z_+$, such that for all $n'\geq n_{\delta}$, $t\in[d_{n'}, d_{n'+1})$, the probability in \eqref{expectation_theta} satisfies 
\begin{align}\label{PSNT0_m}
    \mathbb P[\theta_{S^{N(t)}}(Y^{t-1})\leq e^{-\gamma \bunderline{H} t}]> 1-\delta.
\end{align}

We analyze the asymptotic behavior of $H(X_t|Z_t, Y^{t-1})$ in \eqref{EHZY0} using \eqref{PSNT0_m}. 
Using the boundedness of the source prior $\theta_{S^{N(t)}}(Y^{t-1})\in[0,1]$, we upper bound the expectations in the right side of \eqref{EHZY3} as
\begin{align}\nonumber
    \mathbb E\left[\sqrt{\theta_{S^{N(t)}}(Y^{t-1})}\right]
    &\leq \mathbb P\left[\theta_{S^{N(t)}}(Y^{t-1})> e^{-\gamma \bunderline{H} t}\right] \\\label{expectation_theta}
    &+  e^{-\frac{\gamma \bunderline{H}}{2} t} \mathbb P\left[\theta_{S^{N(t)}}(Y^{t-1})\leq  e^{-\gamma \bunderline{H} t}\right]\\\label{expectation_ub}
    &<\delta + e^{-\frac{\gamma \bunderline{H}}{2} d_{n_{\delta}}} (1-\delta),
\end{align}
$\forall t\in[d_{n'},d_{n'+1})$, where \eqref{expectation_ub} holds due to \eqref{PSNT0_m} and the fact that the function $f(p)=p+\beta (1-p)$, $\beta<1$, is monotonically increasing on $p\in[0,1]$.

Plugging \eqref{expectation_ub} into \eqref{EHZY3}, we conclude that for all $n'\geq n_{\delta}$, it holds that $\forall t\in[d_{n'},d_{n'+1})$,
\begin{align}\label{H_asymptotic}
    H(X_t|Z_t, Y^{t-1})< \alpha\left(\delta + e^{-\frac{\gamma \bunderline{H}}{2} d_{n_{\delta}}} (1-\delta)\right).
\end{align}

We proceed to show \eqref{lemma_I_eq2} using \eqref{H_asymptotic}. Dividing both sides of \eqref{DI_0} by $t_k$ and taking $k\rightarrow\infty$, we upper bound the left side of \eqref{DI_0} as
\begin{subequations}\label{limsup_I0}
\begin{align}\nonumber
    &\limsup_{k\rightarrow\infty}\frac{1}{t_k}I(X^{t_k}\rightarrow Y^{t_k}|S^k)\\\label{limsup_I1}
    \leq ~&\limsup_{k\rightarrow\infty}\frac{1}{t_k}\sum_{t=1}^{t_k} H(X_t|Z_t, Y^{t-1})\\\nonumber
    < ~& \limsup_{k\rightarrow\infty}\frac{1}{t_k}\Bigg(|t_k-d_{n_{\delta}}|\alpha\left(\delta + e^{-\frac{\gamma \bunderline{H}}{2} d_{n_{\delta}}} (1-\delta)\right)\\\label{limsup_I2}
    &+d_{n_{\delta}}\log|\mathcal X|\Bigg)\\\label{limsup_I3}
    =~& \alpha\left(\delta + e^{-\frac{\gamma \bunderline{H}}{2} d_{n_{\delta}}} (1-\delta)\right),
\end{align}
\end{subequations}
where \eqref{limsup_I1} holds by \eqref{DI_3}; \eqref{limsup_I2} holds by upper bounding $H(X_t|Z_t, Y^{t-1})\leq \log|\mathcal X|$ for $t\leq d_{n_{\delta}}$ and upper bounding $H(X_t|Z_t, Y^{t-1})$ by \eqref{H_asymptotic} for $t>d_{n_{\delta}}$; \eqref{limsup_I3} holds since Lemma~\ref{lemma_assump_b} (i) implies that $d_{n_{\delta}}<\infty$ for some $n_{\delta}\in \mathbb Z_+$, and $f<\infty$ implies that $t_k\rightarrow\infty$ as $k\rightarrow\infty$.

Since $\delta$ can be made arbitrarily small while $d_{n_{\delta}}$ can be made arbitrarily large, we conclude \eqref{lemma_I_eq2}. 

\section{Decoding before the final arrival time}\label{decoding_after_tk}
For transmitting the first $k$ source symbols of a $(q,\{t_n\}_{n=1}^{\infty})$ DSS with $p_{S,\max}<1$, we show that if we decode before the final arrival time $t_k$, then the error probability $\mathbb P[S^k\neq \hat S^k_t]$, $t< t_k$ will not vanish with $k$ for any code with instantaneous encoding. 

For any $t< t_k$, $y^t\in\mathcal Y^t$, we lower bound the conditional error probability as
\begin{align}\label{P_lb}
   \mathbb P[S^k\neq \hat S^k_t|Y^t = y^t] \geq 1-\max_{i\in[q]^k}P_{S^k|Y^t}(i|y^t),
\end{align}
where the equality is attained by the MAP decoder.
Taking an expectation of both sides of \eqref{P_lb}, we obtain
\begin{subequations}\label{error_prob_expect}
\allowdisplaybreaks
\begin{align}\nonumber
    &\mathbb P[S^k\neq \hat S^k_t]\\\label{error_prob_expect1}  \geq~& 1-\mathbb E\left[\max_{i\in[q]^k}P_{S^k|Y^t}(i|Y^t)\right] \\\label{error_prob_expect2}
     =~&1-\mathbb E\left[\max_{i\in[q]^k}\sum_{j\in[q]^{N(t)}}P_{S^k|S^{N(t)}}(i|j)P_{S^{N(t)}|Y^t}(j|Y^t)\right]\\\label{error_prob_expect3}
    \geq~& 1 - \max_{i\in[q]^k,j\in[q]^{N(t)}}P_{S^k|S^{N(t)}}(i|j)\\\label{error_prob_expect4}
    \geq~& 1 - \prod_{n=N(t)+1}^{k}\max_{s\in[q],s'\in[q]^{n-1}}P_{S_n|S^{n-1}}(s|s')\\\label{error_prob_expect5}
    \geq~& 1 - (p_{S,\max})^{k-N(t)}\\\label{error_prob_expect6}
    >~&0,
\end{align}
\end{subequations}
where \eqref{error_prob_expect2} holds since $S^k-S^{N(t)}-Y^t$ is a Markov chain; \eqref{error_prob_expect3} holds by upper bounding $P_{S^k|S^{N(t)}}(i|j)$ in \eqref{error_prob_expect2} by its maximum; \eqref{error_prob_expect4} holds by writing $P_{S^k|S^{N(t)}}(\cdot|\cdot)$ as a product of probabilities $\{P_{S_n|S^{n-1}}(\cdot|\cdot)\}_{n=N(t)+1}^k$ and maximizing each term in the product; \eqref{error_prob_expect5} holds by upper bounding each term in the product by $p_{S,\max}$ \eqref{PSMAX}; \eqref{error_prob_expect6} holds by the assumption $p_{S,\max}<1$.

\section{Proof of Remark~\ref{rmk_drop_random}}\label{pf_drop_random}
We show that after the instantaneous encoding phase drops the randomization step \eqref{XbarXunder}--\eqref{PXZY} and only transmits $Z_t$ \eqref{Zt1} as the channel input,
it continues to satisfy the sufficient condition in \eqref{pre_coding_condition} under assumption $(\mathrm{b}^\prime)$. To this end, we first write $I(S^k; Y^{t_k})$ in \eqref{pre_coding_condition} as a sum of mutual informations. Then, we show that all source priors converge pointwise to zero in time during the symbol arriving period $[1,t_k]$ as $k\rightarrow\infty$; this implies that group priors converge pointwise to the capacity-achieving probabilities. Finally, we show that the convergence of the group priors implies that the summands of $I(S^k; Y^{t_k})$ converge to the capacity $C$ and conclude \eqref{pre_coding_condition}. Given channel outputs $y^{t-1}\in\mathcal Y^{t-1}$, we denote the source sequence in $[q]^{N(t)}$ that has the maximum source prior by
\begin{align}\label{pf_drop_random_eq3}
    i^* \triangleq \argmax_{i\in[q]^{N(t)}}\theta_i(y^{t-1}).
\end{align}

To expand $I(S^k; Y^{t_k})$ in \eqref{pre_coding_condition}, we first notice that \eqref{I_0}--\eqref{III2} continue to hold, thus $I(S^k; Y^{t_k})$ is equal to \eqref{III2}. The second term on the right side of \eqref{III2} is equal to zero since $X_t$ is a deterministic function of $(Y^{t-1},S^k)$, thus,
\begin{align}\label{pf_drop_random_eq2}
    I(S^k; Y^{t_k}) = \sum_{t=1}^{t_k}I(X_t; Y_t|Y^{t-1}).
\end{align}

We proceed to analyze the asymptotic behavior of $\theta_{i^*}(y^{t-1})$ \eqref{pf_drop_random_eq3}. Since the encoder drops the randomization step \eqref{XbarXunder}--\eqref{PXZY} and only transmits $Z_t$ \eqref{Zt1} as the channel input, the posterior update \eqref{prior_post} becomes \eqref{prior_post2}.
Upper bounding $P_{S^{N(t)}|S^{N(t-1)}}(\cdot|\cdot)$ in the prior update \eqref{post_prior} by the maximum symbol arriving probability $p_{S,\max}^{N(t)-N(t-1)}$ \eqref{PSMAX}, and upper bounding the numerator by $p_{\max}$ and the denominator by $p_{\min}$ in the fraction on the right side of \eqref{prior_post2}, we obtain an upper bound on the source prior $\theta_{i^*}(y^{t-1})$ as
\begin{align}\label{sp_ub}
    \theta_{i^*}(y^{t-1}) \leq p_{S,\max}^{N(t)}\left(\frac{p_{\max}}{p_{\min}}\right)^{t-1}
\end{align}
for all $i\in[q]^{N(t)}$, $y^{t-1}\in\mathcal Y^{t-1}$.
Given a DSS that satisfies assumption $(\mathrm{b}^\prime)$ with $f<\infty$ and distinct symbol arriving times $d_{n'}$, $n'=1,2,\dots$ \eqref{distinct_time} ($n'$ is not bounded due to Lemma~\ref{lemma_assump_b}), similar to \eqref{gamma1}, we denote the gap between the symbol arriving rate $f$ and the threshold in assumption $(\mathrm{b}^\prime)$ by
\begin{align}\label{def_gamma'}
    \gamma' \triangleq f - \frac{1}{\log\frac{1}{p_{S,\max}}}\left(\log\frac{1}{p_{\min}}-\log\frac{1}{p_{\max}}\right).
\end{align}
For any $t\in[d_{n'}, d_{n'+1})$, $n'=1,2,\dots$, the source prior $\theta_{i^*}(y^{t-1})$ for any $i^*\in[q]^{N(t)}$, $y^{t-1}\in\mathcal Y^{t-1}$ in \eqref{sp_ub} satisfies 
\begin{subequations}\label{pf_drop_random_eq5}
\allowdisplaybreaks
\begin{align}\nonumber
    &\limsup_{n'\rightarrow\infty}\frac{1}{t}\log \theta_{i^*}(y^{t-1})\\\label{pf_drop_random_eq5a}
    \leq~& -\left(\liminf_{n'\rightarrow\infty}\frac{N(t)}{t}\log \frac{1}{p_{S,\max}}\right) + \log\frac{p_{\max}}{p_{\min}}\\\label{pf_drop_random_eq5b}
     \leq~& -\left(\liminf_{n'\rightarrow\infty}\frac{N(d_{n'})}{d_{n'+1}-1}\log \frac{1}{p_{S,\max}}\right) + \log\frac{p_{\max}}{p_{\min}}\\ \label{pf_drop_random_eq5c}
     \leq~& - f\log\frac{1}{ p_{S,\max}} + \log\frac{p_{\max}}{p_{\min}}\\\label{pf_drop_random_eq5d}
     =~& -\gamma'\log \frac{1}{p_{S,\max}},
\end{align}
\end{subequations}
where \eqref{pf_drop_random_eq5a} is by taking the logarithm, dividing by $t$, and taking $n'$ to infinity on both sides of \eqref{sp_ub}; \eqref{pf_drop_random_eq5b} holds since $\frac{N(d_{n'})}{d_{n'+1}-1}\leq  \frac{N(t)}{t}$ for all $t\in[d_{n'},d_{n'+1})$; \eqref{pf_drop_random_eq5c} holds due to Lemma~\ref{lemma_assump_b} (i) and the fact that $\left\{\frac{N(d_{n'})}{d_{n'}}\right\}_{n'=1}^{\infty}$ is a subsequence of $\left\{\frac{n}{t_n}\right\}_{n=1}^{\infty}$; \eqref{pf_drop_random_eq5d} holds by plugging \eqref{def_gamma'} into \eqref{pf_drop_random_eq5c}. Rearranging terms of \eqref{pf_drop_random_eq5}, we conclude that the maximum source prior \eqref{pf_drop_random_eq3} converges pointwise: for any $y^{t-1}\in\mathcal Y^{t-1}$,
\begin{align}\label{as_convergence}
    \lim_{n'\rightarrow\infty}\theta_{i^*}(y^{t-1})  = 0,~\forall t\in[d_{n'},d_{n'+1}),
\end{align}
where $t\rightarrow\infty$ for any $t\in[d_{n'},d_{n'+1})$ as $n'\rightarrow\infty$.

The convergence of the source prior \eqref{as_convergence} implies the convergence of the group prior. The partitioning rule in \eqref{SD} ensures that the group prior $\pi_x(y^{t-1})$, $\forall x\in\mathcal X$ is simultaneously upper and lower bounded as
\begin{subequations}\label{pf_drop_random_eq4}
\begin{align}\label{pf_drop_random_eq4a}
     P_{X}^*(x) + \theta_{i^*}(y^{t-1}) &\geq \pi_x(y^{t-1}) \\\label{pf_drop_random_eq4b}
    & \geq P_{X}^*(x) - |\mathcal X|\theta_{i^*}(y^{t-1}),
\end{align}
\end{subequations}
where the upper bound \eqref{pf_drop_random_eq4a} holds by \eqref{SD} and \eqref{pf_drop_random_eq3};  the lower bound \eqref{pf_drop_random_eq4b} holds since all $|\mathcal X|$ group priors are upper bounded by \eqref{pf_drop_random_eq4a}. From \eqref{as_convergence} and \eqref{pf_drop_random_eq4}, we conclude that for all $x\in\mathcal X$, $y^{t-1}\in\mathcal Y^{t-1}$, 
\begin{align}\label{as_group_convergence}
    \lim_{n'\rightarrow\infty}\pi_x(y^{t-1}) = P_X^*(x), ~t\in[d_{n'},d_{n'+1}).
\end{align}

Next, we show the convergence of the group prior \eqref{as_group_convergence} implies the convergence of the mutual information $I(X_t; Y_t|Y^{t-1})$ in the sum of \eqref{pf_drop_random_eq2}. We expand the mutual information $I(X_t; Y_t|Y^{t-1})$ as
\begin{align}\nonumber
    &I(X_t; Y_t|Y^{t-1})
    = \sum_{y^{t-1}\in\mathcal Y^{t-1}}P_{Y^{t-1}}(y^{t-1})\\ \label{IXYYt-1_a} &\sum_{y\in\mathcal Y}\sum_{x\in\mathcal X}P_{Y|X}(y|x)\pi_{x}(y^{t-1})\log\frac{P_{Y|X}(y|x)}{\sum_{x'\in\mathcal X}P_{Y|X}(y|x')\pi_{x'}(y^{t-1})},
\end{align}
which achieves the channel capacity $C$ if $\pi_x(y^{t-1}) = P_X^*(x)$ for all $x\in\mathcal X$, $y^{t-1}\in\mathcal Y^{t-1}$. Using \eqref{as_group_convergence} and \eqref{IXYYt-1_a}, we conclude 
\begin{align}\label{In'C}
     \lim_{n'\rightarrow\infty}I(X_t; Y_t|Y^{t-1}) = C, ~t\in[d_{n'},d_{n'+1}).
\end{align}
Since $I(X_t; Y_t|Y^{t-1})\leq C$, one can write the equivalent of \eqref{In'C} as: for all $\epsilon>0$, there exists a $n_{\epsilon}\in\mathbb N$, such that for all $n'\geq n_{\epsilon}$, it holds that
\begin{align}\label{CI_convergence}
    I(X_t; Y_t|Y^{t-1}) >C- \epsilon, ~\forall t\in[d_{n'},d_{n'+1}).
\end{align}

We proceed to show \eqref{pre_coding_condition} using \eqref{pf_drop_random_eq2} and \eqref{CI_convergence}. Dividing both sides of \eqref{pf_drop_random_eq2} by $t_k$ and taking $k\rightarrow\infty$, we lower bound the left side of \eqref{pf_drop_random_eq2} as
\begin{subequations}
\begin{align}\label{pf_drop_random_eq6a}
  \lim_{k\rightarrow\infty}\frac{1}{t_k} I(S^k; Y^{t_k}) &=  \lim_{k\rightarrow\infty}\frac{1}{t_k}\sum_{t=1}^{t_k}I(X_t; Y_t|Y^{t-1})\\ \label{pf_drop_random_eq6b}
    & > \lim_{k\rightarrow\infty}\frac{1}{t_k}(t_k-d_{n_{\epsilon}})(C-\epsilon)\\\label{pf_drop_random_eq6c}
    & = C - \epsilon,
\end{align}
\end{subequations}
where \eqref{pf_drop_random_eq6b} holds by lower bounding $I(X_t; Y_t|Y^{t-1})$ by \eqref{CI_convergence} for $t> d_{n_{\epsilon}}$, and lower bounding  $I(X_t; Y_t|Y^{t-1})$ by zero for $t\leq d_{n_{\epsilon}}$.

Since $\epsilon$ in \eqref{pf_drop_random_eq6c} can be made arbitrarily small, and $ \lim_{k\rightarrow\infty}\frac{1}{t_k} I(S^k; Y^{t_k})\leq C$ by data processing, we conclude by the squeeze theorem that under assumption $(\mathrm{b}^\prime)$, the instantaneous encoding phase satisfies \eqref{pre_coding_condition} even if it does not randomize the channel input.

\section{The approximating instantaneous SED rule ensures \eqref{pi0pi1_2}}\label{approximate_SED}
We show that the approximating instantaneous SED rule in step~(iii$^\prime$) ensures \eqref{pi0pi1_2}. Since the left side of \eqref{pi0pi1_2} is equal to the minimum value on the right side of \eqref{nstar}, it suffices to show that the latter is upper bounded by $\theta_{\mathcal S_j}(y^{t-1})$.

We denote
\begin{subequations}\label{cn}
\begin{align}
    c_n &\triangleq (\pi_0(y^{t-1})- n\theta_{\mathcal S_j}(y^{t-1})) - (\pi_1(y^{t-1})+ n\theta_{\mathcal S_j}(y^{t-1}))\\
        &= 2\pi_0(y^{t-1}) - 1 - 2 n\theta_{\mathcal S_j}(y^{t-1}),
\end{align}
\end{subequations}
and we rewrite the minimization problem in \eqref{nstar} as
\begin{align}\label{cn_problem}
    \min_{n\in\{\bunderline{n}, \bar n\}}|c_n|.
\end{align}
By definitions of $\bunderline n$ \eqref{under_n} and $\bar{n}$ \eqref{upper_n}, it holds that $\bar n - \bunderline n = 1$. Thus
\begin{align}\label{diff_cn}
    c_{\bunderline n} - c_{\bar n} = 2\theta_{\mathcal S_j}(y^{t-1}).
\end{align}
Since $c_{\bunderline n}\geq 0$ and $c_{\bar n}\leq 0$, we conclude from \eqref{diff_cn} that
\begin{align}
    \min\{c_{\bunderline n},|c_{\bar n}|\}\leq \theta_{\mathcal S_j}(y^{t-1}),
\end{align}
which means that \eqref{cn_problem} is upper bounded by $\theta_{\mathcal S_j}(y^{t-1})$.

\section{Cardinality of common randomness}\label{Cara}
We adapt the proof in \cite[Theorem 19]{PPV} to our codes with instantaneous encoding to show that for any $\langle k,R,\epsilon\rangle$ code with instantaneous encoding that allows $|\mathcal U|=\infty$, there exists a $\langle k,R,\epsilon\rangle$ code with instantaneous encoding that allows $|\mathcal U|\leq 2$. Fixing a source length $k$, for $u=1,2,\dots,\infty$, we define $\mathcal G_u\subseteq \mathbb R^2$ as
\begin{align}\nonumber
    \mathcal G_u\triangleq \{(R,\epsilon): \exists~ \langle k,R,\epsilon\rangle ~&\text{code with instantaneous encoding}\\ &\text{that allows}~|\mathcal U|\leq u\}.
\end{align}

We show that $\mathcal G_1$ is a connected set. To see this, we arbitrarily select two elements in $\mathcal G_1$, denoted by $\Lambda_1\triangleq (R_1,\epsilon_1)$ and $\Lambda_2\triangleq (R_2,\epsilon_2)$. We denote $\Lambda_3 \triangleq (\min\{R_1,R_2\},\max\{\epsilon_1,\epsilon_2\})$. According to the rate and the error constraints in \eqref{time_constraint}--\eqref{error_constraint}, $\Lambda_i\in \mathcal G_1$, $i\in\{1,2\}$, indicates that all elements $(R,\epsilon)$ that simultaneously satisfy $R\leq R_i$ and $\epsilon\geq \epsilon_i$ belong to $\mathcal G_1$ (see the shaded region in Fig.~\ref{Fig_cara}). As a result, the line segments $L_i\triangleq \{\lambda \Lambda_i + (1-\lambda)\Lambda_3,\lambda\in[0,1]\}$, $i=1,2$, belong to $\mathcal G_1$, and the arc $L_1\cup L_2$ joins $\Lambda_1$ and $\Lambda_2$. 

\begin{figure}[h!]
\centering
\includegraphics[trim = 40mm 220mm 80mm 30mm, clip, width=10cm]{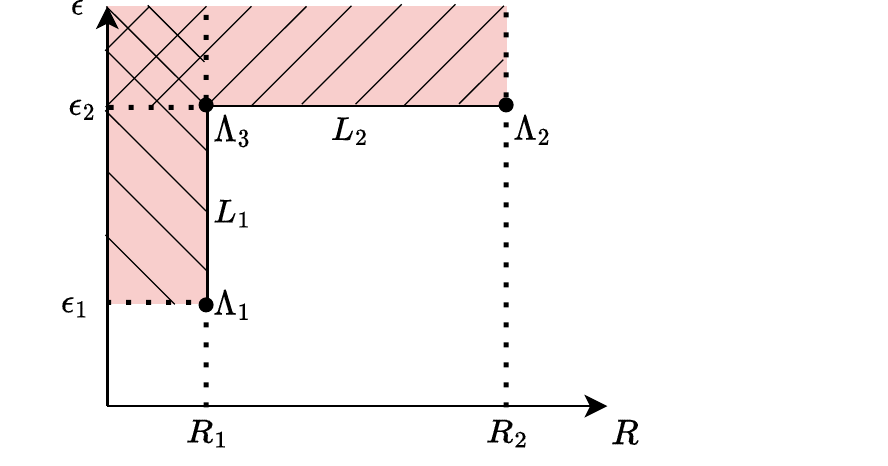}
\caption{Elements $\Lambda_1$ and $\Lambda_2$ are jointed by the arc $L_1\cup L_2$.} 
\label{Fig_cara}
\end{figure}

Since $\mathcal G_1\subseteq \mathbb R^2$, $\mathcal G_1$ is a connected set, and $\mathcal G_{\infty}$ is a convex hull of $\mathcal G_1$, by Fenchel-Eggleston-Carath\'{e}odory’s theorem  for connected sets \cite[Theorem 18(ii)]{Convexity}, any element in $\mathcal G_{\infty}$ can be represented as a convex combination of $2$ elements in $\mathcal G_1$, in other words, $\mathcal G_2 = \mathcal G_{\infty}$. 

\section{Zero-error code for degenerate DMCs}\label{zero-error}
In Appendix~\ref{zero-error_A}, we present our zero-error code with instantaneous encoding and common randomness for transmitting $k$ symbols of a DSS over a degenerate DMC. In Appendix~\ref{zero-error_B}, we present the proof that the code in Appendix~\ref{zero-error_A} achieves zero error for any rate asymptotically below $\frac{C}{H}$.

For a degenerate DMC \eqref{degenerate} in Theorem~\ref{thm_ze}, we denote by $P_{Y|X}\colon \mathcal X\rightarrow\mathcal Y$ its single-letter transition probability and denote by $P_{X}^*$ its capacity-achieving distribution. We relabel $x$ in \eqref{degenerate_a} by $\mathrm{ACK}$, and relabel $x'$ in \eqref{degenerate_b} by $\mathrm{NACK}$. We denote by $E_G(P_{Y|X}, R_c)$ Gallager's error exponent \cite{Gallager_e}, where $R_c$ is the channel coding rate in nats per channel use\footnote{For the consistency of notation, we use the same unit (i.e., nats per channel use) for $R_c$ as that in Gallager's paper \cite{Gallager_e}. The unit of all other rates in this paper is symbols per channel use.}. 
We denote by $R(\ell)$ the rate of the code used in the communication phase of the $\ell$-th block, and we denote by $\hat S^k(\ell)$ the estimate formed at the end of the communication phase of the $\ell$-th block.

\subsection{Zero-error code with instantaneous encoding and common randomness}\label{zero-error_A}
Similar to \cite{Burnashev}\cite{Yamamoto}\cite{Caire}\cite{Truong}, our code is divided into blocks. Each block contains a communication phase and a confirmation phase. The first block is different from the blocks after it, since it uses a Shannon limit-achieving code in the communication phase, whereas the blocks after the first block use random coding for all source sequences in alphabet $[q]^k$. We introduce the first block and the $\ell$-th block, $\ell\geq 2$, respectively. 

The first block is transmitted according to steps i)--ii) below. See Fig.~\ref{Fig_block_code} $(\mathrm{a})$ below for the diagram of the time division of transmitted blocks. See Fig.~\ref{Fig_block_code} $(\mathrm{b})$--(c) for the diagram of the first block.

i) Communication phase. The first $k$ symbols $S^k$ of the DSS in Theorem~\ref{thm_ze} is transmitted via a Shannon limit-achieving code with instantaneous encoding and common randomness at rate $R(1)<\frac{C}{H}$ symbols per channel use. (Such a code has been presented in the proof sketch of Theorem~\ref{thm_ze}. Namely, if $f=\infty$, we use a buffer-then-transmit code that implements the block encoding scheme in \cite[Theorem~2]{V}; if $f<\infty$, we precede the block encoding scheme in \cite[Theorem~2]{V} by an instantaneous encoding phase that satisfies \eqref{pre_coding_condition}.) At the end of the communication phase, the decoder yields an estimate $\hat S^k(1)$ of the source $S^k$ using the channel outputs that it has received in this phase.

ii) Confirmation phase. The encoder knows $\hat S^k(1)$ since it knows the channel outputs through the noiseless feedback. The encoder repeatedly transmits $\mathrm{ACK}$ if $S^k=\hat S^k(1)$, and transmits $\mathrm{NACK}$ if $S^k\neq \hat S^k(1)$, for $n_k$ channel uses. We pick $n_k$ as
\begin{align}\label{nk_1}
 n_k=\delta k,
\end{align}
where $\delta\in(0,1)$ can be made arbitrarily small. 
At the end of the confirmation phase, if the decoder receives a $y$, then it terminates the transmission and output $\hat S^{k}_{\eta_k} = \hat S^{k}(1)$; otherwise, the encoder transmits the next block.

\begin{figure}[h!]
\centering
\includegraphics[trim = 25mm 210mm 30mm 20mm, clip, width=10cm]{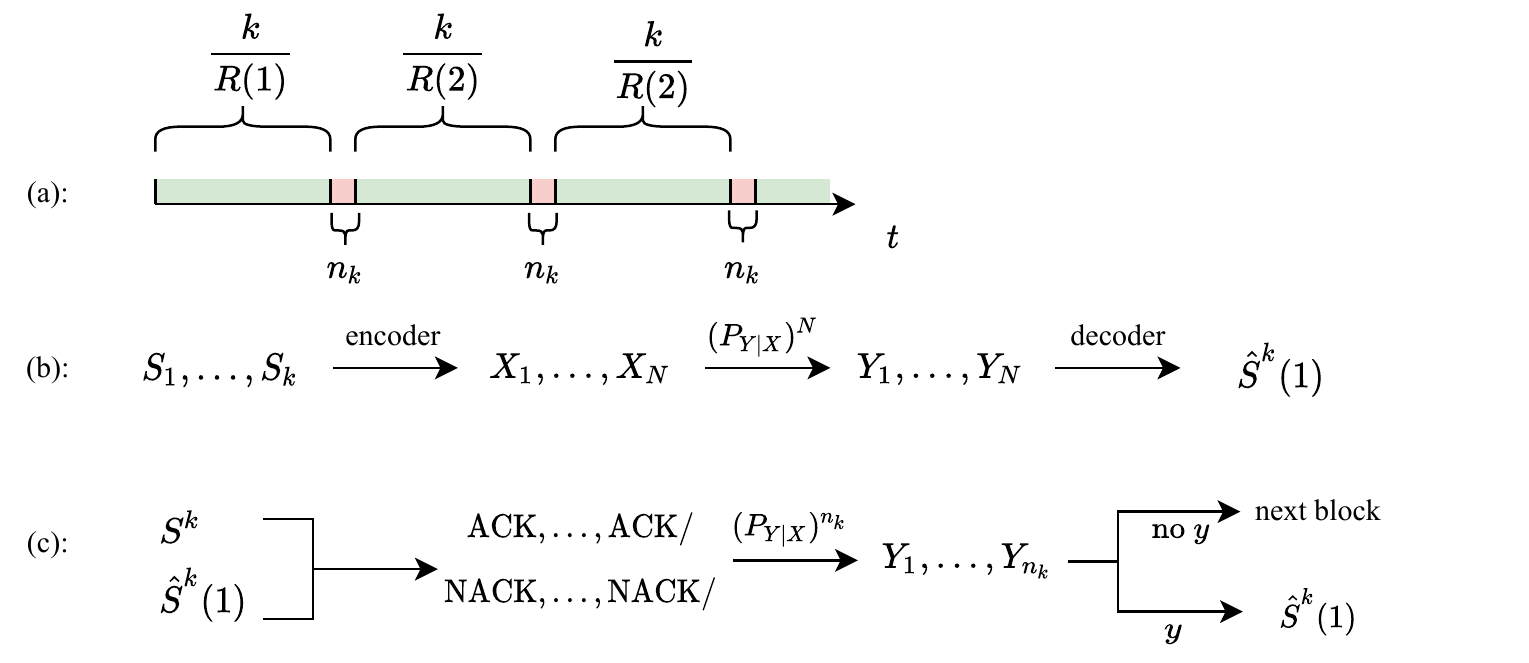}
\caption{(a) Time division of the transmitted blocks. The green regions represent the communication phases, and the red regions represent the confirmation phases. The \emph{expected} length of the first communication phase is $\frac{k}{R(1)}$. The length of the $\ell$-th communication phase, $\ell\geq 2$, is $\frac{k}{R(2)}$ since the random coding scheme has a fixed length. The length of the confirmation phase is $n_k$ \eqref{nk_1}. (b) Communication phase of the first block. The codeword length $N$ can be random with expectation $\mathbb E[N]=\frac{k}{R(1)}$. (c) Confirmation phase of the first block.} 
\label{Fig_block_code}
\end{figure}

The $\ell$-th block, $\ell\geq 2$, is transmitted according to steps iii)--iv) below.

iii) Communication phase. For every sequence in the alphabet $[q]^k$ of $S^k$, the encoder generates a codeword via random coding according to the capacity-achieving distribution $P_X^*$ at rate $R(2)<\frac{C}{\log q}$ symbols per channel use. At the end of the communication phase, the maximum likelihood (ML) decoder yields an estimate $\hat S^k(\ell)$ of the source symbols $S^k$ using the channel outputs that it has received in this phase.

iv) Confirmation phase. The encoder, the decoder, and the stopping rule are the same as those in the first block with $\hat S^k(1)\leftarrow \hat S^k(\ell)$.

The random codebook is refreshed in every retransmitted block and is known by the decoder. This gives rise to the following observations:\\
1) The codewords transmitted in the communication phases of the $\ell=1,2,\dots$ blocks are independent from each other;\\
2) As a result of 1), the channel outputs of the $\ell=1,2,\dots$ blocks are independent from each other;\\
3) The codewords transmitted in the communication phase of the $\ell = 2,3,\dots$ blocks are i.i.d. random vectors. (The codeword in the first block is excluded since the first block need not use random coding in the communication phase);\\
4) As a result of 3), the channel outputs of the $\ell=2, 3, \dots$ blocks are i.i.d. random vectors.\\
We will use observations 2) and 4) in the proof below.

\subsection{Proof of Theorem~\ref{thm_ze}}\label{zero-error_B}
Fix any $R<\frac{C}{H}$. We show that by adjusting $R(1)$, the rate of the Shannon limit-achieving code in the communication phase of the first block, to $R$, the code in Appendix~\ref{zero-error_A} achieves zero error with rate converging to $R$ \eqref{RkR}.

We denote by $\eta_k$ and $T_k$ the stopping time and the number of blocks transmitted after the first block until the stopping time, respectively. We denote by $A_{\ell}$ the event that no $y$ is received in the confirmation phase of the $\ell$-th block. 

Since the decoder will never receive $y$ if $\mathrm{ACK}$ is transmitted in the confirmation phase, the error probability of the code in Appendix~\ref{zero-error_A} is zero, i.e.,
\begin{align}\label{error_zero}
    \mathbb P[S^k\neq \hat S^k(1+T_k)] = 0,
\end{align}
where $1+T_k$ represents the total number of blocks transmitted until the stopping time, and $T_k$ is almost surely finite as a result of Lemmas~\ref{lemma_ETk_ub}~and~\ref{lemma_channel_error} below. This confirms that the code in Section~\ref{zero-error_A} achieves zero error \eqref{error_constraint}. 

To analyze the behavior of the rate $R_k = \frac{k}{\mathbb E[\eta_k]}$, we first observe that since the expected length of the first block is $\frac{k}{R(1)}+\delta k$ and the (fixed) length of the $\ell$-th block, $\ell\geq 2$, is $\frac{k}{R(2)}+\delta k$, the expected decoding time $\mathbb E[\eta_k]$ is equal to 
\begin{align}\label{etaNa}
 \mathbb E[\eta_k] &= \frac{k}{R(1)}+\delta k + \mathbb E[T_k]\left(\frac{k}{R(2)}+\delta k\right).
\end{align}
We bound the expected number of blocks $T_k$ transmitted after the first block using Lemmas~\ref{lemma_ETk_ub}~and~\ref{lemma_channel_error}, stated next.

\begin{lemma}\label{lemma_ETk_ub}
The number of blocks $T_k$ transmitted after the first block satisfies
\begin{align}\label{lemma_ETk_ub_eq}
    \mathbb E[T_k] \leq \frac{\mathbb P[S^k\neq \hat S^k (1)]+(1-P_{Y|X}(y|\mathrm{ACK}))^{\delta k}}{1-\mathbb P[S^k\neq \hat S^k (\ell)]-(1-P_{Y|X}(y|\mathrm{ACK}))^{\delta k}}.
\end{align}
\end{lemma}
\begin{proof}
Appendix~\ref{pf_lemma_ETk_ub}.
\end{proof}

\begin{lemma}\label{lemma_channel_error}
Given a DSS with entropy rate $H>0$ satisfying assumptions (a)--(b) in Theorem~\ref{thm_2}, the probability of erroneously decoding $S^k$ at the end of the communication phase of the $\ell$-th block is upper bounded as
\allowdisplaybreaks
\begin{subequations}\label{channel_error0}
\begin{align}\label{channel_error}
    &\mathbb P[S^k\neq \hat S^k(1)] \leq e^{-\frac{k}{R(1)}\left(\frac{C}{1+o(1)} - \frac{H\left(S^k\right)}{k}R(1)\right)},\\\label{channel_error_l}
   &\mathbb P[S^k\neq \hat S^k(\ell)] \leq e^{-\frac{k}{R(2)} E_G\left(P_{Y|X}, R(2)\log q\right)}, ~\ell=2,3,\dots
\end{align}
\end{subequations}
\end{lemma}
\begin{proof}
Since the block encoding scheme \cite[Theorem~2]{V} satisfies Lemma~\ref{lemma_s1} with $C_1\leftarrow C$, one can follow Appendices~\ref{pf_achieve_B}--\ref{pf_achieve_C} with $C_1\leftarrow C$ to upper bound the expected decoding time of the Shannon limit-achieving code in \cite[Theorem~2]{V} and thereby obtain \eqref{channel_error}.
The error probability \eqref{channel_error_l} holds since the random encoder together with the ML decoder attains Gallager's error exponent \cite{Gallager_e} for channel coding rate (nats per channel use) below $C$. This holds regardless of the distribution of the message because Gallager's error exponent holds under the maximum error probability criterion. 
\end{proof}
Plugging \eqref{lemma_ETk_ub_eq} and \eqref{channel_error0} into the right side of \eqref{etaNa}, we obtain the asymptotic behavior of the rate as
\begin{subequations}
\begin{align}
    \lim_{k\rightarrow \infty}R_k &=  \lim_{k\rightarrow \infty}\frac{k}{\mathbb E[\eta_k]}\\
    &\geq R(1)\frac{1}{1+R(1)\delta}.
\end{align}
\end{subequations}
Letting $R(1)$ be arbitrarily close to $\frac{C}{H}$ and taking $\delta$ to an arbitrarily small number, we conclude \eqref{RkR}.

\subsection{Proof of Lemma~\ref{lemma_ETk_ub}}\label{pf_lemma_ETk_ub}
We establish the pmf of $T_k$ using the probabilities $\mathbb P[A_{\ell}]$, $\ell=1,2,\dots$ The complementary cdf of $T_k$ is given by
\begin{align}\label{Tk1_b}
    \mathbb P[T_k > 0]~& = \mathbb P[A_1],
\end{align}
where \eqref{Tk1_b} holds by the definition of $A_1$ and the stopping rule of the code. We proceed to show the pmf at $T_k= t$, $t\geq 1$ conditioned on $T_k>0$: 
\begin{subequations}\label{Tkl}
\begin{align}\nonumber
    &\mathbb P[T_k = t|T_k>0] \\\label{Tkl_a}
    =~& \mathbb P[A_2\cap\dots\cap A_{t}\cap A^c_{t+1}|A_1]\\\label{Tkl_b}
    =~& \left(\prod_{i=2}^{t}\mathbb P[A_i|A_1,\dots, A_{i-1}]\right)\mathbb P[A^c_{t+1}|A_1,\dots,A_t]\\\label{Tkl_c}
    =~& (\mathbb P[A_2])^{t-1}(1-\mathbb P[A_2]),
\end{align}
\end{subequations}
where \eqref{Tkl_a} is by the stopping rule of the code; \eqref{Tkl_b} is by expanding \eqref{Tkl_a}; \eqref{Tkl_c} is by observations 2) and 4) in Appendix~\ref{zero-error_A}: observation 2) implies that event $A_i$ and its complementary event $A_i^c$ are both independent of $A_1,\dots,A_{i-1}$, $i\geq 2$, observation 4) implies that $\mathbb P[A_i] = \mathbb P[A_2]$, $i\geq 2$. Since the conditional pmf $\mathbb P[T_k = t|T_k>0]$ in \eqref{Tkl} follows a geometric distribution with success probability $1-\mathbb P[A_2]$, its mean is given by
\begin{align}\label{ETk_compute}
    \mathbb E[T_k|T_k>0] = \frac{1}{1-\mathbb P[A_2]}.  
\end{align}
Using \eqref{Tk1_b} and \eqref{ETk_compute}, we obtain $\mathbb E[T_k]$ as
\begin{align}\label{ETk_exact}
    \mathbb E[T_k] = \frac{\mathbb P[A_1]}{1-\mathbb P[A_2]}.
\end{align}
It remains to compute the probability of event $A_{\ell}$ in \eqref{ETk_exact} to conclude \eqref{lemma_ETk_ub_eq}. In the confirmation phase of the $\ell$-th block, $\ell=1,2,\dots$, conditioned on $S^k = \hat S^k(\ell)$, the probability of event $A_{\ell}$ is given by\footnote{For practical implementations, one can choose $\mathrm{ACK}$ as the channel input that achieves the maximum transition probability $\max_{x\in\mathcal X}P_{Y|X}(y|x)$ to increase the probability of receiving a $y$.}
\begin{subequations}\label{confirm_error}
\begin{align}
    \mathbb P[A_{\ell}|S^k = \hat S^k(\ell)] &=  (1-P_{Y|X}(y|\mathrm{ACK}))^{\delta k}.
\end{align}
\end{subequations}
The probability of event $A_\ell$ is upper bounded as
\begin{subequations}\label{PA1}
\begin{align}\nonumber
    \mathbb P[A_{\ell}] &= \mathbb P[A_{\ell}|S^k \neq \hat S^k(\ell)]\mathbb P[S^k \neq \hat S^k(\ell)] \\\label{PA1a}
    &+ \mathbb P[A_{\ell}|S^k = \hat S^k(\ell)]\mathbb P[S^k = \hat S^k(\ell)]\\\label{PA1b}
    &\leq \mathbb P[S^k \neq \hat S^k(\ell)] + \mathbb P[A_{\ell}|S^k = \hat S^k(\ell)],
\end{align}
\end{subequations}
where \eqref{PA1b} holds by upper bounding the first and the last probabilities on the right side of \eqref{PA1a} by $1$. Plugging the upper bound in \eqref{PA1b} into the right side of \eqref{ETk_exact}, we obtain \eqref{lemma_ETk_ub_eq}.

\ifCLASSOPTIONcaptionsoff
  \newpage
\fi







\end{document}